\documentclass[11pt]{article}

\usepackage[a4paper, margin=2.5cm]{geometry}

\usepackage[T1]{fontenc}
\usepackage[utf8]{inputenc}

\usepackage{mathrsfs}
\usepackage{mathtools}%
\mathtoolsset{showonlyrefs}
\usepackage{amsthm}
\usepackage{amsxtra}%
\usepackage{amsfonts}%
\usepackage{amssymb}%

\usepackage[
    backend=biber, 
    style=alphabetic,
    maxbibnames=99, 
    minbibnames=99
]{biblatex}
\AtEveryBibitem{
    \clearfield{urlyear}
    \clearfield{urlmonth}
    \clearfield{urlday}
    \clearfield{abstract}
    \clearfield{issn}
    \clearfield{isbn}
    \clearfield{editor}
    \clearfield{address}
    \clearfield{url}
}
\usepackage{color}
\usepackage{graphicx}
\usepackage{subfig}
\usepackage{array}
\usepackage{booktabs}
\usepackage{bbm}
\usepackage{enumitem}
\usepackage{dsfont}
\usepackage{authblk}
\usepackage[obeyFinal]{todonotes}
\usepackage{mleftright}
\mleftright

\usepackage{tikz}
\usetikzlibrary {graphs,graphs.standard}

\usepackage{zref-clever}
\zcsetup{nameinlink, cap, noabbrev}

\usepackage{hyperref}
\hypersetup{colorlinks,breaklinks,
	linkcolor=blue,urlcolor=blue,
	anchorcolor=blue,citecolor=blue}

\usepackage{orcidlink}
\newtheorem{theorem}{Theorem}
\newtheorem{lemma}[theorem]{Lemma}
\AddToHook{env/lemma/begin}{%
   \zcsetup{countertype={theorem=lemma}}}
\zcRefTypeSetup{lemma}{Name-sg=Lemma}

\newtheorem{corollary}[theorem]{Corollary}
\AddToHook{env/corollary/begin}{%
   \zcsetup{countertype={theorem=corollary}}}
\zcRefTypeSetup{corollary}{Name-sg=Corollary}
\newtheorem{proposition}[theorem]{Proposition}
\AddToHook{env/proposition/begin}{%
   \zcsetup{countertype={theorem=proposition}}}
\zcRefTypeSetup{proposition}{Name-sg=Proposition}

\theoremstyle{definition}
\newtheorem{remark}[theorem]{Remark}
\AddToHook{env/remark/begin}{%
   \zcsetup{countertype={theorem=remark}}}
\zcRefTypeSetup{remark}{Name-sg=Remark}

\newtheorem{definition}[theorem]{Definition}
\AddToHook{env/definition/begin}{\zcsetup{countertype={theorem=definition}}}
\zcRefTypeSetup{definition}{Name-sg=Definition,Name-pl=Definitions}

\AddToHook{env/assumption/begin}{\zcsetup{countertype={theorem=assumption}}}
\zcRefTypeSetup{assumption}{Name-sg=Assumption,Name-pl=Assumptions}

\AddToHook{env/example/begin}{%
   \zcsetup{countertype={theorem=example}}}
\zcRefTypeSetup{example}{Name-sg=Example,Name-pl=Examples}

\numberwithin{equation}{section}
\numberwithin{theorem}{section}

\DeclareMathOperator{\supp}{supp}
\DeclareMathOperator{\sign}{sign}
\DeclareMathOperator{\dist}{dist}
\DeclareMathOperator{\Id}{Id}
\DeclareMathOperator{\Var}{Var}
\DeclareMathOperator{\Cov}{Cov}

\DeclareMathOperator{\exterior}{ext}
\DeclareMathOperator{\Unif}{Unif}

\DeclareMathOperator{\TV}{TV}

\newcommand{\sprod}[2]{\langle #1\rangle_{#2}}

\newcommand{\set}[1]{\left\{#1\right\}}
\newcommand{\norm}[1]{\left\|#1\right\|}
\newcommand{\abs}[1]{\left\lvert #1 \right\rvert}
\newcommand{\expa}[1]{\exp\left( #1 \right)}

\renewcommand{\bar}[1]{{\overline{#1}}}
\newcommand{\eps}{\varepsilon}
\renewcommand{\d}{\,\mathrm{d}}

\renewcommand{\P}{{\mathbbm P}}
\newcommand{\E}{{\mathcal E}}

\newcommand{\R}{\mathbbm{R}}
\newcommand{\N}{\mathbbm{N}}

\newcommand{\Z}{\mathbbm{Z}}
\newcommand{\I}{\mathbbm{1}}

\let\emptyset\varnothing

\newcommand{\phntm}{\phantom{{}={}}}

\title{Entropic repulsion to the middle layer}

\author{Max Mihailescu\,\orcidlink{0009-0002-7382-2390}\thanks{Hausdorff Center for Mathematics and Institute for Applied Mathematics, University of Bonn. \texttt{mmihaile@uni-bonn.de}} \ \addtocounter{footnote}{1}and Ron Peled\,\orcidlink{0000-0002-6449-4666}\thanks{Department of Mathematics, University of Maryland, College Park. \texttt{peledron@umd.edu}}}

\begin{document}

\maketitle
\begin{abstract}
    We consider $\nabla\varphi$ height functions with even and convex interaction energy $W$ on the lattice $\Z^d$, which are restricted to take values in the set $\{-S, \ldots, S\}$ for some integer $S\ge1$. We study the effect of entropic repulsion, which tends to push the spin values to the middle layer.
    
    We prove that the model has a \emph{unique} Gibbs measure, with exponential decay of correlations, in the cases:
    \begin{itemize}
        \item Dimension $d=2$ at \emph{all temperatures}.
        \item Dimensions $d\ge3$ at \emph{all temperatures}, for a wide class of $W$ with non-increasing second derivative, including the family $W(x)=|x|^p$ for $p\in[1,2]$.
        \item Dimensions $d\ge 3$ at both low and high temperatures, $T\in (0,\frac{W(1)d}{4(\ln d+\ln 8)})\cup(d W(2S),\infty]$, with the normalization $W(0)=0$. At low temperatures, our proof provides an alternative to Pirogov--Sinai methods.
    \end{itemize}
    Conversely, we exhibit a class of even and convex interaction energies $W$ which, in high dimensions and suitable temperature regimes, have \emph{multiple} Gibbs measures.
    
    Though uniqueness may fail, we show that the magnetization of every Gibbs measure lies in $(-\frac{1}{2},\frac{1}{2})$. This implies the delocalization of the model restricted to take values in $\{0,1,\ldots\}$ (i.e., conditioned to lie above a floor) for all dimensions, any even and convex $W$, and all temperatures.

    Our methods extend to additional setups: We prove that height functions taking values in the real interval $[-1,1]$ \emph{always} have a unique Gibbs measure, a result previously proved only for the quadratic interaction. For height functions taking values in $\{-S+\frac{1}{2}, \ldots, S-\frac{1}{2}\}$, $S\ge1$ integer, we prove that the magnetization of every Gibbs measure lies in $(-\frac12, \frac12)$.

    The special case $W(x)=x^2$ of our results addresses questions left open in the work of Bricmont--El Mellouki--Fr\"ohlich (1986).
\end{abstract}

\section{Introduction}

Statistical physics provides many examples of interfaces, such as the domain walls between different phases in spin systems, or in the study of elastic membranes. Generally, their properties are determined by a competition between surface tension, which favors flat interfaces, and entropy, which favors roughness. Such interfaces are often modeled by
$\nabla\varphi$ height functions: functions $\varphi$ on $\Z^d$, which may be real- or integer-valued, with formal Hamiltonian
\begin{equation}\label{eq:formal hamiltonian without rho}
    \mathscr H(\varphi) :=  \beta\sum_{\set{i,j} \in \mathcal E(\Z^d)} W\left(\varphi_i - \varphi_j\right),
\end{equation}
where the \emph{interaction energy} $W$ is an even function (i.e., $W(x) = W(-x)$), where $\mathcal E(\Z^d)$ is the set of edges of the $\Z^d$ lattice and where, for convenience in our later arguments, we include the inverse temperature parameter $\beta$ directly in the Hamiltonian. One further restricts to convex $W$, to obtain the stochastic monotonicity of the distribution in its boundary values (FKG).

A natural family of examples is given by $W(x) = |x|^p$ for $p\ge 1$. Of special note are $W(x)=|x|$, whose integer-valued version is named the Solid-On-Solid model and is commonly used to approximate interfaces in the Ising model, and $W(x) = x^2$, the (lattice) Gaussian free field. The real-valued Gaussian free field is exactly solvable, and the extent to which its properties extend to other choices of interaction energies, or to the integer-valued case, have been the subject of many investigations; see Velenik~\cite{Velenik:LocalizationDelocalizationRandom2006a}, Dembo--Funaki~\cite{Dembo:StochasticInterfaceModels2005} and Sheffield~\cite{Sheffield:RandomSurfaces2005} for classical monographs on the subject.

In this work we consider the effect of confining $\nabla\varphi$ height functions between two hard walls (a floor and a ceiling), by requiring that $|\varphi_i|\le S$ for all $i$. Under these constraints, height functions centered around height $0$ enjoy more room to fluctuate and are thus entropically favored. How strong is this \emph{entropic repulsion} effect? A seminal result on this topic was obtained by McBryan--Spencer~\cite{McBryan:DecayCorrelationsSymmetric1977} who considered the effect for the \emph{real-valued} Gaussian free field. Their results, extended by Bricmont--El Mellouki--Fröhlich~\cite[Proposition A1]{Bricmont:RandomSurfacesStatistical1986a} based on an argument of Sokal, show that the hard wall confinement localizes the model in a strong sense: It has a unique Gibbs measure, and this measure is centered at height zero and its correlations decay exponentially. An alternative proof was recently obtained by D’Alimonte--Lammers~\cite[Section 8]{DAlimonte:FreeEnergyAnalyticity2026}. Further, in a series of papers, Sakagawa~\cite{Sakagawa:EntropicRepulsionHigh2006,Sakagawa:BoundsMassHigh2007,Sakagawa:ConfinementTwoDimensional2009}
provides precise information on several quantities, including the rate of exponential decay (the `mass') in this model, and a study of the behavior of a \emph{pair of non-crossing interfaces} confined between a floor and a ceiling.
Uniqueness for the confined Gaussian free field with general ferromagnetic (possibly long-range) couplings was established by Ferrari--Grynberg~ \cite{Ferrari:NoPhaseTransition2007} at all temperatures.

Our goal is to investigate the effect of the hard wall confinement $|\varphi_i|\le S$ for \emph{integer-valued} $\nabla\varphi$ height functions (which can be seen as a spin-$S$ model). Will such models again have a unique Gibbs measure with exponential decay of correlations? 

It is a classical fact that uniqueness and exponential decay hold at \emph{high temperature}, by the Dobrushin uniqueness condition. Further, the entropic repulsion intuition, that concentrating on the middle layer provides more room for low-energy excitations, can be justified at sufficiently \emph{low temperatures} by Pirogov--Sinai theory (see, e.g.,~\cite[Appendix 3]{Bricmont:RandomSurfacesStatistical1986a}, Bricmont--Slawny~\cite[Section III]{Bricmont:FirstOrderPhase1986} for $S=1$, Mazel--Suhov~\cite{Mazel:RandomSurfacesTwosided1991} for the SOS-model and Lebowitz--Mazel--Suhov \cite{Lebowitz:IsingInterfaceTwo1996} for a generalized Ising model with an odd number of spin states). For the latter fact, our paper develops a new proof which avoids Pirogov--Sinai theory and relies instead on a comparison to a model with a magnetic field. All these results, however, leave open the behavior in the \emph{intermediate temperature regime}.

Let us now highlight our main contributions. Our first main result is that uniqueness and exponential decay of correlations hold in \emph{two dimensions at all temperatures}, for all even and convex interaction energies $W$, and all integers $S\ge 1$. Somewhat surprisingly, we discover a more complicated picture in three and higher dimensions. On the one hand, uniqueness and exponential decay continue to hold \emph{for all temperatures} and integers $S\ge 1$ for a wide range of interaction energies, including the family $W(x) = |x|^p$ for $1\le p\le 2$ (in particular, for the Solid-On-Solid and integer-valued Gaussian free field models). 
On the other hand, in high dimensions, we give examples of even and convex interaction energies for which there are \emph{multiple Gibbs measures} at certain temperatures!

In addition, we show that even when uniqueness fails, the models are still approximately centered in the sense that in every dimension $d \geq 2$, temperature, and $S \in \Z_{>0}$, the magnetization of every Gibbs measure lies in the interval $(-\frac12, \frac12)$.

\smallskip
Finally, we digress from the integer-valued setup and briefly consider other options for the spin values. First, we revisit the real-valued case, $\varphi:\Z^d\to[-1,1]$ (studied by~\cite{McBryan:DecayCorrelationsSymmetric1977,Bricmont:RandomSurfacesStatistical1986a,DAlimonte:FreeEnergyAnalyticity2026} for $W(x) = x^2$). Our methods adapt to this case and show that uniqueness of Gibbs measures is a generic phenomenon, occurring for \emph{all} even and convex interactions $W$ (in all dimensions and temperatures). We have not explored the question of exponential decay of correlations in the real valued case.

Second, we consider height functions valued in an \emph{even} number of states, $\varphi:\Z^d\to\{-S+\frac{1}{2},\ldots,S-\frac{1}{2}\}$, $S \in \Z_{>0}$. Here, one may expect multiplicity of Gibbs measures to be the norm at low temperatures (as occurs, e.g., for the Ising model, obtained when $S=1$ and $W(1)>W(0)$). Nevertheless, our methods adapt to this case, allowing to conclude that the magnetization of all Gibbs measures must lie in $(-\frac12,\frac12)$, for all even and convex interactions $W$. A related result for a generalized Ising model with an even number of spin states is given in \cite{Lebowitz:IsingInterfaceTwo1996}, where it is proved that the model has precisely two translation-periodic extremal Gibbs measures at sufficiently low temperatures.

\subsubsection*{Motivation and related literature}

We briefly discuss additional background and motivation for the study of interfaces confined between two hard walls.

First, one physical motivation for this model is the \emph{commensurate-incommensurate transition}. A system in the incommensurate phase close to the transition point consists of regions of commensurate order, which are separated by domain walls. These walls can be modeled as a gas of random surfaces, which interact through the constraint that they must not intersect (see also~\cite{adhikari2026rigidity}). As a simplification for a uniaxial system, Fisher--Fisher~\cite{Fisher:WallWanderingDimensionality1982} proposed to consider a sequence of stacked random interfaces in which every other interface is flat, so the model reduces to the study of a random surface confined between two walls.

This point of view was taken up by Bricmont--El Mellouki--Fröhlich~\cite[Section 3]{Bricmont:RandomSurfacesStatistical1986a} (see also~\cite[Section 4]{Velenik:LocalizationDelocalizationRandom2006a}) who studied random surfaces confined between a floor and a ceiling in both the real- and integer-valued cases. In the integer-valued case, they focused on the Gaussian free field interaction and obtained the following results for the infinite-volume limit with \emph{zero boundary values}: (1) In dimensions $d\ge 3$ and in dimension $d=2$ at low temperatures, the model exhibits exponential decay of correlations for any integer $S\ge 1$ (for $d\ge 3$ they rely on a result of Göpfert--Mack~\cite{Gopfert:ProofConfinementStatic1982b} for the unconfined model). (2) In dimension $d=2$, the model is never ``critical'' in the sense that its susceptibility (i.e., the sum of its two-point function) is finite at all temperatures. Our work complements their study in two ways. First, we show that there is a unique Gibbs measure in all dimensions $d\ge 2$ and all temperatures. Second, they point out that  exponential decay remains open in dimension $d=2$ at intermediate and high temperatures, and ask whether ``some kind of phase transition'' takes place there. We show that exponential decay persists at all temperatures (ruling out an analog of a Berezinskii--Kosterlitz--Thouless transition in the confined model). We emphasize that the Sokal argument from the real-valued case, which derived uniqueness of Gibbs measures from the exponential decay of correlations in the zero-boundary values measure, does not apply in the integer-valued case.

Second, the case $S=1$ and $W(x) = \frac{1}{2} x^2$ of our setup is a special case of the \emph{Blume-Capel} model. The latter is the model on $\varphi:\Z^d\to\{-1,0,1\}$ with Hamiltonian
\begin{equation}\label{eq:Blume-Capel Hamiltonian}
    \mathscr H^{\text{BC}}(\varphi) :=  -\beta\sum_{\set{i,j} \in \mathcal E(\Z^d)} \varphi_i \varphi_j - \Delta\sum_i \varphi_i^2,
\end{equation}
which reduces to~\eqref{eq:formal hamiltonian without rho} when $\Delta = -2d\beta$. The Blume-Capel model can be thought of as a dilute version of the Ising model (with $\varphi_i=0$ indicating that $i$ is absent from the lattice). It has received much attention since its introduction by Blume~\cite{Blume:TheoryFirstOrderMagnetic1966} and Capel~\cite{Capel:PossibilityFirstorderPhase1966}; see~\cite{Bricmont:FirstOrderPhase1986} and~\cite{Friedli:StatisticalMechanicsLattice2017} for a mathematical treatment of the low temperature phase diagram.
It is known that for every $\Delta\in\R$ the model undergoes a magnetization phase transition at some critical value $\beta_c(d,\Delta)\in(0,\infty)$ (see~\cite[Theorem 1.1]{Gunaratnam:ExistenceTricriticalPoint2024a}). Recently, Gunaratnam--Krachun--Panagiotis~\cite{Gunaratnam:ExistenceTricriticalPoint2024a} established several fundamental facts on the model, including details of its critical behavior and a proof of exponential decay of correlations throughout the subcritical regime. The exponential decay result was very recently extended by Panagiotis--Veitch~\cite{Panagiotis:SubcriticalSharpnessRealvalued2026} to a large class of generalized Ising models, including the Hamiltonian~\eqref{eq:Blume-Capel Hamiltonian} when $\varphi$ satisfies our restriction $|\varphi_i|\le S$. A further result on sub-critical sharpness in a class of Ising-like models was very recently obtained by van Engelenburg--Heeney--Lis \cite{Engelenburg:DoubleClusterSwapping2026} by introducing a new representation for these models.

Our uniqueness result implies that the point $\Delta=-2d\beta$ always lies in the subcritical regime\footnote{Conversely, Pirogov--Sinai theory may be used to show that for any $\varepsilon>0$, the model with $\Delta = -(2d-\varepsilon)\beta$ has multiple Gibbs measures for large enough $\beta$. An explicit theorem implying this is~\cite[Theorem 2.2 with condition (18)]{Peled:LongrangeOrderDiscrete2020}, after reweighting to make all single-site weights $1$ using Section 3.4.1 there.}. Consequently, for the interaction $W(x)=\frac{1}{2}x^2$, exponential decay of correlations follows from~\cite{Gunaratnam:ExistenceTricriticalPoint2024a} (for $S=1$) and~\cite{Panagiotis:SubcriticalSharpnessRealvalued2026} (for general $S$). Our results show that such exponential decay is a general property of confined $\nabla\varphi$ height functions.

Third, it is also interesting (e.g., in connection with the notion of wetting) to consider a surface restricted to lie above a floor, in the sense that $\varphi_i\ge 0$ everywhere. Here, entropy considerations may repel the surface from the floor, lifting it to a height that grows with the system size. This entropic repulsion phenomenon has been established at low temperatures (where Pirogov--Sinai theory is applicable) and also when the unrestricted surface is rough, but appears to remain open in other temperature ranges~\cite[Sections 4.1 and 4.2]{Bricmont:RandomSurfacesStatistical1986a},~\cite[Section 3]{Velenik:LocalizationDelocalizationRandom2006a},~\cite{lubetzky2016harmonic}. As observed in~\cite{Bricmont:RandomSurfacesStatistical1986a}, the phenomenon will occur whenever the surface confined between two hard walls (at arbitrary finite distance from each other) is repelled to the middle layer. Our results imply that this delocalization happens at \emph{all temperatures} and in \emph{all dimensions} for all convex and even $W$, in both the integer- and real-valued cases.

\subsection{Main results}
\label{sec:main_results}

We proceed to formally describe our results. Let $S \in \Z_{>0}:=\{1,2,\ldots\}$. The model is described by the set of configurations
\begin{equation}
    \Omega :=  \set{-S, \dots, S}^{\Z^d}
\end{equation}
and the formal Hamiltonian\footnote{From this point on we opt to denote configurations by $\sigma$ rather than $\varphi$, adhering to the more common notation for spin systems with compact state spaces.}
\begin{equation}\label{eq:def_hamiltonian}
    \mathscr H_{\beta, \rho}(\sigma) := \beta \sum_{\set{i,j} \in \mathcal E(\Z^d)} W\left(\sigma_i - \sigma_j\right) + \rho \sum_{i \in \Z^d} \sigma_i^2, \quad \sigma \in \Omega,
\end{equation}
where $W \colon \{-2S,\ldots, 2S\} \to \R$ is the interaction energy, $\beta>0$ the interaction strength, or inverse temperature, parameter and $\rho \in \R$ a parameter; our focus is on the case $\rho=0$, but we allow general (possibly negative!) $\rho$, both as a natural extension of our results and as some of our proofs make use of it. A positive $\rho$ may be thought of as a ``squared mass''.

We assume throughout that $W$ is even (i.e., $W(n)=W(-n)$ for every $n$) and convex, with the latter meaning that 
\begin{equation}\label{eq:second discrete derivative of W}
    W^{(2)}(n) := W(n+1)- 2 W(n)+W(n-1)  \geq 0
\end{equation}
for all $n$ with $|n|<2S$. If strict inequality holds in~\eqref{eq:second discrete derivative of W} for all such $n$ then we say that $W$ is strictly convex (but this is not generally assumed). 
It is convenient to normalize $W$ by requiring that $W(0)=0$.

Our interest is in the infinite-volume behavior of the model, as captured by the notion of Gibbs measures (cf.~\zcref{sec:gibbs_measures})
Let $\mathscr G(d, \beta,\rho)$ be the set of all Gibbs measures for a given dimension $d$, interaction strength $\beta>0$ and $\rho \in \R$. We will study under which assumptions we have uniqueness in the sense that 
\begin{equation}
    \abs{\mathscr G(d, \beta, \rho)} = 1
\end{equation}
and, when uniqueness holds, under which assumptions the correlations of the infinite volume measure decay exponentially fast.
For the latter goal, we say that a measure $\mu$ on $\Omega$ has \emph{exponential decay of correlations with constants $c_1,c_2>0$} if for all $f,g \colon \Omega \to \R$ local,
\begin{equation}\label{eq:exponential_decay_local_functions_main}
\abs{\mu\left(fg\right) - \mu\left(f\right) \mu\left(g\right)}
    \leq \norm{f}_\infty \norm{g}_\infty c_1^{\abs{\supp f} + \abs{\supp{g}}} \expa{-c_2 \dist\left(\supp f, \supp g\right)},
\end{equation}
where $\mu({\cdot})$ denotes the expectation with respect to $\mu$ and where $f$ is called \emph{local} if there is a finite $\Delta \Subset \Z^d$ such that whenever $\sigma_1,\sigma_2\in\Omega$ agree on $\Delta$ then $f(\sigma_1)=f(\sigma_2)$. The smallest such $\Delta$ is denoted by $\supp f$. 

\medskip
The following are our main results. Again, our focus is on the $\rho=0$ case. 
First, it is standard that uniqueness holds at high temperatures. The following theorem, for which we give two proofs (providing slightly different constants), makes this explicit.
\begin{theorem}[High temperature]\label{thm:main_high_temperature}
    Let $d \geq 2$ and $\rho \in \R$. Assume that $W$ is even, convex and normalized such that $W(0)=0$.
    If $\beta < \frac{1}{W(2S)d}$ then $\mathscr G(\beta,\rho) = \set{\mu_{\beta,\rho}}$ and $\mu_{\beta,\rho}$ has exponential decay of correlations with constants $c_1,c_2>0$ depending on $d, \beta, W, S$.
\end{theorem}
\begin{remark}
    In fact, as the proof using Dobrushin's uniqueness criterion (cf.~\zcref{sec:high_temperatures}) shows,~\zcref{thm:main_high_temperature} continues to hold for a general interaction $W:\Z\to\R$ (without assuming that it is even or convex) if one replaces $W(2S)$ by $\max_{s,t \in \set{0, \dots, 2S}} \abs{W(s) - W(t)}$. Also, note that the constant $\rho$ did not enter in any of the estimates.
\end{remark}

Second, Pirogov--Sinai methods can be used to justify uniqueness and exponential decay at low temperatures, due to entropic repulsion. We provide an alternative proof, which relies on a comparison to a model with a magnetic field, yielding the following result.
\begin{theorem}[Low temperature]\label{thm:main_low_temperature}
    Let $d \geq 2$ and assume that $W$ is even, convex and satisfies $0=W(0)<W(1)$. For any $\beta > \frac{4(\ln d+\ln 8)}{W(1) d}$, there exists $\rho_*(\beta) > 0$ such that for every $\rho \in (-\rho_*, \infty)$ we have $\mathscr G(\beta,\rho) = \set{\mu_{\beta,\rho}}$ and  $\mu_{\beta,\rho}$ has exponential decay of correlations with constants $c_1 \equiv c_1(S, d, \beta, W) > 0$ and $c_2 \equiv c_2(d, \beta, W) > 0$.
\end{theorem}
\begin{remark} 
Our proof also yields progress on the problem of ``percolation of finite clusters''~\cite{Grimmett:PercolationFiniteClusters2014a,Bock:PercolationFiniteClusters2020}, as we describe in \zcref{rem:percolation of finite clusters}.
\end{remark}

Given the previous two results, our main interest is whether uniqueness and exponential decay in fact hold at \emph{all} temperatures. Our next result shows that this is always the case in two dimensions.
\begin{theorem}[Two dimensions, all temperatures]\label{thm:main_two_dimensions}
    Let $d=2$ and assume that $W$ is even and convex. 
    For any $\beta > 0$ 
    there exists $\rho_*(\beta) > 0$ such that for every $\rho \in (-\rho_*, \infty)$ we have $\mathscr G(\beta,\rho) = \set{\mu_{\beta,\rho}}$ and  $\mu_{\beta,\rho}$ has exponential decay of correlations with constants 
    $c_1 \equiv c_1(S) > 0$ and
    $c_2 \equiv c_2(S, \beta, \rho, W, d) > 0$.
\end{theorem}
From the above, and the entropic repulsion intuition, it is natural to speculate that uniqueness and exponential decay hold at all temperatures also in dimensions $d\ge 3$. Interestingly, we discover that this is not necessarily the case. Our results are described in the following two theorems, the first of which gives sufficient conditions for uniqueness and exponential decay to hold at all temperatures and the second of which describes examples of non-uniqueness.

If we restrict the class of allowed interactions $W$, \zcref{thm:main_two_dimensions} extends to all dimensions. To this end we introduce the following definitions: First, following Lammers--Ott~\cite{Lammers:DelocalisationAbsolutevalueFKGSolidonsolid2024}, we say that $W$ is \emph{super-Gaussian} if  $W^{(2)}$
is non-increasing for integer $n\ge 0$ (with $W^{(2)}$ defined in~\eqref{eq:second discrete derivative of W}). Second, we call $W \colon \R \to \R$ a \emph{mixture of Gaussians} if for every $\beta>0$ there exists a finite non-negative Borel measure $\psi_\beta$ on $[0,\infty)$ such that for every $n \in \Z$,
    \begin{equation}
        \expa{-\beta W(n)} = \int_{[0, \infty)} \expa{-\frac12 a n^2} \d \psi_\beta(a).
    \end{equation}
As an important example, the function $W(n) = |n|^p$ for $0<p\le 2$ is both super-Gaussian and a mixture of Gaussians~\cite[Section 9.4]{Aizenman:DepinningIntegerrestrictedGaussian2022}, and in the smaller range $1\le p\le 2$ it is also convex and thus satisfies the assumptions of the next result.

\begin{theorem}[Mixtures of Gaussians, all temperatures]\label{thm:main_mixtures_of_gaussians}
    Let $d \geq 2$ and assume that $W$ is even, convex (as a function $\R \to \R$), super-Gaussian and a mixture of Gaussians. For every $\beta > 0$ the following holds:
    
    There exists $\rho_*(\beta) > 0$ such that for every $\rho \in (-\rho_*, \infty)$ we have $\mathscr G(\beta,\rho) = \set{\mu_{\beta,\rho}}$. If in addition $W$ is strictly convex or if $W(x)=\abs{x}$, then $\mu_{\beta,\rho}$ has exponential decay of correlations with constants 
    $c_1 \equiv c_1(S) > 0$ and
    $c_2 \equiv c_2(S, \beta, \rho, W, d) > 0$.
\end{theorem}

\begin{theorem}[Examples of non-uniqueness] \label{thm:non-uniqueness}
    Let $\beta > 0$, $\rho \in \R$ and let $W$ be even. Suppose that there exists $1 \leq k \leq 2S-1$ such that $W(l) = 0$ for all $0 \leq l \leq k$ and $W(l)>0$ for all $k+1 \leq l \leq 2S$. If $\rho=0$, we assume that $k=2S-1$. If $\rho > 0$, we assume that $k$ is odd. Then, there exists $d_0=d_0(W, \beta,\rho)$ such that $\abs{\mathscr G(d, \beta,\rho)} \geq 2$ in every dimension $d\ge d_0$.
\end{theorem}

The proof is an application of a general result of ~\cite{Peled:LongrangeOrderDiscrete2020}. The results there also imply that (under the assumptions of \zcref{thm:non-uniqueness}) the model has exactly two ergodic Gibbs measures, which are also extremal, and provide structural information on samples from these measures.

To gain intuition for the behavior, consider the case $S=1$, $W(0)=W(1)=0$ and $W(2)=\infty$, which is known as the Widom--Rowlinson model; see~\cite[Section 3.1.3]{Peled:LongrangeOrderDiscrete2020} for a discussion of its history and its equivalence to the hard-core model on $\Z^d\times\{0,1\}$ (for the latest results on non-uniqueness in the hard-core model see~\cite{hadas2026critical}). For the Widom--Rowlinson model with, say, $\rho=0$ in high dimensions, there will be a Gibbs measure in which configurations take the values $0$ and $1$ with densities at least $\frac{1}{2}-e^{-cd}$ for some $c>0$ (and, by symmetry, another measure with such densities for $-1$ and $0$). 

We believe that further examples of non-uniqueness may be obtained with suitable extensions of the proof technique in~\cite{Peled:LongrangeOrderDiscrete2020}. For instance, we believe that for $W(n)=|n|^p$ with $p$ and $d$ suitably large, there will be an \emph{intermediate} temperature range where uniqueness fails. We plan to explore this in a forthcoming paper.

While the above result shows that entropic repulsion need not always lead to a unique Gibbs measure, it turns out that it still severely constrains the set of possible Gibbs measures. Our next result shows that the average value taken by any such Gibbs measure must lie in $(-\frac12,\frac12)$ and that the chance that the spin at the origin is positive (or the chance that it is negative) is less than $\frac12$, irrespective of the value of $S$.

\begin{theorem}\label{thm:constrained Gibbs measures}
    Let $d\ge2$ and assume that $W$ is even and convex.    
    For any $\beta > 0$ there exists $\rho_*(\beta) > 0$ such that for every $\rho \in (-\rho_*, \infty)$, every Gibbs measure $\mu\in \mathscr G(\beta,\rho)$ satisfies
    \begin{equation}
        |\mu(\sigma_0)|<\frac{1}{2}
    \end{equation}
    and
    \begin{equation}
        \text{$\mu(\sigma_0>0)< \frac{1}{2}$\quad and\quad $\mu(\sigma_0<0)< \frac{1}{2}$}.
    \end{equation}
\end{theorem}
Lastly, we point out a consequence of our results for integer-valued $\nabla\varphi$ height functions constrained to lie above a floor. By this, we mean the model on $\sigma:\Z^d\to\{0,1,\ldots\}$ with the formal Hamiltonian~\eqref{eq:formal hamiltonian without rho} (we do not include the parameter $\rho$ in this case).
\begin{corollary}[Delocalization above a floor, entropic repulsion]
\label{cor:delocalization}
    Let $d\ge 2$ and let $W:\Z\to\R$ be even and convex. Let $\beta > 0$. The model with Hamiltonian~\eqref{eq:formal hamiltonian without rho} and spins $\sigma:\Z^d\to\{0,1,\ldots\}$ has no Gibbs measures.
\end{corollary}
\begin{proof}
    This follows from \zcref{thm:constrained Gibbs measures} and the FKG inequality (see \zcref{sec:prelim}). Indeed, for each integer $S\ge 1$, let $\mu^{S,-}$ be the minimal Gibbs measure for the model with configurations $\sigma:\Z^d\to\{-S,\ldots, S\}$ and the Hamiltonian~\eqref{eq:formal hamiltonian without rho} with the given $W$ and $\beta$. By minimal, we mean the measure obtained as an infinite-volume limit with $-S$ boundary conditions. \zcref{thm:constrained Gibbs measures} implies that $\mu^{S,-}(\sigma_0\ge 0)>\frac{1}{2}$. The FKG inequality implies that any Gibbs measure $\mu$ of the model with spins $\sigma\ge 0$ must stochastically dominate $S+\sigma$ when $\sigma$ is drawn from $\mu^{S,-}$. In particular, $\mu(\sigma_0\ge S)>\frac{1}{2}$. As $S$ is arbitrary, this implies that $\mu$ cannot exist.
\end{proof}

As far as we are aware, even the integer-valued Gaussian free field (the case $W(x)=x^2$) constrained to lie above a floor was not previously known to be delocalized at all temperatures (as \zcref{cor:delocalization} shows). It was only known to be delocalized at low temperatures by Pirogov--Sinai methods~\cite[Section 4.2]{Bricmont:RandomSurfacesStatistical1986a},~\cite
{lubetzky2016harmonic}, and whenever the unconstrained model is delocalized (high temperatures in two dimensions)~\cite[equation (4.7)]{Bricmont:RandomSurfacesStatistical1986a}.

\subsubsection*{Entropic repulsion for functions with values in $\Z+\frac{1}{2}$ or in the reals}

Moving away from our main theme, we briefly consider the entropic repulsion effect for height functions taking values either in $\{-S+\frac{1}{2},\ldots,S-\frac{1}{2}\}$ (an \emph{even} number of states) for some integer $S\ge 1$, or in the real interval $[-1,1]$. As these are not our main focus, we only note that a minor adaptation of our proof of \zcref{thm:constrained Gibbs measures} establishes entropic repulsion also for these cases. In the discrete case (with an even number of states), while non-uniqueness at low temperatures is common (e.g., for the Ising model, obtained when $S=1$ and $W(1)>W(0)$), our methods still apply to show that the magnetization in every Gibbs measure is not far from $0$. In contrast, in the real-valued case, our methods show that uniqueness \emph{always} happens.

\begin{theorem}[Entropic repulsion with an even number of states]
\label{thm:constrained Gibbs measures even}
     Let $S\in \Z_{>0}$ and consider the model on $\sigma:\Z^d\to \set{-S+\frac12, \dots, S-\frac12}$ with Hamiltonian~\eqref{eq:def_hamiltonian}, where $W\colon \set{-2S +1, \dots,2S-1}\to \R$ is an even and convex function (the latter means that~\eqref{eq:second discrete derivative of W} holds for $|n|<2S+1$), $\beta>0$ the interaction strength and $\rho\in\R$ a parameter. 
    Let $d\ge2$ and assume that $W$ is even and convex. 
    For any $\beta > 0$ there exists $\rho_*(\beta) > 0$ such that for every $\rho \in (-\rho_*, \infty)$, every Gibbs measure $\mu\in \mathscr G(\beta,\rho)$ satisfies
    \begin{equation}
        |\mu(\sigma_0)|< \frac12
    \end{equation}
    and
    \begin{equation}
        \mu\left(\sigma_0> \tfrac12\right)< \frac{1}{2}\quad \text{and}\quad \mu\left(\sigma_0< -\tfrac12\right) < \frac{1}{2}.
    \end{equation}
\end{theorem}

\begin{theorem}\label{thm:real_valued_entropic_repulsion}(Entropic repulsion with real values)
    Consider the model on $\sigma:\Z^d\to [-1,1]$ with Hamiltonian~\eqref{eq:def_hamiltonian}, where $W\colon [-2,2]\to\R$ is an even and convex function, $\beta>0$ the interaction strength and $\rho\in\R$ a parameter. Let $d\ge 2$. For any $\beta > 0$ 
    there exists $\rho_*(\beta) > 0$ such that for every $\rho \in (-\rho_*, \infty)$ the model has a \emph{unique} Gibbs measure.
\end{theorem}
For the real-valued case, as far as we are aware, uniqueness was only proved for the confined Gaussian free field ($W(x)=x^2$) model~\cite{McBryan:DecayCorrelationsSymmetric1977},~\cite[Proposition A1]{Bricmont:RandomSurfacesStatistical1986a},~\cite[Section 8]{DAlimonte:FreeEnergyAnalyticity2026}, with proofs that seemed to rely on features of this specific interaction energy. \zcref{thm:real_valued_entropic_repulsion} shows that uniqueness is a generic phenomenon for real-valued height functions, holding for all even and convex interaction energies (unlike the integer-valued case where we have shown that additional assumptions are required). Moreover, uniqueness is even shown to hold in an interval of \emph{negative} $\rho$ (with the interval depending on $\beta,W$ and $d$).

One may also deduce from these results and the FKG inequality the delocalization above a floor, with a similar proof to that of \zcref{cor:delocalization}.

As our focus in this paper is on integer-valued height functions, we leave for further study the question of exponential decay of correlations in the unique Gibbs measure of the confined real-valued height function.

\subsection{Organization of the paper and proof ideas}

We start with some preliminaries and notation in \zcref{sec:prelim} and present a sufficient condition for uniqueness of the infinite volume Gibbs measure and for exponential decay of correlations in \zcref{sec:uniqueness}. 

We make the following observation: If one extends the model by a percolation measure on the edges, with edge-probabilities appropriately chosen, then the spins in any open cluster are symmetric under a spin flip around $0$. This implies that $\mu^+(\sigma_0 \mid \mathcal A)=0$, where $\mathcal A$ is the event that open edges do not percolate. We make this argument fully rigorous in \zcref{sec:Bernoulli_representation}, by showing that there is uniqueness if and only if the open edges do not percolate, and that there is exponential decay of correlations, whenever the probability of one-arm events decays exponentially fast. As a direct application of this result, we give the proof of uniqueness at high temperatures (\zcref{thm:main_high_temperature}) in \zcref{sec:high_temperatures} (and an alternative using Dobrushin's uniqueness criterion).

An important innovation in our proof is the following \emph{stochastic domination}: our model (which takes values in $\set{-S, \dots, S}$) with $+S$ boundary conditions is dominated by a model with the same interaction, also with $+S$ boundary conditions, but taking values in the smaller space $\set{-S+1, \dots, S}$ and additionally having a \emph{magnetic field} which favors spins $\leq 0$. The advantage of the latter model is the symmetry of its state space around $\frac12$, which, making use of the magnetic field, allows to deduce that its magnetization decays to a limit strictly less than $\frac12$, exponentially fast in the side length of the domain (at any temperature). This is already enough to prove \zcref{thm:constrained Gibbs measures}, which says that any Gibbs measure, even in the non-uniqueness regime, has magnetization in $(-\frac12,\frac12)$. It also suffices, by adapting the argument, for the corresponding results in the case of an even number of spins and for real valued systems (\zcref{thm:constrained Gibbs measures even,thm:real_valued_entropic_repulsion}). In two dimensions, using planarity, we can upgrade the magnetization decay of the dominating model to the exponential decay of clusters of positive spins. This yields \zcref{thm:main_two_dimensions} (see \zcref{sec:minority_percolation_2d}) using the percolation representation of~\zcref{sec:Bernoulli_representation}. In higher dimensions, exponential decay of clusters of positive spins is proven at low temperature (yielding \zcref{thm:main_low_temperature}, cf. \zcref{sec:minority_percolation_low_temperature}).

To prove uniqueness at all temperatures in all dimensions (\zcref{thm:main_mixtures_of_gaussians}) for the restricted class of $W$, we dominate our model differently, by a \emph{massive} integer valued field with the same interaction and $+S$ boundary conditions (\zcref{sec:uniqueness_mixture_of_gaussians}). This domination is proved following an argument of \cite{DAlimonte:FreeEnergyAnalyticity2026} (which uses the super-Gaussianity), where the authors establish such a domination for real valued fields with Gaussian interaction $W(x)=x^2$. We then wish to control the magnetization of the massive integer-valued field by a comparison to its real-valued version (which is itself controlled via the Brascamp--Lieb inequality), but this comparison is valid only for zero boundary conditions (the comparison uses the mixture of Gaussians assumption~\cite[Lemma D.1]{Aizenman:DepinningIntegerrestrictedGaussian2022}). To proceed, we apply the comparison in a box of side length $L$ with zero boundary conditions, to deduce that the probability that the \emph{average spin} exceeds $\epsilon>0$ is $\mathcal O(\exp(-c(\eps)L^d))$. Then, we argue that replacing the zero boundary condition by a $+S$ boundary condition only contributes $\mathcal O(L^{d-1})$ energy to the system and hence cannot change this probability significantly. Via the domination, these arguments imply that the magnetization of the confined model tends to $0$ in the infinite-volume limit, establishing the uniqueness.

We wish to derive the exponential decay of correlations from the uniqueness result, by adapting to our setup the ``sharpness of the phase transition'' technique of~\cite{Duminil-Copin:SharpPhaseTransition2019}, relying on establishing a differential inequality via the OSSS inequality. To this end, we need to embed our model inside a one-parameter family of ``sub-critical models'', with the family chosen so that the OSSS inequality provides control on the derivative of a suitable connectivity event (for the edge percolation of~\zcref{sec:Bernoulli_representation}) with respect to the family's parameter. It turns out that neither the natural temperature parameter nor our parameter $\rho$ yield suitable parameterizations. Instead, we extend the edge percolation measure to a generalization of the dilute random cluster (DRC) measure (introduced in~\cite{Graham:RandomClusterRepresentationBlume2006} to study the Blume--Capel model). In a similar spirit to the analysis of models with a quadratic interaction~\cite{Gunaratnam:ExistenceTricriticalPoint2024a,Panagiotis:SubcriticalSharpnessRealvalued2026}, we find in this generalized DRC measure a parameter suitable for the differential inequality technique. To prove ``sub-criticality'' of this parameterized family, we show that the connectivity observable in the DRC measure is dominated by the magnetization of the original spin model, with a lower temperature and smaller value of $\rho$ (and it is thus useful that our uniqueness result applies at all temperatures and even for small negative $\rho$). The last obstacle is the fact that the DRC representation is not monotone in the sense needed for the OSSS inequality. However,~\cite{Gunaratnam:ExistenceTricriticalPoint2024a} showed that the OSSS inequality applies also for measures on $\set{0, 1}^V \times \set{0,1}^E$ (where $G=(V,E)$ is a finite graph) which are \emph{weakly} monotone in a suitable sense. We further extend this to measures on $\set{0, \dots, S}^V \times \set{0,1}^E$ (by a reduction to the $S=1$; see~\zcref{appendix:OSSS}) and prove that our measure satisfies the required weak monotonicity (an alternative generalization of the OSSS inequality is in~\cite{Panagiotis:SubcriticalSharpnessRealvalued2026}). 

The main part of the article concludes with a proof of \zcref{thm:non-uniqueness} in \zcref{sec:non-uniqueness}, by showing that the interaction energies of that theorem lead to multiple (equivalent) dominant patterns in the sense of~\cite{Peled:LongrangeOrderDiscrete2020}, from which it follows that multiple Gibbs measures arise in high dimensions.

\section{Preliminaries and notation}
\label{sec:prelim}

\subsection{Gibbs measures}\label{sec:gibbs_measures}

To formally define the notion of Gibbs measures, we need to first explicitly define the model in finite volume. Fix an integer $S\ge 1$ and an interaction energy $W$ satisfying the assumptions in \zcref{sec:main_results}.
For a finite subset $\Lambda \Subset \Z^d$, the finite-volume version of the Hamiltonian~\eqref{eq:def_hamiltonian} is
\begin{equation}\label{eq:def_finite_volume_hamiltonian}
    \mathscr H_{\Lambda; \beta, \rho}(\sigma) := \beta \sum_{\set{i,j} \in \mathcal E^b(\Lambda)} W\left(\sigma_i - \sigma_j\right) + \rho \sum_{i \in \Lambda} \sigma_i^2, \quad \sigma \in \Omega,
\end{equation}
where
\begin{equation}
    \mathcal E^b(\Lambda) := \set{\set{i, j} \subset \Z^d \mid \abs{i-j}_1=1, \,i \in \Lambda \text{ or } j \in \Lambda}
\end{equation}
is the set of nearest neighbor edges with at least one endpoint in $\Lambda$.

Given a boundary condition $\eta \in \Omega$ we let
\begin{equation}
    \Omega_\Lambda^\eta := \set{\sigma \in \Omega \mid \sigma_i = \eta_i \text{ for } i \in \Lambda^c}
\end{equation}
be the set of configurations which coincide with $\eta$ outside of $\Lambda$. For $\eta \in \Omega$ we define a probability measure on $(\Omega, \mathcal F)$, where $\mathcal F$ is the $\sigma$-algebra generated by events depending on a finite number of sites, by
\begin{equation}\label{eq:def_P_Lambda}
    \mu_{\Lambda;\beta,\rho}^\eta(\sigma) := \begin{cases}
        \left(Z_{\Lambda;\beta, \rho}^\eta\right)^{-1} e^{-{\mathscr H}_{\Lambda;\beta,\rho}(\sigma)} & \text{if } \sigma \in \Omega_\Lambda^\eta, \\
        0 & \text{otherwise,}
    \end{cases}
\end{equation}
with ${Z}_{\Lambda;\beta,\rho}^\eta$ the partition function (or normalization factor). Note that $\mu_{\Lambda;\beta,\rho}^\eta$ is well defined, in the sense that $0<{Z}_{\Lambda;\beta,\rho}^\eta<\infty$, by our assumptions. Given a function $f \colon \Omega \to \R$, we write $\mu_{\Lambda;\beta,\rho}^\eta(f)$ for the associated expectation.

Since there is a natural order on $\Omega$ and our model satisfies the FKG inequality (cf.~\zcref{sec:prelim}), we will often work with the extremal measures, which are the limits of finite volume measures with $S$ and $-S$ boundary conditions. We will call them the $+$ and $-$ measure respectively, writen as $\mu^+_{\beta,\rho}$ and $\mu^-_{\beta,\rho}$. Accordingly, we set in finite volume
\begin{equation}
    \Omega^\pm_\Lambda := \set{\sigma \in \Omega \mid \sigma_i = \pm S \text{ for } i \in \Lambda^c}
\end{equation}
and shorthand $\mu_{\Lambda;\beta,\rho}^+ := \mu_{\Lambda;\beta,\rho}^S$ and $\mu_{\Lambda;\beta,\rho}^- := \mu_{\Lambda;\beta,\rho}^{-S}$.

Finally, an (infinite-volume) \emph{Gibbs measure} associated with the Hamiltonian~\eqref{eq:def_hamiltonian} is a probability measure $\mu$ on $(\Omega, \mathcal F)$, which satisfies the following DLR condition for every finite $\Lambda \Subset \Z^d$:
\begin{equation}\label{eq:DLR}
    \mu (A) = \int_{\Omega} \mu_{\Lambda;\beta,\rho}^\eta(A) \d\mu(\eta), \quad \text{for all } A \in \mathcal F.
\end{equation}

We remark that though our main results on exponential decay stated in Section~\ref{sec:main_results} are phrased for (infinite-volume) Gibbs states, our proofs will actually establish exponential decay in finite volume, for the measures $\mu_{\Lambda;\beta,\rho}^\pm$ in cube domains.

\subsection{Preliminaries}

\begin{theorem}[FKG condition]\label{thm:FKG}
    Let $\nu$ be a probability measure on $\Z^{\Z^d}$ which satisfies the \emph{FKG lattice condition}
    \begin{equation}\label{eq:FKG_lattice_condition}
        \nu(\psi \vee \varphi) \nu(\psi \wedge \varphi) \geq\nu(\psi)\nu(\varphi) 
        \quad
        \text{for all } \psi,\varphi \in \Z^{\Z^d}.
    \end{equation}
    Then
    \begin{equation}
        \nu(f g) \geq \nu(f)\nu(g)
    \end{equation}
    for any two non-decreasing functions $f,g \colon \Z^{\Z^d} \to \R$.

    Similarly, if $\nu$ satisfies the \emph{absolute value FKG lattice condition}~\cite{Lammers:DelocalisationAbsolutevalueFKGSolidonsolid2024} 
    \begin{equation}
        \nu(\abs{\sigma} = \psi \vee \varphi) \nu(\abs{\sigma} = \psi \wedge \varphi) \geq\nu(\abs{\sigma} = \psi)\nu(\abs{\sigma} = \varphi) 
        \quad
        \text{for all } \psi,\varphi \in \set{0, 1, \dots}^{\Z^d},
    \end{equation}
    then
    \begin{equation}
        \nu(f(\abs\sigma) g(\abs\sigma)) \geq \nu(f(\abs\sigma))\nu(g(\abs\sigma))
    \end{equation}
    for any two non-decreasing functions $f,g \colon \set{0,1,\dots}^{\Z^d} \to \R$.
\end{theorem}

\begin{remark}\label{remark:FKG_for_convex_W}
    Let $\Lambda \Subset \Z^d$ and $\psi \in \Z^{\Z^d}$, and assume that 
    \begin{equation}
        \nu(\sigma) \propto \begin{cases}
            e^{-\beta \sum_{i,j \in \Lambda} W(\sigma_i - \sigma_j) - \sum_{i \in \Lambda} V_i(\sigma_i)} & \text{if } \sigma_i = \psi_i \text{ for all } i \notin \Lambda\\
            0 & \text{else},
        \end{cases}
    \end{equation}
    with $W, V_i \colon \Z \to \R$. If $W$ is even and convex, then $\nu$ satisfies the FKG lattice condition. Moreover, if $W$ is super-Gaussian, then $\nu$ satisfies the absolute value FKG lattice condition~\cite[Theorem 2.8]{Lammers:DelocalisationAbsolutevalueFKGSolidonsolid2024}. This remains true if one restricts the spin space to a finite interval of $\Z^d$.
\end{remark}

\begin{remark}
    Since $\mathscr H_{\Lambda;\beta,\rho}$ defines a quasilocal specification (cf.~\cite[Section 6.4]{Friedli:StatisticalMechanicsLattice2017}), for any boundary condition $\eta \in \Omega$ there exists a measure $\mu_{\beta;\rho}^\eta \in \mathscr G(\beta,\rho)$ such that $\mu_{\Lambda;\beta,\rho}^\eta \Rightarrow \mu_{\beta,\rho}^\eta$ as $\Lambda \uparrow \Z^d$ for some subsequence, which shows that $\mathscr G(\beta,\rho) \neq \emptyset$. 
\end{remark}

A consequence of the $FKG$ inequality is the existence of the thermodynamic limit for $+$ and $-$ boundary conditions for every sequence $\Lambda \uparrow \Z^d$.
\begin{theorem}[Adaptation of {\cite[Thm 3.17]{Friedli:StatisticalMechanicsLattice2017}}]
    Let $\beta > 0$ and $\rho \in \R$. For any sequence $\Lambda_n \uparrow \Z^d$ we have
    \begin{equation}
        \mu_{\Lambda_n;\beta,\rho}^\pm \Rightarrow \mu_{\beta,\rho}^\pm \text{ as } n\to \infty
    \end{equation}
    in the sense of weak convergence of measures, independent of the sequence $\Lambda_n$. In particular, for any local $f \colon \Omega \to \R$ (i.e., it depends only on the values of finitely many spins, which we call the support $\supp f$) we have
    \begin{equation}
        \mu_{\Lambda_n;\beta,\rho}^\pm\left(f\right) \to \mu_{\beta,\rho}^\pm\left(f\right)
    \end{equation}
\end{theorem}

Finally, we remark the following standard fact.

\begin{theorem}\label{thm:mu+_mu-_maximal}
    For all $\beta > 0$ and $\rho \in \R$ the measures $\mu_{\beta,\rho}^\pm$ are translation invariant and ergodic. Moreover, for any Gibbs measure $\mu \in \mathscr G(\beta,\rho)$ it holds $\mu_{\beta,\rho}^- \leq_{st} \mu \leq_{st} \mu_{\beta,\rho}^+$, where $\leq_{st}$ denotes stochastic domination of measures.
\end{theorem}

\subsection{Notation}

We denote the edges of the \emph{nearest neighbor} lattice $\Z^d$ by $\E(\Z^d)$.
For a graph $G = (V,E)$ such that $V \subset \Z^d$ and $E \subset \E(\Z^d)$, we define the set of exterior edges by
\begin{equation}
    \partial_{\exterior} G := \set{\set{i,j} \in \E(\Z^d) \setminus E \mid \set{i,j} \cap V \neq \emptyset}.
\end{equation}
Further, we write $\Lambda \subset \Z^d$ for a finite subset of $\Z^d$ and it as a subgraph of $\Z^d$ with edges
\begin{equation}
    \mathcal E(\Lambda) := \set{\set{u,v} \subset \Lambda \mid \abs{u-v} = 1} \subset \E(\Z^d).
\end{equation}
For $u,v \in \Z^d$ with $\set{u,v} \in \mathcal E(\Z^d)$ we will write $u \sim v$. We further set
\begin{equation}
    \mathcal E^b(\Lambda) := \mathcal E(\Lambda) \cup \partial_{\exterior} \Lambda
\end{equation}
and moreover define the \emph{boundary} and the \emph{external boundary} of $\Lambda$ by
\begin{equation}
    \partial \Lambda := \set{u \in \Lambda \mid \exists v \in \Z^d \colon \set{u,v} \in \partial_{\exterior} \Lambda}
    \text{ and }
    \partial_e \Lambda := \set{u \in \Lambda^c \mid \exists v \in \Lambda \colon \set{u,v} \in \partial_{\exterior} \Lambda}.
\end{equation}
Finally, for $x \in \Z^d$ and $N \in \N$ let
\begin{equation}
    B_N(x) := \set{y \in \Z^d \mid \abs{x-y}_\infty \leq N}
    \text{ and } B_N := B_N(0)
\end{equation}
be the box of size $N$ centered at $x$.
Further, for $\mu$ and $\mu'$ two probability measures, we write
\begin{equation}
    \mu \leq_{st} \mu'
\end{equation}
if $\mu'$ stochastically dominates $\mu$.

\section{Uniqueness and exponential decay of correlations}
\label{sec:uniqueness}

In this section we give an equivalent condition for uniqueness of the infinite volume limit, namely that the magnetization in the $+$ measure is zero. Moreover, we show that exponential decay of correlations follows from exponential decay of the two-point functions. We remark that similar results have been used before for other FKG models (see e.g.~\cite[Theorem 3.28]{Friedli:StatisticalMechanicsLattice2017} and~\cite{Lebowitz:BoundsCorrelationsAnalyticity1972} for the Ising case).

It will be convenient to introduce the following local functions on $\Omega$. 
\begin{equation}
    n_j := \frac1{2S} (\sigma_j + S), \quad n^K := \prod_{i \in A} n_i^{k_i}, \quad N^K := \sum_{i \in A} k_i n_i,
\end{equation}
where $j \in \Z^d$, $A \Subset \Z^d$ and $K \in \set{0,1,\ldots}^A$.
\begin{lemma}\label{thm:decomposition_of_local_functions}
    For any local function $f \colon \Omega \to \R$, there exist real coefficients $\hat f_K$, $K \in \set{0, \dots, 2S}^{\supp f}$, such that
    \begin{equation}
        f = \sum_{K \in \set{0, \dots, 2S}^{\supp f}} \hat f_K  n^K.
    \end{equation}
    Moreover, there exists $c_0 \equiv c_0(S)$ such that
    \begin{equation}\label{eq:coefficients_interpolating_polynomial}
        \abs{\hat f_K} \leq c_0^{\abs{\supp f}} \norm{f}_\infty
    \end{equation}
    for every  $K \in \set{0, \dots, 2S}^{\supp f}$.
\end{lemma}
\begin{proof}
    We write $f$ as a sum of Lagrange polynomials. For $\sigma \in \Omega$,
    \begin{equation}\label{eq:repr_lagrange_polynomials}
        f(\sigma) = \tilde f(n) =\sum_{\bar n \in \set{0, \frac1{2S}, \ldots, 1}^{\supp f}} \tilde f(\bar n) \prod_{x \in \supp f}\ell_{\bar n_x}(n_x),
    \end{equation}
    where $\tilde f$ is obtained from $f$ after change of coordinates from $\sigma$ to $n$, and where
    \begin{equation}
        \ell_{b}(a) := \prod_{\omega \in \set{0, \frac1{2S}, \ldots, 1} \setminus \set{b}} \frac{a - \omega}{b - \omega},
        \qquad
        a,b \in \set{0, \tfrac1{2S}, \tfrac2{2S}, \ldots, 1}.
    \end{equation}
    These polynomials have the property
    \begin{equation}
        \ell_{b}(a) = \I_{a = b}.
    \end{equation}
    The $\hat f_K$ are obtained by multiplying out the product in \eqref{eq:repr_lagrange_polynomials} and grouping together terms with the same monomial. 
    Further, it follows that each coefficient $\hat f_K$ is bounded by
    \begin{equation}
        \abs{\hat f_K} \leq \norm{f}_{\infty} c_0^{\abs{\supp f}},
    \end{equation}
    where $c_0$ depends only on the one dimensional Lagrange polynomials $\ell_{b}$ and thus only on $S$.
\end{proof}

\begin{lemma}\label{thm:non_decreasing_quantities}
    Let $\sigma \in \Omega$ and $i \in \Z^d$, $A \Subset \Z^d$ and $K \in \set{0,1,\ldots}^A$. Then $\sigma_i$, $n_i$, $N^K$, $n^K$ and $N^K - n^K$ are non-decreasing.
\end{lemma}
\begin{proof}
    It is clear for $\sigma_i$, $n_i$ and $N^K$. Let $\sigma \leq \sigma'$ and let $n^K := n^K(\sigma)$, $N^K := N^K(\sigma)$, $\bar n^K := n^K(\sigma')$ and $\bar N^K := N^K(\sigma')$. We need to show
    \begin{equation}
        n^K \leq \bar n^K \text{ and } N^K - n^K \leq \bar N^K - \bar n^K.
    \end{equation}
    Observe that it suffices to assume that $\sigma$ and $\sigma'$ differ at a single vertex $x \in A$ for which $k_x \geq 1$, since $A$ is finite and the general case thus follows inductively. We have
    \begin{equation}
        n^K = n_x^{k_x} \prod_{i \in A \setminus\set{x}} n_i^{k_i} = n_x^{k_x}  \prod_{i \in A \setminus \set{x}} (n_i')^{k_i} \leq (n_x')^{k_x}  \prod_{i \in A \setminus\set{x}} (n_i')^{k_i} = \bar n^K,
    \end{equation}        
    because $n'_i \geq 0$ and $n_x \leq n'_x$. Moreover,
    \begin{equation}\begin{split}
        N^K - n^K
        &= \sum_{i \in A \setminus \set{x}} k_i n_i + k_x n_x - n_x^{k_x}  \prod_{i \in A \setminus\set{x}} n_i^{k_i} \\
        &= \sum_{i \in A \setminus \set{x}} k_i n_i' + k_x n_x - n_x^{k_x}  \prod_{i \in A \setminus\set{x}} (n_i')^{k_i} 
        = \sum_{i \in A \setminus \set{x}} k_i n_i' + F(n_x),
    \end{split}\end{equation}
    where $F(y) := k_x y - \alpha y^{k_x}$ with $\alpha :=  \prod_{i \in A \setminus \set{x}} (n_i')^{k_i}$. We have that $0 \leq \alpha \leq 1$, thus $F' \geq 0$ on $[0,1]$. Therefore $F(n_x) \leq F(n_x')$. Hence,
    \begin{equation}
        N^K - n^K \leq \sum_{i \in A \setminus \set{x}} k_i n_i' + F(n_x') = \bar N^K - \bar n^K
    \end{equation}
    as claimed.
\end{proof}

\begin{theorem}\label{thm:unique_gibbs_measure_if_magnetization_zero}
    Let $\beta > 0$ and $\rho \in \R$. The following two statements are equivalent:
    \begin{enumerate}
        \item There is a unique infinite volume Gibbs measure on $(\Omega, \mathcal F)$ for $\mathscr H_{\beta,\rho}$, i.e., $\mathscr G(\beta,\rho) = \set{\mu_{\beta,\rho}}$.
        \item It holds that $\mu_{\beta,\rho}^+\left(\sigma_0\right) = 0$.
    \end{enumerate}
\end{theorem}
\begin{proof}
    Assume that $\mathscr G(\beta,\rho) = \set{\mu_{\beta,\rho}}$. Then
    \begin{equation}
        \mu_{\beta,\rho}^+(\sigma_0) = \mu_{\beta,\rho}(\sigma_0) = \lim_{n \to \infty} \mu_{B_n;\beta,\rho}^0 (\sigma_0) = 0,
    \end{equation}
    where 
    $\mu_{B_n;\beta,\rho}^0$ is the measure with zero boundary conditions in $B_n$,
    which is symmetric around $0$. 
    The second equality is due to the fact that $\mu_{B_n;\beta,\rho}^0 \Rightarrow \mu_{\beta,\rho}^0 =\mu_{\beta,\rho}$ along a subsequence and the uniqueness of the infinite volume measure.

    Since the opposite direction is a classical result, we only give a sketch. By \zcref{thm:mu+_mu-_maximal} it suffices to show that $\mu_{\beta,\rho}^+ = \mu_{\beta,\rho}^-$. Since the $+$ measure stochastically dominates the $-$ measure, there exists a coupling $\P$ on $(\Omega \times \Omega, \mathcal F \otimes \mathcal F)$ of the two measures such that for any $A \in \mathcal F$ we have $\P({A} \times \Omega)=\mu_{\beta,\rho}^+(A)$ and $\P(\Omega \times A)=\mu_{\beta,\rho}^-(A)$, and moreover $\P(\set{(x,y) \in \Omega \times \Omega \mid x \geq y})=1$. Let $i \in \Z^d$. We write an element in $\Omega \times \Omega$ as $(\sigma^+, \sigma^-)$. By translation invariance 
    \begin{equation}
        \mathbbm E(\abs{\sigma_i^+ - \sigma_i^-}) 
        = \mathbbm E(\sigma_i^+ - \sigma_i^-)
        = \mu_{\beta,\rho}^+(\sigma_0) - \mu_{\beta,\rho}^-(\sigma_0) = 0,
    \end{equation}
    where we used that by symmetry $\mu_{\beta,\rho}^-(\sigma_0) = - \mu_{\beta,\rho}^+(\sigma_0) = 0$ in the last equality.
    Thus, $\P$ almost surely $\sigma_i^+ = \sigma_i^-$ for all $i \in \Z^d$ and thus $\mu_{\beta,\rho}^+=\mu_{\beta,\rho}^-$.
\end{proof}

A second result which goes back to a result of Lebowitz~\cite{Lebowitz:BoundsCorrelationsAnalyticity1972}, who proved the analogous statement for the spin-$\frac12$ Ising model, shows that the two-point correlation function dominates all higher correlations.

\begin{proposition}[Correlations are dominated by the two-point function]\label{thm:two_point_correlations_dominate}
    Let $f, g \colon \Omega \to \R$ local and $\eta \in \Omega$. Then, there exists a constant $c \equiv c(S) > 0$ such that for all $\Lambda \Subset \Z^d$ and $\beta > 0$, $\rho \in \R$,
    \begin{equation}
    \begin{aligned}
        &\phntm\abs{\mu_{\Lambda;\beta,\rho}^\eta\left(fg\right) - \mu_{\Lambda;\beta,\rho}^\eta\left(f\right)\mu_{\Lambda;\beta,\rho}^\eta\left(g\right)}\\
        &\leq \norm{f}_\infty \norm{g}_\infty c^{\abs{\supp f}+ \abs{\supp g}} \sum_{\substack{i \in \supp f\\j \in \supp g}} 
            \left[\mu_{\Lambda;\beta,\rho}^\eta\left(\sigma_i \sigma_j\right) - \mu_{\Lambda;\beta,\rho}^\eta\left(\sigma_i\right)\mu_{\Lambda;\beta,\rho}^\eta\left(\sigma_j\right)\right].
    \end{aligned}
    \end{equation}
\end{proposition}
\begin{proof}
    To lighten the notation we write $\sprod{\cdot}{} = \mu_{\Lambda;\beta,\rho}^\eta(\cdot)$ for the rest of this proof. Let $A,B \Subset \Z^d$ be finite sets and let $K \in \set{0,\dots,2S}^A$ and $K'\in \set{0,\dots,2S}^B$.
    By \zcref{thm:non_decreasing_quantities}, the functions $n^K$ and $N^K - n^K$ are non-decreasing. Thus, by the FKG inequality, we immediately obtain the following:
    \begin{align}
        \sprod{n^K n^{K'}}{} - \sprod{n^K}{}\sprod{n^{K'}}{} &\geq 0, \\
        \sprod{n^K (N^{K'} - n^{K'})}{} - \sprod{n^K}{} \sprod{N^{K'} - n^{K'}}{} &\geq 0, \text{ and}\\
        \sprod{N^K (N^{K'} - n^{K'})}{} - \sprod{N^K}{} \sprod{N^{K'} - n^{K'}}{} &\geq 0. 
    \end{align}
    This implies
    \begin{equation}\begin{split}\label{eq:bound_nk_correlations_by_simga_i}
        0 \leq \sprod{n^K n^{K'}}{} - \sprod{n^K}{} \sprod{n^{K'}}{} \nonumber
        &\leq \sprod{N^K n^{K'}}{} - \sprod{N^K}{} \sprod{n^{K'}}{} \nonumber\\
        &\leq  \sprod{N^K N^{K'}}{} - \sprod{N^K}{} \sprod{N^{K'}}{} \nonumber\\
        &= \sum_{i \in A} \sum_{j \in B} k_i (k')_j \left(\sprod{n_i n_j}{} - \sprod{n_i}{} \sprod{n_j}{}\right) \nonumber\\
        &\leq \sum_{i \in A} \sum_{j \in B} \left(\sprod{\sigma_i \sigma_j}{} - \sprod{\sigma_i}{} \sprod{\sigma_j}{}\right),
    \end{split}\end{equation}
    where we used $k_i, (k')_j \leq 2S$.
   By \zcref{thm:decomposition_of_local_functions}, there exist coefficients $\hat f_K$,  $K \in \set{0,\dots, 2S}^{\supp f}$, and $\hat g_{K'}$, $K' \in \set{0,\dots, 2S}^{\supp g}$, such that
   \begin{equation}
       f = \sum_{K \in \set{0, \dots, 2S}^{\supp f}} \hat f_K n^K
   \end{equation}
   and 
   \begin{equation}
       g= \sum_{K' \in \set{0, \dots, 2S}^{\supp g}} \hat g_{K'} n^{K'},
   \end{equation}
   which satisfy
   \begin{equation}
       \abs{\hat f_K} \leq c_0^{\abs{\supp f}} \norm{f}_\infty
       \quad\text{and}\quad
       \abs{\hat g_{K'}} \leq c_0^{\abs{\supp g}} \norm{g}_\infty,
   \end{equation}
   for a constant $c_0 \equiv c_0(S) > 0$.
   Thus,
   \begin{equation}
       \abs{\sprod{f g}{} - \sprod{f}{}\sprod{g}{}}
       \leq \norm{f}_\infty \norm{g}_{\infty} ((2S+1) c_0)^{\abs{\supp f} + \abs{\supp g}} \sum_{i \in A} \sum_{j \in B} \left(\sprod{\sigma_i \sigma_j}{} - \sprod{\sigma_i}{} \sprod{\sigma_j}{}\right)
   \end{equation}
   concluding the proof.
\end{proof}

\section{Extensions to percolation models}
\label{sec:percolation_model}

To prove uniqueness and exponential decay in our model, it will be helpful to express it differently. In this section we develop two percolation representations which are reminiscent of the random cluster representation of the Ising model. The first one, developed in \zcref{sec:Bernoulli_representation}, can be seen as a Bernoulli percolation model in a random environment (see~\cite[Section 2]{cohen2020rarity} for a discussion of such ``reflection transformations''). The second one, which we present in \zcref{sec:dilute_random_cluster}, is a generalization of the dilute random cluster representation introduced in~\cite{Graham:RandomClusterRepresentationBlume2006} to study the Blume-Capel model, and which was used in~\cite{Gunaratnam:ExistenceTricriticalPoint2024a} to prove properties of its phase diagram. 

\subsection{Bernoulli percolation in a random environment}
\label{sec:Bernoulli_representation}

Given $\sigma \in \Omega$ and $\set{i,j} \in \mathcal E\left(\Z^d\right)$ define
\begin{equation}
    p(\sigma_i, \sigma_j) := 1 - e^{\beta W(\sigma_i - \sigma_j) - \beta W(\abs{\sigma_i} + \abs{\sigma_j})}.
\end{equation}
Since $W$ is non-decreasing on $\Z_{\geq 0}$, we have that $p(a,b) \in [0,1)$.

Let $\mathcal G$ be the $\sigma$-algebra generated by the cylinder sets of $\set{0,1}^{\mathcal E(\Z^d)}$. Further, set $\bar \Omega := \Omega \times \set{0,1}^{\mathcal E(\Z^d)}$. We define a probability measure on $(\bar \Omega, {\mathcal F} \otimes\mathcal G)$ by
\begin{equation}
    \bar \mu_{\Lambda;\beta,\rho}^\eta (\sigma, \omega) := 
\mu_{\Lambda;\beta,\rho}^\eta(\sigma) \prod_{\set{i,j} \in \mathcal E^b(\Lambda)} p(\sigma_i, \sigma_j)^{\omega_{ij}} \left(1-p(\sigma_i, \sigma_j)\right)^{1-\omega_{ij}} 
\end{equation}
if $\omega_e = 1$ for all $e \in \mathcal E^b(\Lambda)^c$ and $\bar \mu_{\Lambda;\beta,\rho}^\eta (\sigma, \omega)=0$ otherwise.
We denote the expectation with respect to $\bar \mu_{\Lambda;\beta,\rho}^\eta$ by $\bar \mu_{\Lambda;\beta,\rho}^\eta(\cdot)$.
We call an edge $e \in \mathcal E(\Z^d)$ \emph{open} if $\omega_e=1$ and \emph{closed} otherwise.

\begin{remark}
    Let us briefly motivate this extension. Consider a closed edge $\set{i,j} \in \mathcal E^b(\Lambda)$. The contribution to the measure is then given by
    \begin{equation}\label{eq:contribution_closed_edge}
        e^{-\beta W(\sigma_i - \sigma_j)} (1-p(\sigma_i, \sigma_j))
        = e^{-\beta W(\abs{\sigma_i} + \abs{\sigma_j})}.
    \end{equation}
    This is invariant if one of the two spins is flipped, e.g., $\sigma_i \mapsto -\sigma_i$, and we can thus think of a closed edge as an analog of zero boundary conditions. For example, if $\Delta \Subset \Z^d$ and all edges in $\partial_{\exterior} \Delta$ are closed, then we can flip the spins inside $\Delta$ while leaving the ones in $\Delta^c$ the same, while preserving the measure. In this case the expectation of any antisymmetric local function supported in $\Delta$ would vanish. We make this argument precise in the proof of \zcref{thm:equivalence_no_percolation_uniqueness}.

    When $W(x)=x^2$, we have that $1-p(a,b) = e^{-2\beta \abs{ab}-2\beta ab}$. This equals the probability that a Brownian bridge from $a$ to $b$ with variance $\sigma^2=\frac1{2\beta}$ in the time interval $[0,1]$ crosses the level zero. In this sense, one can view the spin-flip invariance mentioned above as a domain Markov property of the model when extended to the cable graph.
\end{remark}

\begin{remark}
    We remark that $\mu_{\Lambda;\beta,\rho}^\eta$ is the $\sigma$-marginal of $\bar \mu_{\Lambda;\beta,\rho}^\eta$. In particular, if $f \colon \Omega \to \R$ is a function which does not depend on the values of $\omega$, then
    \begin{equation}
        \bar \mu_{\Lambda;\beta,\rho}^\eta(f) = \mu_{\Lambda;\beta,\rho}^\eta(f).
    \end{equation}
\end{remark}

\begin{remark}
    It is possible to extend the $\bar\mu_{\Lambda;\beta,\rho}^+$ measure to infinite volume in the following way. 
    Given a configuration $\sigma \in \Omega$ define the product measure of independent Bernoulli variables on $(\set{0,1}^{\mathcal E(\Z^d)}, \mathcal G)$ by
    \begin{equation}
        \mathbf P_{p(\sigma)} := \prod_{\set{i,j} \in \mathcal E(\Z^d)} \lambda_{ij}^\sigma,
    \end{equation}
    where $\lambda_{ij}^\sigma$ is the Bernoulli measure on $\set{0,1}$ given by
    \begin{equation}
        \lambda_{ij}^\sigma(\omega_{ij} = 1) = p(\sigma_i, \sigma_j) \qquad \text{and} \qquad \lambda_{ij}^\sigma(\omega_{ij} = 0) = 1 -p(\sigma_i, \sigma_j).
    \end{equation}
    We define a measure on $(\bar \Omega, \mathcal F \otimes \mathcal G)$ by
    \begin{equation}\label{eq:def_barmu_inf_volume}
        \bar\mu_{\beta,\rho}^+(A \times B) := \int_A \mathbf P_{p(\sigma)}(B) \d\mu_{\beta,\rho}^+(\sigma), \qquad A \in \mathcal F, ~ B \in \mathcal G.
    \end{equation}
    This definition suggests that we can view $\bar\mu_{\beta,\rho}^+$ as the annealed measure of percolation in a random environment, where the environment governs the probability of edges to be open or closed.
    In particular the definition implies that if $E \Subset \mathcal E(\Z^d)$ is finite, $\mathcal B \subset \set{0,1}^E$ and 
    \begin{equation}
        B := 
        \set{\omega \in \set{0,1}^{\E(\Z^d)} \mid \omega_E \in \mathcal B} \in \mathcal G
    \end{equation}
    is an event which depends only on finitely many edges ($\omega_E$ is the restriction of $\omega$ to $E$), then, for any $A$ which depends only on a finite number of sites,
    \begin{equation}\label{eq:limit_bar_mu_n}
        \bar\mu_{\beta,\rho}^+(A \times B) 
        = \int_A f_B(\sigma) \d\mu_{\beta,\rho}^+(\sigma) 
        = \mu_{\beta,\rho}^+(\I_A f_B)
        = \lim_{N \to \infty} \mu_{B_N;\beta,\rho}^+(\I_A f_B)
        = \lim_{N \to \infty} \bar\mu_{B_N;\beta,\rho}^+(A \times B),
    \end{equation}
    where we defined
    \begin{equation}
        f_B(\sigma) 
        := \mathbf P_{p(\sigma)}(B)
        =\sum_{\omega_E \in \mathcal B} \prod_{\set{i,j} \in E} p(\sigma_i, \sigma_j)^{\omega_{ij}}(1-p(\sigma_i, \sigma_j))^{1-\omega_{ij}},
    \end{equation}
    which is a local function (in $\sigma$).
\end{remark}

\begin{definition}\label{def:connections}
    For $x, y \in \Z^d$ we write 
    \begin{equation}\begin{split}
            \phantom{{}={}}
            \set{x \overset{+}{\longleftrightarrow} y}
            := \{&\sigma \in \Omega \mid \exists k \in \N, \, z_0, \dots, z_k \in \Z^d \colon \\
            &z_0 = x, \,z_k=y, \,\set{z_i, z_{i+1}} \in \E(\Z^d),  \,\sigma_{z_i} \geq 1 \text{ for all } i \}
    \end{split}\end{equation}
    for the configurations in which $x$ is connected to $y$ by a path of spins $\geq 1$. The event $\set{x \overset{-}{\longleftrightarrow} y}$ is defined analogously, replacing the condition $\sigma_{z_i} \geq 1$ by $\sigma_{z_i} \leq -1$.
    Moreover, define
    \begin{equation}\begin{split}
            \phantom{{}={}}
            \set{x \overset{\text{open}}{\longleftrightarrow} y}
            := \{& \omega \in \set{0,1}^{\mathcal E(\Z^d)} \mid \exists k \in \N, \, z_0, \dots z_k \in \Z^d \colon \\
            &z_0 = x, \,z_k=y, \,\set{z_i, z_{i+1}} \in \E(\Z^d),  \,\omega_{z_i z_{i+1}} = 1 \text{ for all } i\}
    \end{split}\end{equation}
    for the edge configurations in which $x$ is connected to $y$ by a path of open edges. 
    For $A \subset \Z^d$ we define
    \begin{equation}
        \set{x \overset{+}{\longleftrightarrow} A} := \bigcup_{y \in A} \set{ x \overset{+}{\longleftrightarrow} y},
    \end{equation}
    and
    \begin{equation}
        \set{x \overset{\text{open}}{\longleftrightarrow} A} := \bigcup_{y \in A} \set{ x \overset{\text{open}}{\longleftrightarrow} y}.
    \end{equation}
    Given $\omega \in \set{0,1}^{\mathcal{E}(\Z^d)}$ let $\mathscr C(\omega)$ be the collection of the maximal connected components of $\Z^d$ after removing the closed edges of $\omega$. We call the elements of $\mathscr C(\omega)$ the \emph{open clusters of $\omega$}. For $x \in \Z^d$ let $C_x(\omega) \in \mathscr C(\omega)$ be the unique cluster which contains $x$ and for $\Delta \subset \Z^d$ let $\mathscr C_{\Delta}(\omega) := \set{C \in \mathscr C(\omega) \mid C \cap \Delta \neq \emptyset}$ be the clusters of $\omega$ intersecting $\Delta$.
    
\end{definition}

With those definitions at hand, we can draw a connection between questions of the spin system to a type of question common in percolation.

\begin{theorem}\label{thm:equivalence_no_percolation_uniqueness}
    Let $\beta > 0$ and $\rho \in \R$. The following two statements are equivalent
    \begin{enumerate}
        \item There is a unique infinite volume Gibbs measure on $(\Omega, \mathcal F)$ for $\mathscr H_{\beta,\rho}$, i.e., $\mathscr G(\beta,\rho) =\set{\mu_{\beta,\rho}}$.
        \item In the $\bar \mu_{\beta,\rho}^+$ measure open edges do not percolate, i.e.,
        \begin{equation}
            \bar \mu_{\beta,\rho}^+\left(\bigcap_{N \in \N} \set{0 \overset{\text{open}}{\longleftrightarrow} \partial B_N}\right)
            =:\bar \mu_{\beta,\rho}^+\left(0 \overset{\text{open}}{\longleftrightarrow} \infty\right) 
            = 0.
        \end{equation}
    \end{enumerate}
\end{theorem}
\begin{proof}
    Let $\Lambda \Subset \Z^d$ and $N \in \N$. 
    Define the event $\mathcal A_N := \set{0 \overset{\text{open}}{\longleftrightarrow} \partial B_N} \in \mathcal G$. Then,
    \begin{equation}
        \mu_{\Lambda;\beta,\rho}^+\left(\sigma_0\right)
        = 
        \bar\mu_{\Lambda;\beta,\rho}^+\left(\sigma_0 \I_{\mathcal A_{N+1}}\right) + \bar\mu_{\Lambda;\beta,\rho}^+\left(\sigma_0 \I_{\mathcal A_{N+1}^c}\right).
    \end{equation}
    Before proving the equivalence claimed in the statement, we will show that if $B_{N} \subset \Lambda$, then
    \begin{equation}\label{eq:expectation_origin_in_closed_circuit_vanishes}
        \bar\mu_{\Lambda;\beta,\rho}^+\left(\sigma_0 \I_{\mathcal A_{N+1}^c}\right) = 0.
    \end{equation}
    Let $C_0(\omega) \in \mathscr C(\omega)$ be the open cluster which contains the origin. On the event $ \mathcal A_{N+1}^c$ we must have $C_0 \subset B_{N}$.
    Thus,
    \begin{equation}
        \bar\mu_{\Lambda;\beta,\rho}^+\left(\sigma_0 \I_{\mathcal A_{N+1}^c}\right)
        = \sum_{G} \bar\mu_{\Lambda;\beta,\rho}^+\left(\sigma_0 \I_{C_0 = G}\right),
    \end{equation}
    where the sum is over all connected subgraphs $G \subset B_{N}$ containing the origin. For a fixed $G=(V,E) \subset B_{N}$ with $0 \in V$ we have
    \begin{equation}\label{eq:spin_flip_mu}
        \begin{aligned}
            &\phntm Z_{\Lambda;\beta,\rho}^+ \bar \mu_{\Lambda;\beta,\rho}^+\left(\sigma_0 \I_{C_0 = G}\right) \\
            &= \sum_{\sigma \in \Omega_\Lambda^+} \sigma_0 \prod_{\substack{\set{i,j} \in \mathcal E^b(\Lambda) \\ i,j \notin G}} e^{-\beta W(\sigma_i - \sigma_j)} \prod_{\set{i,j} \in E} e^{-\beta W(\sigma_i - \sigma_j)} p(\sigma_i, \sigma_j) \prod_{\set{i,j} \in \partial_{\exterior} G} e^{-\beta W(\abs{\sigma_i} + \abs{\sigma_j})} \prod_{i \in \Lambda} e^{-\rho \sigma_i^2}  \\
            &=-Z_{\Lambda;\beta,\rho}^+ \bar \mu_{\Lambda;\beta,\rho}^+\left(\sigma_0 \I_{C_0 = G}\right),
        \end{aligned}
    \end{equation}
    where we used \eqref{eq:contribution_closed_edge}. The second equality holds after performing the spin flip
    \begin{equation}
        \sigma_i \mapsto \begin{cases}
            -\sigma_i & \text{if } i \in G,\\
            \sigma_i & \text{else},
        \end{cases}
    \end{equation}
    using that $W$ is even. This shows \eqref{eq:expectation_origin_in_closed_circuit_vanishes}.
    
    Let us first show that uniqueness implies the no percolation statement ($1. \Rightarrow 2.$). 
    As before, we decompose
    \begin{equation}
        \mu_{B_N;\beta,\rho}^+\left(\sigma_0\right)
        = \bar\mu_{B_N;\beta,\rho}^+\left(\sigma_0 \I_{\mathcal A_{N+1}}\right) + \bar\mu_{B_N;\beta,\rho}^+\left(\sigma_0 \I_{\mathcal A_{N+1}^c}\right) \\
        = \bar\mu_{B_N;\beta,\rho}^+\left(\sigma_0 \I_{\mathcal A_{N+1}}\right),
    \end{equation}
    where we used \eqref{eq:expectation_origin_in_closed_circuit_vanishes}.  
    Since $p(\sigma_i, \sigma_j)=0$ if $\sigma_i\sigma_j \leq 0$, the sign of the spins on $C_0$ is constant. As we are considering $+$ boundary conditions, it follows
    \begin{equation}
        \bar\mu_{B_N;\beta,\rho}^+\left(\sigma_0\right)
        \geq \bar\mu_{B_N;\beta,\rho}^+\left(\sigma_0 \geq 1 \text{ and } 0 \overset{\text{open}}{\longleftrightarrow} \partial B_{N+1}\right)
        = \mu_{B_N;\beta,\rho}^+\left(f_{N+1}\right),
    \end{equation}
    where we defined $f_m(\sigma) := \I_{\sigma_0 \geq 1} \mathbf P_{p(\sigma)}\left(0 \overset{\text{open}}{\longleftrightarrow} \partial B_{m}\right)$, $m \in \N$. The function $f_m$ depends only on the spins in $B_m$ and is non-decreasing in $\sigma$, because increasing $\sigma$ can only grow connected clusters of spins $\geq 1$ and because $p(a,b)$ is increasing in $a$ and $b$ if both $a,b \geq 1$. By the FKG inequality we therefore obtain
    \begin{equation}
        \mu_{B_N;\beta,\rho}^+\left(\sigma_0\right)
        \geq \mu_{B_N;\beta,\rho}^+(f_{N+1})
        \geq \mu_{\beta,\rho}^+(f_{N+1})
        = \bar\mu_{\beta,\rho}^+\left(\sigma_0 \geq 1 \text{ and } 0 \overset{\text{open}}{\longleftrightarrow} \partial B_{N+1}\right),
    \end{equation} 
    where we used \eqref{eq:def_barmu_inf_volume}.
    By assumption $\mathscr G(\beta,\rho) = \set{\mu_{\beta,\rho}}$ and therefore $\mu_{\beta,\rho}^+ = \mu_{\beta,\rho}^-$. Since $\mathbf P_{p(\sigma)} = \mathbf P_{p(-\sigma)}$ by choice of $p(\cdot, \cdot)$ and symmetry of the potential $W$, a global spin flip implies
    \begin{equation}
        \bar\mu_{\beta,\rho}^+\left(\sigma_0 \geq 1 \text{ and } 0 \overset{\text{open}}{\longleftrightarrow} \partial B_{N+1}\right)
        = \bar\mu_{\beta,\rho}^+\left(\sigma_0 \leq -1 \text{ and } 0 \overset{\text{open}}{\longleftrightarrow} \partial B_{N+1}\right).
    \end{equation}
    Further, note that the event $\set{\sigma_0 = 0 \text{ and }0 \overset{\text{open}}{\longleftrightarrow} \partial B_{N+1}}$ is a $\bar\mu_{\beta,\rho}^+$-null set, because if $\sigma_0=0$ then all edges adjacent to the origin are closed with probability one. Therefore,
    \begin{equation}
        2\mu_{B_N;\beta,\rho}^+\left(\sigma_0\right)
        \geq \bar\mu_{\beta,\rho}^+\left(0 \overset{\text{open}}{\longleftrightarrow} \partial B_{N+1}\right).
    \end{equation}
    We take the limit $N \to \infty$ on both sides. By \zcref{thm:unique_gibbs_measure_if_magnetization_zero} we have that 
    $\mu_{\beta,\rho}^+\left(\sigma_0\right) = 0$
    so that
    \begin{equation}
        \bar\mu_{\beta,\rho}^+\left(0 \overset{\text{open}}{\longleftrightarrow} \infty\right)
        = \lim_{N \to \infty} \bar\mu_{\beta,\rho}^+\left(0 \overset{\text{open}}{\longleftrightarrow} \partial B_{N+1}\right)
        = 0
    \end{equation}
    as claimed.

    To show that the no percolation statement implies uniqueness of the infinite volume Gibbs measure ($2. \Rightarrow 1.$), we will show that 
    $\mu_{\beta,\rho}^+(\sigma_0) = 0$,
    which is sufficient by \zcref{thm:unique_gibbs_measure_if_magnetization_zero}.
    Let $\Lambda \subset \Z^d$ and $N \in \N$ such that $B_{N} \subset \Lambda$. Then
    \begin{equation}
        \abs{\mu_{\Lambda;\beta,\rho}^+(\sigma_0)}
        \leq S \bar\mu_{\Lambda;\beta,\rho}^+\left(\mathcal A_N\right) + \bar\mu_{\Lambda;\beta,\rho}^+\left(\sigma_0 \I_{\mathcal A_N^c}\right)
        = S \bar\mu_{\Lambda;\beta,\rho}^+\left(\mathcal A_N\right)
        = S \mu_{\Lambda;\beta,\rho}^+\left(\mathbf P_{p(\sigma)}(\mathcal A_N)\right)
    \end{equation}
    by \eqref{eq:expectation_origin_in_closed_circuit_vanishes}. Since $\mathbf P_{p(\sigma)}(\mathcal A_N)$ is a local function of $\sigma$ (it depends only on the spins inside $B_{N}$), we can take the limit $\Lambda \uparrow \Z^d$ on both sides of the inequality and obtain
    \begin{equation}
        \abs{\mu_{\beta,\rho}^+(\sigma_0)}
        \leq S \mu_{\beta,\rho}^+\left(\mathbf P_{p(\sigma)}(\mathcal A_N)\right)
        = S \bar\mu_{\beta,\rho}^+\left(0 \overset{\text{open}}{\longleftrightarrow} \partial B_N\right)
    \end{equation}
    By assumption, the right hand side converges to zero as $N \to \infty$.
\end{proof}

\begin{lemma}\label{thm:two_point_correlation_relation_percolation}
    Let $\beta > 0$, $\rho \in \R$ and $\eta \in \Omega$. 
    Moreover, let $N \in \N$ and $\Lambda \Subset \Z^d$ such that $B_N \subset \Lambda$. For all $x \in \Z^d$ such that $\abs{x}_\infty \geq N$ it holds that
    \begin{equation}        
        \abs{\mu_{\Lambda;\beta,\rho}^\eta(\sigma_0 \sigma_x)}
        \leq S^2 \bar\mu_{\Lambda;\beta,\rho}^\eta\left(0 \overset{\text{open}}{\longleftrightarrow} \partial B_N\right).
    \end{equation}
\end{lemma}
\begin{proof}
    Since the proof is very similar to the one of $\zcref{thm:equivalence_no_percolation_uniqueness}$ we only give a sketch.
    
    Let $\Lambda \Subset \Z^d$ such that $B_N \subset \Lambda$. As before, define the event $\mathcal A_N := \set{0 \overset{\text{open}}{\longleftrightarrow} \partial B_N} \in \mathcal G$. Then,
    \begin{equation}
        \mu_{\Lambda;\beta,\rho}^\eta(\sigma_0 \sigma_x)
        = 
        \bar\mu_{\Lambda;\beta,\rho}^\eta\left(\sigma_0 \sigma_x \I_{\mathcal A_N}\right) + \bar\mu_{\Lambda;\beta,\rho}^\eta\left(\sigma_0 \sigma_x \I_{\mathcal A_N^c}\right).
    \end{equation}
    The first term is bounded by
    \begin{equation}\label{eq:two_spins_AN}
        \bar\mu_{\Lambda;\beta,\rho}^\eta\left(\sigma_0 \sigma_x \I_{\mathcal A_N}\right)
        \leq S^2  \bar\mu_{\Lambda;\beta,\rho}^\eta\left(0 \overset{\text{open}}{\longleftrightarrow} \partial B_N\right)
    \end{equation}
    whereas the second term vanishes, because we can again write
    \begin{equation}
        \bar\mu_{\Lambda;\beta,\rho}^\eta\left(\sigma_0 \sigma_x \I_{\mathcal A_N^c}\right)
        = \sum_{G} \bar\mu_{\Lambda;\beta,\rho}^\eta\left(\sigma_0 \sigma_x \I_{C_0=G}\right),
    \end{equation}
    summing over all subgraphs $G \subset B_N$ containing the origin. Since we assume $\abs{x} \geq N$, we have that $x \notin G$. Therefore flipping all the spins in $G$ leaves the measure and $\sigma_x$ invariant, while $\sigma_0 \mapsto - \sigma_0$. Hence
    \begin{equation}
        \bar\mu_{\Lambda;\beta,\rho}^\eta\left(\sigma_0 \sigma_x \I_{C_0=G}\right)
        = -\bar\mu_{\Lambda;\beta,\rho}^\eta\left(\sigma_0 \sigma_x \I_{C_0=G}\right),
    \end{equation}
    and therefore
    \begin{equation}\label{eq:two_spins_ANc_vanish}
        \bar\mu_{\Lambda;\beta,\rho}^\eta\left(\sigma_0 \sigma_x \I_{\mathcal A_N^c}\right) = 0.
    \end{equation}
    Combining \eqref{eq:two_spins_AN} and \eqref{eq:two_spins_ANc_vanish} we obtain
    \begin{equation}
        \mu_{\Lambda;\beta,\rho}^\eta(\sigma_0 \sigma_x) 
        \leq S^2  \bar\mu_{\Lambda;\beta,\rho}^\eta\left(0 \overset{\text{open}}{\longleftrightarrow} \partial B_N\right)
    \end{equation}
    concluding the proof.
\end{proof}

Since edges between spins of different sign are $\overline\mu_{\beta,\rho}^+$ almost surely closed (by our choice that $p(a,b)=0$ if $ab \leq 0$), the probability of percolating using open edges is smaller than percolating using spins $\geq 1$. This remark will be useful in dimension two and in the general case at low temperatures, where we show that spins $\geq 1$ do not percolate in $\mu_{\beta,\rho}^+$, from which we can then deduce uniqueness.

\begin{lemma}\label{thm:bound_clustersize_open_by_plus}
    Let $\beta > 0$, $\rho \in \R$ and $\Lambda \Subset \Z^d$. Then,
    \begin{equation}
        \bar\mu_{\Lambda;\beta,\rho}^+ \left(0 \overset{\text{open}}{\longleftrightarrow} \partial B_N\right) 
        \leq 2\mu_{\Lambda;\beta,\rho}^+ \left(0 \overset{+}{\longleftrightarrow} \partial B_N\right).
    \end{equation}
\end{lemma}
\begin{proof}
    We can decompose
    \begin{equation}
    \begin{aligned}
        &\phntm \bar\mu_{\Lambda;\beta,\rho}^+ \left(0 \overset{\text{open}}{\longleftrightarrow} \partial B_N\right) \\
        &= \bar\mu_{\Lambda;\beta,\rho}^+ \left(0 \overset{+}{\longleftrightarrow} \partial B_N \text{ and }0 \overset{\text{open}}{\longleftrightarrow} \partial B_N\right) 
        + \bar\mu_{\Lambda;\beta,\rho}^+ \left(0 \overset{-}{\longleftrightarrow} \partial B_N \text{ and }0 \overset{\text{open}}{\longleftrightarrow} \partial B_N\right) .
    \end{aligned}
    \end{equation}
    Indeed, let $(\sigma, \omega) \in \set{0 \overset{\text{open}}{\longleftrightarrow} \partial B_N}$. Then we find $z_0 \cdots z_k$, $z_0=0$, $z_k \in \partial B_N$ an open path from $0$ to $\partial B_N$. If there exists $j \in \set{1, \dots, k}$ such that $\sigma_{z_j}\sigma_{z_{j-1}}\leq0$, then $p(\sigma_{z_j},\sigma_{z_{j-1}})=0$ and therefore $\bar\mu_{\Lambda;\beta,\rho}^+(\sigma,\omega)=0$. On the other hand, if $\bar\mu_{\Lambda;\beta,\rho}^+(\sigma,\omega)>0$, it follows that 
    \begin{equation}
        \sigma \in \set{0 \overset{+}{\longleftrightarrow} \partial B_N} \,\dot{\cup}\, \set{0 \overset{-}{\longleftrightarrow} \partial B_N}.
    \end{equation}
    Thus,
    \begin{equation} 
    \begin{aligned}
        \bar\mu_{\Lambda;\beta,\rho}^+ \left(0 \overset{\text{open}}{\longleftrightarrow} \partial B_N\right)
        &\leq \mu_{\Lambda;\beta,\rho}^+ \left(0 \overset{+}{\longleftrightarrow} \partial B_N\right) 
        + \mu_{\Lambda;\beta,\rho}^+ \left(0 \overset{-}{\longleftrightarrow} \partial B_N\right)\\
        &= \mu_{\Lambda;\beta,\rho}^+ \left(0 \overset{+}{\longleftrightarrow} \partial B_N\right) 
        + \mu_{\Lambda;\beta,\rho}^- \left(0 \overset{+}{\longleftrightarrow} \partial B_N\right)\\
        &\leq 2 \mu_{\Lambda;\beta,\rho}^+ \left(0 \overset{+}{\longleftrightarrow} \partial B_N\right) 
    \end{aligned}       
    \end{equation}
    where the equality in the second line is by a global spin flip and the second inequality is by FKG.
\end{proof}

\subsection{Uniqueness at high temperatures}
\label{sec:high_temperatures}

We give two different proofs of uniqueness and exponential decay at high temperatures (\zcref{thm:main_high_temperature}). The first showcases the percolation representation developed in \zcref{sec:Bernoulli_representation}, while the second is a standard application of Dobrushin's uniqueness criterion. We give both proofs, since the first is more enlightening with respect to our techniques, while the second produces better constants (of the same order as in the first proof).

\begin{proof}[Proof of the high temperature case \zcref{thm:main_high_temperature} by comparison with Bernoulli edge percolation]
    For this proof we restrict to a smaller range of $\beta$, which satisfy $\beta W(2S) < - \ln(1- p_c(d))$ for the critical probability $p_c(d)$ of Bernoulli bond percolation of $\Z^d$.
    
    By \zcref{thm:equivalence_no_percolation_uniqueness,thm:two_point_correlations_dominate,thm:two_point_correlation_relation_percolation} it suffices to show that 
    \begin{equation}
        \bar\mu_{\beta,\rho}^+\left(0 \overset{\text{open}}{\longleftrightarrow} \partial B_N\right) \leq c_1 e^{-c_2 N}.
    \end{equation}
    holds for positive constants independent of $N$. Let 
    \begin{equation}
        P_p := \prod_{\set{i,j} \in \mathcal E(\Z^d)} \lambda_{ij}
    \end{equation}
    be the product measure of independent Bernoulli random variables on $\mathcal E(\Z^d)$ with parameter $p$, where $\lambda_{ij}$ is a Bernoulli measure given by
    \begin{equation}
        \lambda_{ij}(\omega_{ij} = 1) = p \qquad \text{and} \qquad \lambda_{ij}(\omega_{ij} = 0) = 1 - p.
    \end{equation}
    It is well known (cf.~\cite{Grimmett:Percolation1999}) that for all $p < p_c(d)$ 
    \begin{equation}
        P_p\left(0 \overset{\text{open}}{\longleftrightarrow} \partial B_N\right) \leq c_1 e^{-c_2 N}.
    \end{equation}
    Moreover, the value of ${P_p\left(0 \overset{\text{open}}{\longleftrightarrow} \partial B_N\right)}$ is increasing in $p$. Therefore, if $p(\sigma_i, \sigma_j) \leq p_0 < p_c$, then
    \begin{equation}
        \bar\mu_{\beta,\rho}^+\left(0 \overset{\text{open}}{\longleftrightarrow} \partial B_N\right)
        = \mu_{\beta,\rho}^+\left(\mathbf P_{p(\sigma)}\left(0 \overset{\text{open}}{\longleftrightarrow} \partial B_N\right)\right)
        \leq 
        P_{p_0}\left(0 \overset{\text{open}}{\longleftrightarrow} \partial B_N\right) 
        \leq  c_1 e^{-c_2 N}.
    \end{equation}
    By convexity of $W$ we have that
    \begin{equation}
        p(\sigma_i, \sigma_j) \leq 1- e^{-\beta W(2S)},
    \end{equation}
    so that $p(\sigma_i, \sigma_j) < p_c$ whenever $\beta W(2S) < - \ln(1- p_c(d))$. 
\end{proof}

\begin{proof}[Proof of the high temperature case \zcref{thm:main_high_temperature} with Dobrushin uniqueness]
    The high temperature case has a quick proof using Dobrushin's uniqueness criterion.

    Let
    \begin{equation}
        \d_{ij} 
        := \frac 12 \sup \set{\norm{\mu^\eta_{\set{j};\beta,\rho} - \mu^{\bar\eta}_{\set{j};\beta,\rho}}_{{\TV}} \mid \eta, \bar\eta \in \Omega \colon \eta_k =\bar\eta_k \text{ for all } k \neq i},
    \end{equation}
    where
    \begin{equation}
        \norm{\mu^\eta_{\set{j};\beta,\rho} - \mu^{\bar\eta}_{\set{j};\beta,\rho}}_{{\TV}}
        := 2\sup_{A \subset \set{-S, \dots, S}}\abs{\mu^\eta_{\set{j};\beta,\rho}(A) - \mu^{\bar\eta}_{\set{j};\beta,\rho}(A)}
    \end{equation}
    is the total variation distance.
    Dobrushin's uniqueness theorem~\cite[Theorem V.1.3]{Simon:StatisticalMechanicsLattice1993} states that if
    \begin{equation}
        \alpha := \sup_{j \in \Z^d} \sum_{i \in \Z^d} \d_{ij} < 1,
    \end{equation}
    then $\abs{\mathscr G(d,\beta,\rho)} = 1$. As observed in~\cite{Simon:RemarkOnDobrushin1979}, see also~\cite[Theorem V.1.1]{Simon:StatisticalMechanicsLattice1993}, we have
    \begin{equation}
        \d_{ij} \leq \frac\beta2 \sup_{t,t' \in \set{-S, \dots, S}} \sup_{s \in \set{-S, \dots, S}} \abs{W(s - t) - W(s - t')} \leq \frac\beta2 W(2S)
    \end{equation}
    by convexity of $W$. Since we consider only nearest neighbor interactions, we thus obtain
    \begin{equation}
        \alpha \leq d \beta W(2S).
    \end{equation}
    A result by Föllmer~\cite[Theorem V.2.4]{Simon:StatisticalMechanicsLattice1993} implies the exponential decay of correlations of the unique Gibbs measure for all $\alpha < 1$.
\end{proof}

\begin{remark}
    Note that in high dimensions $p_c(d) \sim \frac1{2d}$. Therefore the bound on $\beta$ we obtain by the comparison with Bernoulli percolation is worse by a factor of two.
\end{remark}

\subsection{A generalized dilute random cluster model}
\label{sec:dilute_random_cluster}

Let us now introduce the generalized dilute random cluster measure which extends the works~\cite{Graham:RandomClusterRepresentationBlume2006,Gunaratnam:ExistenceTricriticalPoint2024a} on the Blume-Capel model. We will make use of this measure in \zcref{sec:three_dimensions} to prove exponential decay of correlations for $\mu_{\beta,\rho}$ in $d\geq3$ for a class of potentials $W$.

\begin{definition}
    Let $p, \lambda \colon \N \times \N \to [0,\infty)$ by symmetry in the sense that $p(a,b)=p(b,a)$ (and the same for $\lambda$) and assume that $p \in [0,1)$. 
    We define the generalized dilute random cluster measure $\pi^+_{\Lambda;p, \lambda}$ on $\set{0, \dots, S}^{\Z^d} \times \set{0,1}^{\E(\Z^d)}$ as follows:
    \begin{equation}\label{eq:def_DRC}
        \pi^+_{\Lambda;p, \lambda}(\varphi, \omega) :=
            \frac{2^{K(\omega)-\abs{\set{\varphi=0}}}}{Z_{\Lambda;p, \lambda}} \prod_{\set{i,j} \in \E^b(\Lambda)} \left(\frac{p(\varphi_i, \varphi_j)}{1-p(\varphi_i, \varphi_j)}\right)^{\omega_{ij}}
        \lambda(\varphi_i, \varphi_j)
    \end{equation}
    if $\varphi_i=S$ for all $i \notin \Lambda$ and $\omega_e=1$ for all $e \notin \E^b(\Lambda)$ and $\pi^+_{\Lambda;p, \lambda}(\varphi, \omega) :=0$ otherwise.
    Here, $K(\omega)=\abs{\mathscr C(\omega)}$ is the number of open clusters of $\omega$.
\end{definition}

Let us first draw a connection to the measure $\bar\mu_{\Lambda;\beta,\rho}$ defined in the previous section.
\begin{theorem}\label{thm:RC_as_spin_model}
    Let $\beta > 0$, $\rho \in \R$ and $W \colon \Z \to [0, \infty)$ be even and increasing on $\Z_{\geq 0}$. Define for $a,b \in \set{0,\dots,S}$
    \begin{equation}
        p(a,b) := 1 - e^{-\beta (W(a+ b)-W(a-b))} \in [0,1)
        \qquad\text{and}\qquad
        \lambda_\rho(a,b) := e^{-\rho(a^2 + b^2)-\beta W(a+b)}.
    \end{equation}
    If $\Lambda \Subset \Z$ is finite and $\Lambda^c$ is connected, then
    \begin{equation}\label{eq:equality_pi_barmu}
        \pi^+_{\Lambda;p,\lambda_\rho}(\varphi, \omega)
        = \overline\mu_{\Lambda;\beta,2d\rho}^+\left(\set{\sigma \in \Omega \mid \abs{\sigma}=\varphi} \times \set{\omega}\right)
    \end{equation}
    for all $(\varphi, \omega) \in \set{0, \dots, S}^{\Z^d} \times \set{0,1}^{\E(\Z^d)}$.
\end{theorem}
\begin{proof}
    We only need to consider $\omega \in \set{0,1}^{\E(\Z^d)}$ such that $\omega_e=1$ for all $e \notin \E^b(\Lambda)$.

    First, consider $(\varphi, \omega) \in \set{0, \dots, S}^{\Z^d} \times \set{0,1}^{\E(\Z^d)}$ such that there exist $u,v \in \Z^d$ such that $\set{u,v} \in \E(\Z^d)$ and $\varphi_u \varphi_v = 0$ and $\omega_{uv}=1$. In this case $p(\varphi_u, \varphi_v)=0$ and therefore
    \begin{equation}
        \pi^+_{\Lambda;p,\lambda_\rho}(\varphi, \omega) = 0.
    \end{equation}
    Similarly, if $\sigma \in \Omega$ with $\abs{\sigma}=\varphi$, then also
    \begin{equation}
        \overline\mu_{\Lambda;\beta,2d\rho}^+\left(\sigma, \omega\right) = 0.
    \end{equation}
    Thus, we assume for the rest of the proof that $\omega_{uv}=0$ if $\varphi_u\varphi_v=0$.
    We have,
    \begin{equation}
    \begin{aligned}
        &\phntm\sum_{\substack{\sigma \in \Omega_\Lambda^+,\\\abs{\sigma}=\varphi}} \prod_{i \in \Lambda} e^{-2d\rho \sigma_i^2} \prod_{\set{i,j} \in \E^b(\Lambda)} e^{-\beta W(\sigma_i-\sigma_j)} p(\sigma_i, \sigma_j)^{\omega_{ij}} (1-p(\sigma_i, \sigma_j))^{1-\omega_{ij}} \\
        &=\prod_{i \in \Lambda} e^{-2d\rho{\varphi_i^2}}\sum_{\substack{\sigma \in \Omega_\Lambda^+,\\\abs{\sigma}=\varphi}} \prod_{\set{i,j} \in \E^b(\Lambda)} e^{-\beta W(\abs{\sigma_i}+\abs{\sigma_j})} \left(\frac{p(\sigma_i, \sigma_j)}{1-p(\sigma_i, \sigma_j)}\right)^{\omega_{ij}} \\
        &=\alpha\prod_{\set{i,j} \in \E^b(\Lambda)} e^{-\rho(\varphi_i^2 + \varphi_j^2)-\beta W(\varphi_i + \varphi_j)} \sum_{\substack{\sigma \in \Omega_\Lambda^+,\\\abs{\sigma}=\varphi}} \prod_{\set{i,j} \in \E^b(\Lambda)} \left(\frac{p(\sigma_i, \sigma_j)}{1-p(\sigma_i, \sigma_j)} \I_{\sigma_i\sigma_j \geq1}\right)^{\omega_{ij}}  \\
        &=\alpha\prod_{\set{i,j} \in \E^b(\Lambda)} e^{-\rho(\varphi_i^2 + \varphi_j^2)-\beta W(\varphi_i+ \varphi_j)} \left(\frac{p(\varphi_i, \varphi_j)}{1-p(\varphi_i, \varphi_j)}\right)^{\omega_{ij}} \sum_{\substack{\sigma \in \Omega_\Lambda^+,\\\abs{\sigma}=\varphi}} \prod_{\omega_{ij} = 1} \I_{\sigma_i\sigma_j \geq1},
        \end{aligned}
    \end{equation}
    where we used that if $\sigma_i\sigma_j \geq 1$, then 
    \begin{equation}
        \frac{p(\sigma_i, \sigma_j)}{1-p(\sigma_i, \sigma_j)} 
        = \frac{p(\abs{\sigma_i}, \abs{\sigma_j})}{1-p(\abs{\sigma_i}, \abs{\sigma_j})}
    \end{equation}
    and where we defined the constant
    \begin{equation}
        \alpha = \prod_{j \in \Lambda^c} \prod_{i \sim j, i \in \Lambda} e^{\rho S^2}.
    \end{equation}
    Further, using that $\omega_{ij}=0$ whenever $\varphi_i\varphi_j=0$,
    \begin{equation}
        \sum_{\substack{\sigma \in \Omega_\Lambda^+,\\\abs{\sigma}=\varphi}}  \prod_{\omega_{ij}=1} \I_{\sigma_i \sigma_j \geq 1}
        = 2^{-\abs{\varphi = 0}}\sum_{{a \in \set{\pm}^{\Z^d} \colon ~a_i=1 \text{ if } i \in \Lambda^c}} \prod_{\omega_{ij}=1} \I_{a_i=a_j}
        = 2^{K(\omega)-\abs{\varphi = 0}-1},
    \end{equation}
    where a factor of one-half is due to the choice of $+$ boundary conditions, because there exists exactly one open cluster which intersects $\Lambda^c$ (because that set is connected), and on this cluster the signs are forced to be equal to one. 
    Thus,
    \begin{equation}\begin{aligned}
        &\phntm \frac2\alpha\sum_{\substack{\sigma \in \Omega_\Lambda^+,\\\abs{\sigma}=\varphi}} \prod_{i \in \Lambda} e^{-2d\rho \sigma_i^2} \prod_{\set{i,j} \in \E^b(\Lambda)}e^{-\beta W(\sigma_i - \sigma_j)} p(\sigma_i, \sigma_j)^{\omega_{ij}} (1-p(\sigma_i, \sigma_j))^{1-\omega_{ij}} \\
       &= 2^{K(\omega)-\abs{\set{\varphi=0}}}\prod_{\set{i,j} \in \E^b(\Lambda)} \left(\frac{p(\varphi_i, \varphi_j)}{1-p(\varphi_i, \varphi_j)}\right)^{\omega_{ij}} \lambda_\rho(\varphi_i, \varphi_j).
        \end{aligned}
    \end{equation}
    Summing over $\varphi$ and $\omega$ yields
    \begin{equation}
        Z_{\Lambda;p, \lambda_\rho} = \frac2\alpha Z^+_{\Lambda;\beta, 2d\rho},
    \end{equation}
    where $Z^+_{\Lambda;\beta, 2d\rho}$ is the partition function of $\mu^+_{\Lambda;\beta,2d\rho}$. This implies the claim \eqref{eq:equality_pi_barmu}.
\end{proof}

It will be useful to decompose expectations with respect to $\pi^+_{\Lambda;p,\lambda}$ into two successive expectations.
\begin{definition}\label{def:Psi_Phi_successive_expectations}
    Let $E \subset \E^b(\Lambda)$ and $\xi \in \set{0,1}^E$. For $\varphi \in \set{0,\dots,S}^{\Z^d}$ define the probability measure
    \begin{equation}
        \Psi_{\Lambda;p,\lambda}^\xi(\varphi) 
        := \begin{cases}
            \left(Z^\xi_{\Lambda;p, \lambda}\right)^{-1} 2^{-\abs{\set{\varphi=0}}}\prod_{\set{i,j} \in \E^b(\Lambda) \setminus E} \lambda(\varphi_i, \varphi_j) Z^{\xi}_{\Lambda;p(\varphi)} &\text{if $\varphi_i=S$ for all $i \notin \Lambda$},\\
            0 & \text{else}.
        \end{cases}
    \end{equation}
    Here ${\set{\varphi=0}} := \set{i \in \Z^d \mid \varphi_i = 0}$ is the zero set of $\varphi$ and $Z^\xi_{\Lambda;p, \lambda}$ is a normalization constant, while $Z^{\xi}_{\Lambda;p(\varphi)}$ is the partition function of the random cluster measure on 
    \begin{equation}
        \Sigma^\xi_\Lambda := \set{\omega \in \set{0,1}^{\E(\Z^d)} \mid \omega_e = 1 \text{ for all } e \in \E(\Z^d) \setminus \E^b(\Lambda) \text{ and } \omega_e = \xi_e \text{ for all $e \in E$}}
    \end{equation}
    given by
    \begin{equation}
        \phi^{\xi}_{\Lambda;p(\varphi)}(\omega) := 
        \begin{cases}
            \left(Z^{\xi}_{\Lambda;p(\varphi)}\right)^{-1} 2^{K(\omega)} \prod_{\set{i,j} \in \E^b(\Lambda) \setminus E} \left(\frac{p(\varphi_i,\varphi_j)}{1-p(\varphi_i,\varphi_j)}\right)^{\omega_{ij}} & \text{if } \omega \in \Sigma^\xi_\Lambda,\\
            0 & \text{else}.
        \end{cases}
    \end{equation}
    Note that in these definitions we allow $E\equiv \emptyset$. In that case we will write $\Psi_{\Lambda;p,\lambda}^\emptyset$ and $\phi^{\emptyset}_{\Lambda;p(\varphi)}$.
\end{definition}

\begin{lemma}\label{thm:mu_conditioned_as_iterated_measures}
    Let $X \subset \Z^d$ and $E \subset \E^b(\Lambda)$ such that if $\set{i,j} \in E$ then $i,j \in X$. Let $\eta \in \set{0, \dots, S}^X$ with $\eta_i = S$ if $i \in \Lambda^c$, and $\xi \in \set{0,1}^{E}$ such that
    \begin{equation}
        \pi^+_{\Lambda;p,\lambda}(\varphi_X = \eta \text{ and } \omega_E = \xi) > 0,
    \end{equation}
    where $\varphi_X$ is the restriction of $\varphi$ to $X$ and $\omega_E$ is defined analogously. 
    Then, for any function $f \colon \set{0, \dots, S}^{\Z^d} \times \set{0,1}^{\E(\Z^d)} \to \R$ the following holds:
    \begin{equation}
        \pi^+_{\Lambda;p,\lambda}(f \mid \varphi_X = \eta \text{ and } \omega_E = \xi)
        = \Psi_{\Lambda;p,\lambda}^\xi(\phi^{\xi}_{\Lambda;p(\varphi)}(f(\varphi, \cdot)) \mid \varphi_X=\eta).
    \end{equation}
\end{lemma}
\begin{proof}
    We have that
    \begin{equation}
    \begin{aligned}
        &\phntm\pi^+_{\Lambda;p,\lambda}(\varphi_X = \eta \text{ and } \omega_E = \xi)\pi^+_{\Lambda;p,\lambda}(f \mid \varphi_X = \eta \text{ and } \omega_E = \xi) \\
        &= \prod_{\set{i,j} \in E} \lambda(\eta_i, \eta_j) \left(\frac{p(\eta_i,\eta_j)}{1-p(\eta_i,\eta_j)}\right)^{\xi_{ij}} 
        \sum_{\varphi_X=\eta} 2^{-\abs{\set{\varphi=0}}} \prod_{\set{i,j} \in \E^b(\Lambda) \setminus E} \lambda(\varphi_i, \varphi_j) \times\\
        &\qquad\qquad\times\sum_{\omega \in \Sigma^\xi_\Lambda} 2^{K(\omega)} \prod_{\set{i,j} \in \E^b(\Lambda) \setminus E} \left(\frac{p(\varphi_i,\varphi_j)}{1-p(\varphi_i,\varphi_j)}\right)^{\omega_{ij}}f(\varphi,\omega)
    \end{aligned}
    \end{equation}
    which implies
    \begin{equation}
    \begin{aligned}
        &\phntm\pi^+_{\Lambda;p,\lambda}(f \mid \varphi_X = \eta \text{ and } \omega_E = \xi)\\
        &= \frac{\sum_{\varphi_X=\eta} 2^{-\abs{\set{\varphi=0}}} Z_{\Lambda;p(\varphi)}^\xi \phi_{\Lambda;p(\varphi)}^\xi(f(\varphi, \cdot)) \prod_{\set{i,j} \in \E^b(\Lambda) \setminus E} \lambda(\varphi_i, \varphi_j)}{\sum_{\varphi_X=\eta} 2^{-\abs{\set{\varphi=0}}} Z_{\Lambda;p(\varphi)}^\xi \prod_{\set{i,j} \in \E^b(\Lambda) \setminus E} \lambda(\varphi_i, \varphi_j)} \\
        &=\Psi_{\Lambda;p,\lambda}^\xi(\phi^{\xi}_{\Lambda;p(\varphi)}(f(\varphi, \cdot)) \mid \varphi_X=\eta),
    \end{aligned}
    \end{equation}
    where we used that
    \begin{equation}
        Z^{\xi}_{\Lambda;p(\varphi)}  = \sum_{\omega \in \Sigma_\Lambda^\xi} 2^{K(\omega)} \prod_{\set{i,j} \in \E^b(\Lambda) \setminus E} \left(\frac{p(\varphi_i,\varphi_j)}{1-p(\varphi_i,\varphi_j)}\right)^{\omega_{ij}}
    \end{equation}
    concluding the proof.
\end{proof}

This decomposition implies the following domain monotonicity for $\pi_{\Lambda;p,\lambda}^+$. Note that by \zcref{thm:RC_as_spin_model} it also gives a domain monotonicity property for $\bar\mu_{\Lambda;\beta,\rho}^+$.

\begin{lemma}\label{thm:pi_domain_monotonicity}
    Let $\Lambda' \subset \Lambda \Subset \Z^d$. Let $\lambda,p \colon \set{0,\dots,S} \times \set{0,\dots,S} \to [0,\infty)$ be symmetric and assume that $\lambda > 0$ and $p \in [0,1)$. Additionally, assume that $a \mapsto p(a,b)$ is non-decreasing for any $b \in \set{0,\dots,S}$ and that $\Psi^\emptyset_{\Lambda;p,\lambda}$ satisfies the lattice FKG condition \eqref{eq:FKG_lattice_condition}.

    Then $\pi_{\Lambda;p,\lambda}^+ \leq_{st} \pi_{\Lambda';p,\lambda}^+$, i.e., for any $f \colon \set{0, \dots, S}^{\Z^d} \times \set{0,1}^{\E(\Z^d)} \to \R$ which is coordinate-wise non-decreasing, it holds that
    \begin{equation}\label{eq:pi_domain_monotonicity}
        \pi_{\Lambda;p,\lambda}^+(f) \leq \pi_{\Lambda';p,\lambda}^+(f).
    \end{equation}
\end{lemma}
\begin{proof}
    Using \zcref{thm:mu_conditioned_as_iterated_measures} we write
    \begin{equation}
        \pi_{\Lambda';p,\lambda}^+(f)
        =\Psi^\emptyset_{\Lambda';p,\lambda}\left(\phi^\emptyset_{\Lambda';p(\varphi)}(f(\varphi,\cdot))\right) 
        =\Psi^\emptyset_{\Lambda';p,\lambda}\left(g_{\Lambda'}\right),
    \end{equation}
    where $g_{\Lambda'}(\varphi) := \phi^\emptyset_{\Lambda';p(\varphi)}(f(\varphi,\cdot))$. 
    Let $X := \Lambda \setminus \Lambda'$ and $E= \E^b(\Lambda) \setminus \E^b(\Lambda')$. A computation shows that
    \begin{equation}
            \Psi^\emptyset_{\Lambda';p,\lambda}(g_{\Lambda'}) 
            =\frac{\Psi^\emptyset_{\Lambda;p,\lambda}\left(\phi^\emptyset_{\Lambda;p(\varphi)}(\omega_E=1) \,g_{\Lambda'} \mid \varphi_X=S\right)}{\Psi^\emptyset_{\Lambda;p,\lambda}\left(\phi^\emptyset_{\Lambda;p(\varphi)}(\omega_E=1)\mid \varphi_X=S\right)}
            \geq \Psi^\emptyset_{\Lambda;p,\lambda}\left(g_{\Lambda'} \mid \varphi_X=S\right)
    \end{equation}
    by the FKG inequality for $\Psi^\emptyset_{\Lambda;p,\lambda}$.

    By the FKG property of the random cluster measure we have that
    \begin{equation}
        g_{\Lambda'}(\varphi) 
        =\phi^\emptyset_{\Lambda;p(\varphi)}(f(\varphi,\cdot) \mid \omega_E=1) 
        \geq g_\Lambda(\varphi).
    \end{equation}
    Since $p(\varphi)$ is non-decreasing in $\varphi$, the monotonicity of the random cluster measure in $p$ implies that $g_{\Lambda}$ is non-decreasing in $\varphi$. Thus,
    \begin{equation}
        \pi_{\Lambda';p,\lambda}^+(f)
        \geq \Psi^\emptyset_{\Lambda;p,\lambda}\left(g_{\Lambda} \mid \varphi_X=S\right)
        \geq \Psi^\emptyset_{\Lambda;p,\lambda}\left(g_{\Lambda}\right) = \pi_{\Lambda;p,\lambda}(f),
    \end{equation}
    by another application of the FKG inequality for $\Psi^\emptyset_{\Lambda;p,\lambda}$ and \zcref{thm:mu_conditioned_as_iterated_measures}, yielding \eqref{eq:pi_domain_monotonicity}.
\end{proof}

We have the following stochastic ordering of the measures $\Psi_{\Lambda;p,\lambda}^\xi$.

\begin{proposition}\label{thm:stochastic_ordering_dilute_measures}
    Let $X \subset \Lambda$, $u \in \Lambda \setminus X$ and $E \subset \E^b(\Lambda)$, and let $\eta_1,\eta_2 \in \set{0, \dots, S}^X$ and $\xi_1,\xi_2 \in \set{0,1}^E$ such that $(\eta_1,\xi_1)\leq(\eta_2,\xi_2)$. 
    Moreover, let $\lambda_1,\lambda_2,p_1,p_2 \colon \set{0,\dots,S} \times \set{0,\dots S} \to [0, \infty)$ be symmetric and assume that $\lambda_1,\lambda_2 > 0$ and $p_1,p_2 \in [0,1)$. Additionally, assume the following relations
    \begin{enumerate}[label=(\roman*)]
        \item We have $p_1 \leq p_2$ and $a \mapsto {p_k(a,b)}$ is non-decreasing for any $b \in \set{0, \dots, S}$ for $k \in \set{1,2}$, \label{item:p_non_increasing} 
        \item For any $a,b \in \set{0,\dots,S}$, $a \leq S-1$ it holds \label{item:lambda1/lambda2_non_increasing}
        \begin{equation}
            \lambda_1(a+1,b) \lambda_2(a,b) \leq \lambda_1(a,b) \lambda_2(a+1,b)
        \end{equation}
        \item For any $a,b \in \set{0,\dots,S}$, $a \leq S-1$, it holds that \label{item:p1_p2_relation}
        \begin{equation}
            \frac{p_1(a,b)}{1-p_1(a,b)}\frac{1-p_1(a+1,b)}{p_1(a+1,b)}
            \geq 
            \frac{p_2(a,b)}{1-p_2(a,b)}\frac{1-p_2(a+1,b)}{p_2(a+1,b)}.
        \end{equation}
        Here we use the convention that $\frac00=1$. For example, \zcref[noname]{item:p_non_increasing} implies that if $p_k(a+1,b)= 0$, then also $p_k(a,b)=0$ and we treat the corresponding term as 
            $(1-p_k(a+1,b))/({1-p_k(a,b)})=1$.
    \end{enumerate}
    Moreover, assume that $\Psi^{\xi_k}_{\Lambda;p_k,\lambda_k}$ satisfies the lattice FKG condition \eqref{eq:FKG_lattice_condition}, for either $k=1$ or $k=2$. Then, for any $s \in \set{0, \dots, S}$,
    \begin{equation}
        \Psi^{\xi_1}_{\Lambda;p_1,\lambda_1}(\cdot \mid \varphi_u \geq s, \, \varphi_X = \eta_1) \leq_{st} \Psi^{\xi_2}_{\Lambda;p_2,\lambda_2}(\cdot \mid \varphi_u \geq s, \, \varphi_X = \eta_2).
    \end{equation}
\end{proposition}
\begin{proof}
    Define two measures $\Psi_1$ and $\Psi_2$ for $\varphi \in \set{0,\dots,S}^{\Z^d \setminus X}$ satisfying $\varphi_i=S$ for $i \in \Lambda^c$ and $\phi_u \geq s$, by
    \begin{equation}
        \Psi_k(\varphi) 
        := \Psi^{\xi_k}_{\Lambda;p_k,\lambda_k}(\varphi \times \eta_k \mid \varphi_u \geq s,\, \varphi_X = \eta_k)
        =\frac{\Psi^{\xi_k}_{\Lambda;p_k,\lambda_k}(\varphi \times \eta_k)}{\sum_{\bar\varphi_u \geq s}\Psi^{\xi_k}_{\Lambda;p_k,\lambda_k}(\bar\varphi \times \eta_k)},
    \end{equation}
    where $\varphi \times \eta$ denotes the configuration on $\Z^d$ which equals $\varphi$ on $\Z^d \setminus X$ and $\eta$ on $X$. It suffices to show that
    \begin{equation}\label{eq:Psi2_Psi1_stoch_domination}
        \Psi_1 \leq_{st} \Psi_2.
    \end{equation}
    Assume that $\Psi^{\xi_1}_{\Lambda;p_1,\lambda_1}$ satisfies the FKG lattice condition, the other case is similar. In that case, also $\Psi_1$ satisfies the lattice FKG inequality.
    Let $f \colon \set{0,\dots,S}^{\Z^d} \to \R$ be an increasing function. Then,
    \begin{equation}
        \Psi_2(f) = \Psi_1(fr),
    \end{equation}
    where
    \begin{equation}
        r(\varphi) := \frac{\Psi_2(\varphi)}{\Psi_1(\varphi)},
    \end{equation}
    which is well defined because $\lambda_1 > 0$ implies that $\Psi_1 > 0$. We claim that $r$ is also an increasing function. In that case, the FKG inequality for $\Psi_1$ implies
    \begin{equation}
        \Psi_2(f) \geq \Psi_1(f) \Psi_1(r) = \Psi_1(f)
    \end{equation}
    concluding the proof of \eqref{eq:Psi2_Psi1_stoch_domination}.
    
    Fix $\varphi \in \set{0,\dots,S}^{\Z^d \setminus X}$. Let $x \in \Lambda \setminus X$ and $a \in \set{0,\dots, S-1}$ and $a' := a+1$, with the additional restriction $a \geq s$ if $x=u$. Define
    \begin{equation}\label{eq:def_phia_ab}
        \varphi_i^a := \begin{cases}
            a & \text{if } i = x,\\
            \varphi_i & \text{else}.
        \end{cases}
    \end{equation}
    To show that $r$ is indeed increasing, by induction it suffices to prove
    \begin{equation}\label{eq:Psi_stoch_dom_second_condition}
        \Psi_2(\varphi^{a'})\Psi_1(\varphi^{a}) \geq \Psi_2(\varphi^{a})\Psi_1(\varphi^{a'}).
    \end{equation}
    The rest of the proof is devoted to proving this inequality.
    By the FKG inequality for $\Psi^{\xi_1}_{\Lambda;p_1,\lambda_1}$
    \begin{equation}
        {\Psi_1(\varphi^{a'})}{\Psi^{\xi_1}_{\Lambda;p_1,\lambda_1}(\varphi^{a} \times \eta_2)} 
        \leq {\Psi_1(\varphi^{a})}\Psi^{\xi_1}_{\Lambda;p_1,\lambda_1}(\varphi^{a'} \times \eta_2).
    \end{equation}
    Multiplying both sides by $\Psi_2(\varphi^a)$, it remains to show that
    \begin{equation}
        \Psi_2(\varphi^{a}) \Psi^{\xi_1}_{\Lambda;p_1,\lambda_1}(\varphi^{a'} \times \eta_2)
        \leq \Psi_2(\varphi^{a'}) \Psi^{\xi_1}_{\Lambda;p_1,\lambda_1}(\varphi^{a} \times \eta_2).
    \end{equation}
    Inserting the definitions, we need to show that
    \begin{equation}
    \begin{aligned}
        &\phntm Z^{\xi_2}_{\Lambda,p_2(\varphi^{a}\times \eta_2)} Z^{\xi_1}_{\Lambda,p_1(\varphi^{a'}\times \eta_2)} \prod_{\set{j, x} \in \E^b(\Lambda) \setminus E} \lambda_2(a, (\varphi^{a} \times \eta_2)_j) \lambda_1(a', (\varphi^{a'} \times \eta_2)_j)  \\
        &\leq Z^{\xi_2}_{\Lambda,p_2(\varphi^{a'}\times \eta_2)} Z^{\xi_1}_{\Lambda,p_1(\varphi^{a}\times \eta_2)}\prod_{\set{j, x} \in \E^b(\Lambda) \setminus E} \lambda_2(a', (\varphi^{a'} \times \eta_2)_j) \lambda_1(a, (\varphi^{a} \times \eta_2)_j) .
    \end{aligned}
    \end{equation}
    Let $\tilde \varphi := \varphi^{a'}\times \eta_2$. For all $j \neq x$ we have $\tilde \varphi_j = (\varphi^{a}\times \eta_2)_j$. Therefore, by assumption \zcref[noname]{item:lambda1/lambda2_non_increasing}
    \begin{equation}
        \prod_{\set{j, x} \in \E^b(\Lambda) \setminus E} \lambda_2(a, \tilde\varphi_j) \lambda_1(a', \tilde\varphi_j) 
        \leq \prod_{\set{j, x} \in \E^b(\Lambda) \setminus E} \lambda_2(a', \tilde\varphi_j) \lambda_1(a, \tilde\varphi_j) 
    \end{equation}
    and
    \begin{equation}
        \frac{Z^{\xi_1}_{\Lambda,p_1(\varphi^{a}\times \eta_2)}}{Z^{\xi_1}_{\Lambda,p_1(\varphi^{a'}\times \eta_2)}}
        = \phi_{\Lambda,p_1(\tilde\varphi)}^{\xi_1}\left(\prod_{\set{j, x} \in \E^b(\Lambda) \setminus E} \left(\frac{p_1(a,\tilde\varphi_j)}{1-p_1(a,\tilde\varphi_j)}\frac{1-p_1(a+1,\tilde\varphi_j)}{p_1(a+1,\tilde\varphi_j)}\right)^{\omega_{jx}}\right)
    \end{equation}
    again using the convention that $\frac00=1$. Since we assumed $a\mapsto {p_1(a,b)}$ to be non-decreasing, this is the expectation of a non-increasing function of $\omega$. Thus changing the boundary condition from $\xi_1$ to $\xi_2$ decreases the expectation and similarly does changing the probability from $p_1$ to $p_2$, since we assumed \zcref[noname]{item:p_non_increasing} that $p_1 \leq p_2$ (using standard monotonicity properties for the random cluster measure, cf.~\cite[Theorem 3.21]{Grimmett:RandomClusterModel2006}). Hence,
    \begin{equation}
        \frac{Z^{\xi_1}_{\Lambda,p_1(\varphi^{a}\times \eta_2)}}{Z^{\xi_1}_{\Lambda,p_1(\varphi^{a'}\times \eta_2)}}
        \geq \phi_{\Lambda,p_2(\tilde\varphi)}^{\xi_2} \left(\prod_{\set{j, x} \in \E^b(\Lambda) \setminus E} \left(\frac{p_1(a,\tilde\varphi_j)}{1-p_1(a,\tilde\varphi_j)}\frac{1-p_1(a+1,\tilde\varphi_j)}{p_1(a+1,\tilde\varphi_j)}\right)^{\omega_{jx}}\right)
    \end{equation}
    Applying assumption \zcref[noname]{item:p1_p2_relation} we obtain
    \begin{equation}
        \frac{Z^{\xi_1}_{\Lambda,p_1(\varphi^{a}\times \eta_2)}}{Z^{\xi_1}_{\Lambda,p_1(\varphi^{a'}\times \eta_2)}}
        \geq \phi_{\Lambda,p_2(\tilde\varphi)}^{\xi_2} \left(\prod_{\set{j, x} \in \E^b(\Lambda) \setminus E} \left(\frac{p_2(a,\tilde\varphi_j)}{1-p_2(a,\tilde\varphi_j)}\frac{1-p_2(a+1,\tilde\varphi_j)}{p_2(a+1,\tilde\varphi_j)}\right)^{\omega_{jx}}\right) 
        = \frac{Z^{\xi_2}_{\Lambda,p_2(\varphi^{a}\times \eta_2)}}{Z^{\xi_2}_{\Lambda,p_2(\varphi^{a'}\times \eta_2)}}
    \end{equation}
    which concludes the proof.
\end{proof}

To conclude this section, we give the following general condition which implies that $\Psi^\xi_{\Lambda;p,\lambda}$ satisfies the FKG inequality.
\begin{lemma}\label{thm:vertex_marginal_FKG_general}
    Let $E \subset \E^b(\Lambda)$ and $\xi \in \set{0,1}^E$.
    Let $\lambda,p \colon \set{0,\dots,S} \times \set{0,\dots S} \to [0, \infty)$ be symmetric and assume that $p \in [0,1)$ and $\lambda > 0$. Further, assume that $\set{0, \dots, S} \ni a \mapsto {p(a,b)}$ is non-decreasing for any $b \in \set{0, \dots, S}$, and that, setting $a'=a+1$ and $b'=b+1$,
    \begin{equation}\label{eq:FKG_condition_p_am}
        \frac{\lambda(a',b')\lambda(a,b)}{\lambda(a',b) \lambda(a,b')}
        \geq \frac{1-p(a',b')}{1-p(a',b)} \frac{1-p(a,b)}{1-p(a,b')} \max\left(1,  \frac{2-p(a',b)}{2-p(a',b')}\frac{2-p(a,b')}{2-p(a,b)}\right).
    \end{equation}
    Then $\Psi^\xi_{\Lambda;p, \lambda}$ satisfies the lattice FKG condition.
\end{lemma}
\begin{proof}
    Let $\varphi \in \set{0,\dots,S}^{\Z^d}$ such that $\varphi_i=S$ for all $i \in \Lambda^c$ and for $x,y\in \Lambda$ define
    $\varphi^{ab}$ by
    \begin{equation}
        \varphi_i^{ab} := \begin{cases}
            a & \text{if } i = x,\\
            b & \text{if } i = y,\\
            \varphi_i & \text{else}.\\
        \end{cases}
    \end{equation}
    For $a,b \in \set{0, \dots, S-1}$, $a' := a+1$, $b' := b+1$, it suffices to check (cf.~\cite[Section 6.3]{Lammers:DelocalisationAbsolutevalueFKGSolidonsolid2024})
    \begin{equation}
        \Psi^\xi_{\Lambda;p, \lambda}(\varphi^{a'b'})\Psi^\xi_{\Lambda;p, \lambda}(\varphi^{ab})
        \geq \Psi^\xi_{\Lambda;p, \lambda}(\varphi^{a'b})\Psi^\xi_{\Lambda;p, \lambda}(\varphi^{ab'}).
    \end{equation}
    We will assume that the right hand side is positive, otherwise the claim is trivial. Inserting the definitions, the condition is equivalent to
    \begin{equation}\label{eq:equivalent_condition_FKG_Psi}
    \begin{aligned}
        &\phntm \left({\lambda(a',b')\lambda(a,b)}\right)^{\I_{\set{x,y} \in \E^b(\Lambda)\setminus E}} {Z_{\Lambda;p(\varphi^{a'b'})}^{\xi} Z_{\Lambda;p(\varphi^{ab})}^{\xi}}\\
        &\geq \left({\lambda(a',b) \lambda(a,b')}\right)^{\I_{\set{x,y} \in \E^b(\Lambda)\setminus E}} {Z_{\Lambda;p(\varphi^{a'b})}^{\xi}Z_{\Lambda;p(\varphi^{ab'})}^{\xi}},
    \end{aligned}
    \end{equation}
    where we used that
    \begin{equation}
        \abs{\set{\varphi^{a'b'}=0}} + \abs{\set{\varphi^{ab}=0}} = \abs{\set{\varphi^{a'b}=0}} + \abs{\set{\varphi^{ab'}=0}}.
    \end{equation}
    To showcase the argument, assume that $\set{x,y} \in \E^b(\Lambda)$. The other case is similar and slightly easier. The claim \eqref{eq:equivalent_condition_FKG_Psi} follows from \eqref{eq:FKG_condition_p_am} if we show that
    \begin{equation}\label{eq:ZZ/ZZ_leq_p_cond}
        \frac{1-p(a',b')}{1-p(a',b)} \frac{1-p(a,b)}{1-p(a,b')} \max\left(1,  \frac{2-p(a',b)}{2-p(a',b')}\frac{2-p(a,b')}{2-p(a,b)}\right)
        \geq \frac{{Z_{\Lambda;p(\varphi^{a'b})}^{\xi}}}{Z_{\Lambda;p(\varphi^{a'b'})}^{\xi}} \frac{{Z_{\Lambda;p(\varphi^{ab'})}^{\xi}}}{Z_{\Lambda;p(\varphi^{ab})}^{\xi}}.
    \end{equation}
    To simplify the notation, we will denote the random cluster measure which has $Z_{\Lambda;p(\varphi^{a'b'})}^{\xi}$ as its partition function by $\phi^{a'b'}$, thereby dropping the dependence on $\Lambda$ and $\xi$.
    Denoting the restriction of $\omega$ to $\E(\Z^d) \setminus \set{\set{x,y}}$ by $\omega_{\langle xy\rangle}$, we obtain
    \begin{equation}\label{eq:rewrite_ZZ_by_RC}
        \frac{Z_{\Lambda;p(\varphi^{a'b})}^{\xi}}{Z_{\Lambda;p(\varphi^{a'b'})}^{\xi}}
        = \phi^{{a'b'}}\left(h(a',b)^{\omega_{xy}} F(\omega_{\langle xy\rangle})\right)
        \quad\text{and}\quad
        \frac{Z_{\Lambda;p(\varphi^{ab})}^{\xi}}{Z_{\Lambda;p(\varphi^{ab'})}^{\xi}}
        = \phi^{{ab'}}\left(h(a,b)^{\omega_{xy}} F(\omega_{\langle xy\rangle})\right),
    \end{equation}
    where
    \begin{equation}
        h(c,d):= \frac{p(c,d)(1-p(c,d+1))}{(1-p(c,d))p(c,d+1)}
    \end{equation}
    and
    \begin{equation}
        F(\omega_{\langle xy\rangle}) 
        := \prod_{\set{j,y} \in \E^b(\Lambda) \setminus (E \cup \set{x,y})} \left(\frac{p(b, \varphi_j)(1-p(b', \varphi_j))}{(1-p(b, \varphi_j))p(b', \varphi_j)}\right)^{\omega_{yj}}.
    \end{equation}
    Further,
    \begin{equation}
        \phi^{{a'b'}}\left(h(a',b)^{\omega_{xy}} F(\omega_{\langle xy\rangle})\right)
        = \phi^{{a'b'}}\left(\phi^{{a'b'}}\left(h(a',b)^{\omega_{xy}} \mid \omega_{\langle xy \rangle}\right) F(\omega_{\langle xy\rangle})\right)
    \end{equation}
    and
    \begin{equation}
        \phi^{{ab'}}\left(h(a,b)^{\omega_{xy}} F(\omega_{\langle xy\rangle})\right)
        = \phi^{{ab'}}\left(\phi^{{ab'}}\left(h(a,b)^{\omega_{xy}} \mid \omega_{\langle xy \rangle}\right) F(\omega_{\langle xy\rangle})\right).
    \end{equation}
    We can compute the inner expectations explicitly using the formulas for conditional expectations of the random cluster measure~\cite[Theorem 3.1]{Grimmett:RandomClusterModel2006}. We consider two cases:

    \emph{Case 1: $x$ and $y$ are connected in $\omega_{\langle xy \rangle}$}. Then
    \begin{equation}
        \begin{aligned}
            \phi^{{a'b'}}\left(h(a',b)^{\omega_{xy}} \mid \omega_{\langle xy \rangle}\right)
            &= \phi^{{a'b'}}\left(\omega_{xy} = 0 \mid \omega_{\langle xy \rangle}\right) + h(a',b) \phi^{{a'b'}}\left({\omega_{xy}=1} \mid \omega_{\langle xy \rangle}\right)\\
            &= 1-p(a',b')+ h(a',b) p(a',b') \\
            &= \frac{1-p(a',b')}{1-p(a',b)}
        \end{aligned}
    \end{equation}
    and similarly,
    \begin{equation}
        \phi^{{ab'}}\left(h(a,b)^{\omega_{xy}} \mid \omega_{\langle xy \rangle}\right)
        = \frac{1-p(a,b')}{1-p(a,b)}.
    \end{equation}

    \emph{Case 2: $x$ and $y$ are not connected in $\omega_{\langle xy \rangle}$}. Then
    \begin{equation}
        \begin{aligned}
            \phi^{{a'b'}}\left(h(a',b)^{\omega_{xy}} \mid \omega_{\langle xy \rangle}\right)
            &= \frac{1-p(a',b')}{1-p(a',b)} \frac{2-p(a',b)}{2-p(a',b')}
        \end{aligned}
    \end{equation}
    and similarly,
    \begin{equation}
        \phi^{{ab'}}\left(h(a,b)^{\omega_{xy}} \mid \omega_{\langle xy \rangle}\right)
        = \frac{1-p(a,b')}{1-p(a,b)} \frac{2-p(a,b)}{2-p(a,b')}.
    \end{equation}
    In particular, these computations imply that
    \begin{equation}\label{eq:phi_a'b'_c}
        \phi^{{a'b'}}\left(h(a',b)^{\omega_{xy}} F(\omega_{\langle xy\rangle})\right)
        \leq c(a,b) 
        \phi^{{a'b'}}\left(\phi^{{ab'}}\left(h(a,b)^{\omega_{xy}} \mid \omega_{\langle xy \rangle}\right) F(\omega_{\langle xy\rangle})\right)
    \end{equation}
    with
    \begin{equation}
        c(a,b) := \frac{1-p(a',b')}{1-p(a',b)} \frac{1-p(a,b)}{1-p(a,b')} \max\left(1,  \frac{2-p(a',b)}{2-p(a',b')}\frac{2-p(a,b')}{2-p(a,b)}\right).
    \end{equation}
    Since $h \in [0,1]$, the monotonicity of the random cluster measure in its boundary conditions \parencite[Lemma 4.14]{Grimmett:RandomClusterModel2006} implies that $\phi^{{ab'}}\left(h(a,b)^{\omega_{xy}} \mid \omega_{\langle xy \rangle}\right)$ is non-increasing in $\omega_{\langle xy \rangle}$. Similarly, $F$ is non-increasing. Thus, the stochastic domination when increasing the edge probabilities implies
    \begin{equation}\label{eq:phi_a'b'_final_monotonicity}
        \begin{aligned}
        \phi^{{a'b'}}\left(\phi^{{ab'}}\left(h(a,b)^{\omega_{xy}} \mid \omega_{\langle xy \rangle}\right) F(\omega_{\langle xy\rangle})\right)
        &\leq \phi^{{ab'}}\left(\phi^{{ab'}}\left(h(a,b)^{\omega_{xy}} \mid \omega_{\langle xy \rangle}\right) F(\omega_{\langle xy\rangle})\right)\\
        &= \phi^{{ab'}}\left(h(a,b)^{\omega_{xy}} F(\omega_{\langle xy\rangle})\right).
        \end{aligned}
    \end{equation}
    Combining \eqref{eq:phi_a'b'_c} with \eqref{eq:phi_a'b'_final_monotonicity} and \eqref{eq:rewrite_ZZ_by_RC} yields \eqref{eq:ZZ/ZZ_leq_p_cond}
\end{proof}

\section{Stochastic domination by a model with a magnetic field}
\label{sec:domination_Ising_half_integer}

In this section we will establish \zcref{thm:main_two_dimensions,thm:main_low_temperature}, which state that uniqueness and exponential decay hold in two dimensions for any convex, even interaction energy, or at sufficiently low temperature independent of $S$, provided one additionally assumes $W(1)>W(0)$. Further, we will prove \zcref{thm:constrained Gibbs measures},  which states that for all convex, even $W$ and at all temperatures, the magnetization of any Gibbs measure is in $(-\frac12, \frac12)$, provided $\rho > -\rho_*$.

In fact, we show
\begin{equation}\label{eq:two_dimensions_or_low_temp_plus_cluster_size}
    \mu_{\beta,\rho}^+ \left(0 \overset{+}{\longleftrightarrow} \partial B_N\right) \leq e^{-c_\beta N}
\end{equation}
under the assumptions of the respective theorems.
We then conclude by combining \zcref{thm:equivalence_no_percolation_uniqueness,thm:two_point_correlations_dominate,thm:two_point_correlation_relation_percolation,thm:bound_clustersize_open_by_plus}.
In dimensions $d \geq 3$, we expect that there is an intermediate range of temperatures where both the $+$ and $-$ clusters percolate, so that the quantity
\eqref{eq:two_dimensions_or_low_temp_plus_cluster_size} does not decay in that regime. This is known for the Ising model in sufficiently high dimensions (\cite{Aizenman:PercolationMinoritySpins1987}), which is why we treat the case of $d \geq 3$ with different methods in \zcref{sec:three_dimensions}.

The main tool is a domination of our model by a model with the same interaction on a different spin space, which additionally has a magnetic field that favors spins $\leq 0$.
We extend the Hamiltonian \eqref{eq:def_hamiltonian} through the addition of a magnetic field term $h \colon \Z \to \R$
\begin{equation}\label{eq:def_hamiltonian_with_h}
    \mathscr H_{\Lambda; \beta,\rho,h}(\sigma) := \beta \sum_{\set{i,j} \in \mathcal E^b(\Lambda)} W\left(\sigma_i - \sigma_j\right) + \rho \sum_{i \in \Lambda} \sigma_i^2 + \sum_{i \in \Lambda} h(\sigma_i), \quad \sigma \in \Omega.
\end{equation}
We define the restricted spin space
\begin{equation}
    \Sigma :=  
    \set{-S+1, \dots, S}^{\Z^d} \subset \Omega.
\end{equation}
For $\eta \in \Sigma$, the set of configurations with boundary condition $\eta$ outside $\Lambda$ is denoted as
\begin{equation}
    \Sigma_\Lambda^\eta := \set{\sigma \in \Sigma \mid \sigma_i = \eta_i \text{ for } i \in \Lambda^c} \subset \Omega_\Lambda^\eta
\end{equation}
and we consider a new model on $(\Omega, \mathcal F)$, defined by
\begin{equation}\label{eq:def_measure_spinM-12}
    \nu_{\Lambda;\beta,\rho,h}^\eta(\sigma) := \begin{cases}
        \left(\mathcal{Z}_{\Lambda;\beta,\rho,h}^\eta\right)^{-1} e^{-{\mathscr H}_{\Lambda;\beta,\rho, h}(\sigma)} & \text{if } \sigma \in \Sigma_\Lambda^\eta,\\
        0 & \text{else.}
    \end{cases}
\end{equation}
As before, $\mathcal{Z}_{\Lambda;\beta,\rho,h}^\eta$ is a constant which normalizes $\nu_{\Lambda;\beta,\rho,h}^\eta$ to be a probability measure. 
For a function $f \colon \Omega \to \R$ we denote the associated expectation by
$\nu_{\Lambda;\beta,\rho,h}^\eta(f)$. Analogous to before, we write $\nu_{\Lambda;\beta,\rho,h}^+ \equiv \nu_{\Lambda;\beta,\rho,h}^{S}$. 

In this case the infinite volume limit for $+$ boundary conditions also exists and can be obtained as a weak limit
\begin{equation}
    \nu_{B_N;\beta,\rho,h}^+ \Rightarrow \nu_{\beta,\rho,h}^+ \quad\text{as } N \to \infty.
\end{equation}
The key tool we establish is the following stochastic domination.

\begin{proposition}\label{thm:dominate_M_by_M12}
    Let $\beta > 0$ and $\rho \in \R$. For any $\Lambda \Subset \Z^d$ and $\eta \in \Sigma$, the stochastic domination
    \begin{equation}\label{eq:dominate_M_by_M12}
        \mu_{\Lambda;\beta,\rho}^\eta \leq_{st} \nu_{\Lambda;\beta,\rho, \lambda_0 H}^{\eta}
    \end{equation}
    holds with $\lambda_0:=-\ln \mu_{\set{0};\beta,\rho}^+\left(\sigma_0 \geq -S+1\right)>0$ and any $H \colon \Z \to [0,1]$. 
\end{proposition}
We note that the FKG inequality trivially implies the stochastic domination for $H\equiv 0$. The key point of the proposition is that this domination continues to hold for $H\in[0,1]$ provided that $\lambda_0$ is sufficiently small. Our application will take $H(k) := \I_{k \geq 1}$.

\begin{proof}[Proof of \zcref{thm:dominate_M_by_M12}]
    We couple the two measures using Glauber dynamics. We define a Markov chain $(X_n, Y_n)$ on $(\Omega \times \Omega, \mathcal F \otimes \mathcal F)$ in the following way. Let $X_0 = Y_0 = \eta$. Assume now that the chain is defined up to step $n$. Sample uniformly, independently from all other sources of randomness, $i_{n+1} \sim \Unif (\Lambda)$ and $U_{n+1} \sim \Unif ([0,1))$.

    Define for $j \in \Z^d$ and $s \in \set{-S, \dots, S}$
    \begin{equation}
        (X_{n+1})_j := \begin{cases}
            s & \text{if $i_{n+1} = j$ and } U_{n+1} \in \left[\mu_{\set{j};\beta,\rho}^{X_n}(\sigma_j > s), ~\mu_{\set{j};\beta,\rho}^{X_n}(\sigma_j \geq s) \right), \\
            (X_{n})_j & \text{else},
        \end{cases}
    \end{equation}
    and for $s \in \set {-S + 1, \dots, S}$ 
    \begin{equation}
        (Y_{n+1})_j := \begin{cases}
            s & \text{if $i_{n+1} = j$ and } U_{n+1} \in \left[\nu_{\set{j};\beta,\rho,\lambda_0 H}^{Y_n}(\sigma_j > s), ~\nu_{\set{j};\beta,\rho,\lambda_0 H}^{Y_n}(\sigma_j \geq s) \right), \\
            (Y_{n})_j & \text{else}.
        \end{cases}
    \end{equation}
    $X_n$ and $Y_n$ define two Markov chains and we denote their joint law by $\P_\eta$. Both are irreducible (because one can go to any other configuration with positive probability by changing one vertex at a time) and aperiodic (since $\P_\eta(X_n=X_{n+1})>0$), c.f also~\cite[Section 3.10.3]{Friedli:StatisticalMechanicsLattice2017}. Thus their law converges weakly to the stationary distribution, which we claim are $\mu_{\Lambda;\beta,\rho}^\eta$ and $\nu_{\Lambda;\beta,\rho,\lambda_0H}^\eta$ respectively.

    Indeed, if we consider two configurations $\sigma, \sigma'$ which differ only at $j \in \Lambda$, then the transition probability from $\sigma$ to $\sigma'$ for $X_n$ is
    \begin{equation}
        \frac1{\abs{\Lambda}} \frac{\mu_{\Lambda;\beta,\rho}^\eta(\sigma')}{\sum_{s=-S}^S\mu_{\Lambda;\beta,\rho}^\eta(\hat\sigma_s)}
    \end{equation}
    where $\hat\sigma_s(j)=s$ and $\hat\sigma_s(i)=\sigma(i)=\sigma'(i)$ for all $i \neq j$. Hence, one sees that $X_n$ is reversible with respect to $\mu_{\Lambda;\beta,\rho}^\eta$. A similar argument shows that $Y_n$ is reversible with respect to $\nu_{\Lambda;\beta,\rho,\lambda_0 H}^\eta$.
    
    In particular, for all $\sigma \in \Omega$, $\bar\sigma \in \Sigma$
    \begin{equation}\label{eq:weak_convergence_Xn_Yn}
        \mu_{\Lambda;\beta,\rho}^\eta\left(\sigma\right) = \lim_{n \to \infty} \P_\eta\left(X_n=\sigma\right)
        \quad\text{and} \quad
        \nu_{\Lambda;\beta,\rho,\lambda_0 H}^\eta\left(\bar\sigma\right) = \lim_{n \to \infty} \P_\eta\left(Y_n=\bar\sigma\right).
    \end{equation}
    We claim that $X_n \leq Y_n$ for all $n \in \N$. We prove the statement inductively. By definition it is true for $n=0$. Assume thus that $X_n \leq Y_n$. Then we must show that for any $s \in \set{-S, \dots, S}$ we have $(Y_{n+1})_{i_{n+1}} \geq s$ if $(X_{n+1})_{i_{n+1}} \geq s$. 
    By definition of the chains, this is equivalent to proving
    \begin{equation}\label{eq:nu_dominates_mu_pointwise}
        \mu_{\set{i_{n+1}};\beta,\rho}^{X_n}(\sigma_{i_{n+1}} \geq s) 
        \leq \nu_{\set{{i_{n+1}}};\beta,\rho,\lambda_0 H}^{Y_n}(\sigma_{i_{n+1}} \geq s)
    \end{equation}
    for any $s \in \set{-S, \dots, S}$. Since it is obvious for $s =  -S$, assume that $s \geq -S +1$.
    By the FKG inequality,
    \begin{equation}
    \begin{aligned}
        \mu_{\set{{i_{n+1}}};\beta,\rho}^{X_n}(\sigma_{i_{n+1}} \geq s) 
        &\leq \mu_{\set{{i_{n+1}}};\beta,\rho}^{Y_n}(\sigma_{i_{n+1}} \geq s) \\
        &= \frac{\sum_{\sigma \in \set{s, \dots, S}}\expa{-\beta\sum_{j\sim {i_{n+1}}} W(\sigma-(Y_n)_j)-\rho\sigma^2}}{\sum_{\sigma \in \set{-S, \dots, S}}\expa{-\beta\sum_{j\sim {i_{n+1}}} W(\sigma-(Y_n)_j)-\rho\sigma^2}}.
    \end{aligned}
    \end{equation}
    Using translation invariance of the lattice,
    \begin{equation}
        \frac{\sum_{\sigma \in \set{s, \dots, S}}\expa{-\beta\sum_{j\sim {i_{n+1}}} W(\sigma-(Y_n)_j)-\rho\sigma^2}}{\sum_{\sigma \in \set{-S, \dots, S}}\expa{-\beta\sum_{j\sim {i_{n+1}}} W(\sigma-(Y_n)_j)-\rho \sigma^2}} \\
        \leq \gamma \, \nu_{\set{{i_{n+1}}};\beta,\rho,0}^{Y_n}(\sigma_{i_{n+1}} \geq s),
    \end{equation}
    where
    \begin{equation}
        \begin{aligned}
            \gamma
            :=&\sup_{\eta \in \Sigma} \frac{\sum_{\sigma \in \set{-S+1, \dots, S}}\expa{-\beta\sum_{j\sim {0}} W(\sigma-\eta_j)-\rho\sigma^2}}{\sum_{\sigma \in \set{-S, \dots, S}}\expa{-\beta\sum_{j\sim {0}} W(\sigma-\eta_j)-\rho\sigma^2}} \\
            &= \mu_{\set{0};\beta,\rho}^+\left(\sigma_0 \geq -S+1\right) < 1,
        \end{aligned}
    \end{equation}
    where the equality follows again by FKG. Since $H\in[0,1]$, it holds that
    \begin{equation}
    \begin{aligned}
        \nu_{\set{{i_{n+1}}};\beta,\rho,0}^{Y_n}(\sigma_{i_{n+1}} \geq s)
        &= \frac{\nu_{\set{{i_{n+1}}};\beta,\rho,0}^{Y_n}(\sigma_{i_{n+1}} \geq s)}{\nu_{\set{{i_{n+1}}};\beta,\rho,\lambda_0 H}^{Y_n}(\sigma_{i_{n+1}} \geq s)} \nu_{\set{{i_{n+1}}};\beta,\rho,\lambda_0 H}^{Y_n}(\sigma_{i_{n+1}} \geq s)\\
        &\leq e^{\lambda_0} \nu_{\set{{i_{n+1}}};\beta,\rho,\lambda_0 H}^{Y_n}(\sigma_{i_{n+1}} \geq s).
    \end{aligned}
    \end{equation}
    Recall that $\lambda_0 = -\ln\gamma >0$, so we obtain \eqref{eq:nu_dominates_mu_pointwise}.
    To prove \eqref{eq:dominate_M_by_M12}, note that by the weak convergence of the laws of $X_n$ and $Y_n$ \eqref{eq:weak_convergence_Xn_Yn} and because $X_n \leq Y_n$ $\P_\eta$-almost-surely, we have
    \begin{equation}
            \nu_{\Lambda;\beta,\rho,\lambda_0 H}^\eta(f) - \mu_{\Lambda;\beta,\rho}^\eta(f)
            = \lim_{n \to \infty} \sum_{{\sigma \in \Omega, \bar \sigma \in \Sigma}} \left(f(\bar\sigma)-f(\sigma)\right)\I_{\bar\sigma \geq \sigma} \P_\eta\left(X_n=\sigma, ~Y_n = \bar\sigma\right)
            \geq 0
    \end{equation}
    for any $f\colon \Omega \to \R$ non-decreasing.
\end{proof}

As for the measure $\mu_{\Lambda;\beta,\rho}^\eta$, it will be useful to consider a percolation model in an environment determined by $\nu_{\Lambda;\beta,\rho,h}^\eta$. Given $\sigma \in \Sigma$ and $\set{i,j} \in \mathcal E\left(\Z^d\right)$ define
\begin{equation}
    q(\sigma_i, \sigma_j) := 1 - e^{\beta W(\sigma_i - \sigma_j) - \beta W(\abs{\sigma_i - \frac12} + \abs{\sigma_j - \frac12})}.
\end{equation}
Recall that $\bar\Omega := \Omega \times \set{0,1}^{\E(\Z^d)}$. We define a new probability measure on $(\bar\Omega, \mathcal F \otimes \mathcal G)$ (where $\mathcal F$ and $\mathcal G$ are the cylinder $\sigma$-algebras on $\Omega$ and $\set{0,1}^{\E(\Z^d)}$ respectively) by
\begin{equation}
    \bar \nu_{\Lambda;\beta,\rho,h}^\eta (\sigma, \omega) := 
        \nu_{\Lambda;\beta,\rho,h}^\eta(\sigma) \prod_{\set{i,j} \in \mathcal E^b(\Lambda)} q(\sigma_i, \sigma_j)^{\omega_{ij}} \left(1-q(\sigma_i, \sigma_j)\right)^{1-\omega_{ij}}
\end{equation}
if $\omega_e = 1$ for all $e \in \mathcal E^b(\Lambda)^c$ and $\bar \nu_{\Lambda;\beta,\rho,h}^\eta (\sigma, \omega) := 0$ otherwise.
We denote the expectation with respect to $\bar \nu_{\Lambda;\beta,\rho,h}^\eta$ by $\bar \nu_{\Lambda;\beta,\rho,h}^\eta(\cdot)$.

\begin{lemma}\label{thm:>12_cluster_exponential_decay}
    Let $\beta > 0$, $\rho \in \R$ and $h \colon \set{-S+1,\dots,S} \to \R$.
    Assume that for all $s \in \set{1, \dots, S}$
    \begin{equation}
        h(s)-h(1-s)+\rho(2s - 1) > 0.
    \end{equation}
    There exist constants $c_1 \equiv c_1(S, h, \rho, W, d, \beta) > 0$ and $c_2 \equiv c_2(S, h, \rho) > 0$ such that
    for all $\Lambda \Subset \Z^d$, $x \in \Lambda$ and $N \in \N$ with $B_N(x) \subset \Lambda$ and $N \geq \frac{9^d}4$ it holds
    \begin{equation}\label{eq:>12_cluster_exponential_decay}
        \bar\nu_{\Lambda;\beta,\rho,h}^+\left(\sigma_x \geq 1 \text{ and } x \overset{\text{open}}{\longleftrightarrow} \partial B_{N}(x)\right)
        \leq c_1 e^{-c_2 N}.
    \end{equation}
\end{lemma}
\begin{proof}
    Without loss of generality assume that $\frac N{4d} \in \N$. Otherwise, repeat the same argument for $4d\lfloor \frac N{4d}\rfloor$ and adapt the constants.

    Let $\omega \in \set{0,1}^{\mathcal E(\Z^d)}$ and $L \leq N$. Consider the configuration $\omega_{L}^x$ obtained from $\omega$ by closing all edges of $\E(\Z^d) \setminus \E(B_L(x))$. We define $\mathcal T_L^x(\omega)$ to be the open cluster of $x$ in $\omega_L^x$.
    Then,
    \begin{equation}
        \bar\nu_{\Lambda;\beta,\rho,h}^+\left(\sigma_x \geq 1 \text{ and } x \overset{\text{open}}{\longleftrightarrow} \partial B_{N}(x)\right)
        = \bar\nu_{\Lambda;\beta,\rho,h}^+\left(\sigma_x \geq 1 \text{ and } \mathcal T_N^x \cap \partial B_N(x) \neq \emptyset\right).
    \end{equation}
    To simplify the notation, we will drop the dependence on $x$ when writing $\mathcal T_N^x$ or $B_N(x)$ (and also when replacing $N$ by some other value) for the rest of this proof.

    The main idea is to perform a spin flip around $\frac12$ of the spins in $\mathcal T_N$. This would gain a factor proportional to $\abs{\mathcal T_N}$, because each spin $\geq 1$ is penalized by $h$. However, at the same time it might incur a penalty of the same order from the boundary edges $\partial_{\exterior}\mathcal T_N \cap \partial_{\exterior} B_N$, because these could be open and would need to be closed before the spin flip. The following argument deals with this issue.

    Partition the set $\set{N/2+1, \dots, N}$ into $2d$ intervals $I_1, \dots, I_{2d}$, each containing $\frac{N}{4d}$ points and let $k(L) \in \set{1,\dots, 2d}$ be the unique index such that $L \in I_{k(L)}$. Moreover, let
    \begin{equation}
        L_*(\omega) := \min \set{L \in \set{N/2+1, \dots, N} \mid \abs{\mathcal T_L(\omega) \cap \partial B_L} \leq 3^dd N^{\frac{k(L)-1}2}}.
    \end{equation}
    This is well defined, since $k(N)=2d$ and
    \begin{equation}
        \abs{\mathcal T_N(\omega) \cap \partial B_N}
        \leq \abs{\partial B_N} \leq 2d (2N+1)^{d-1} \leq 3^d d N^{\frac{k(N)-1}2},
    \end{equation}
    so we are taking the minimum over a finite, non-empty set. 
    If $k(L_*)=1$, then,
    \begin{equation}
            \abs{\mathcal T_{L_*}}
            = \sum_{L=0}^{L_*} \abs{\mathcal T_{L_*} \cap \partial B_L}
            \geq \sum_{L=0}^{N/2} \abs{\mathcal T_{L} \cap \partial B_L}
            \geq \frac N2 \geq \frac{\sqrt N}{4d} \abs{\mathcal T_{L_*} \cap \partial B_{L_*}},
    \end{equation}
    since $\abs{\mathcal T_{L_*} \cap \partial B_{L_*}} \leq 3^d d$ by definition of $L_*$. 
    On the other hand, if $k(L_*) \geq 2$ we have
    \begin{equation}
        \begin{aligned}
            \abs{\mathcal T_{L_*}}
            = \sum_{L=0}^{L_*} \abs{\mathcal T_{L_*} \cap \partial B_L}
            \geq \sum_{L \colon k(L)=k(L_*)-1} \abs{\mathcal T_{L} \cap \partial B_L}
            \geq \frac{3^{d}d}{4d} N^{\frac{k(L_*)}2} 
            \geq \frac{\sqrt N}{4d} \abs{\mathcal T_{L_*} \cap \partial B_{L_*}},
        \end{aligned}
    \end{equation}
    where we used that for $L < L_*$ we have $\abs{\mathcal T_{L_*} \cap \partial B_L} \geq \abs{\mathcal T_{L} \cap \partial B_L} \geq 3^d d N^{\frac{k(L)-1}2}$. 
    Thus,
    \begin{equation}
        \begin{aligned}\label{eq:T_N_with_fixed_L*_as_sum_over_graphs}
            &\phntm\bar\nu_{\Lambda;\beta,\rho,h}^+\left(\sigma_x \geq 1 \text{ and } \mathcal T_N \cap B_N \neq \emptyset\right)\\
            &\leq \bar\nu_{\Lambda;\beta,\rho,h}^+\left(\sigma_x \geq 1 \text{ and } \abs{\mathcal T_{L_*}} \geq L_* \text{ and } \sqrt N\abs{\mathcal T_{L_*} \cap \partial B_{L_*}} \leq 4d \abs{\mathcal T_{L_*}}\right) \\
            &\leq \sum_L\sum_T \bar\nu_{\Lambda;\beta,\rho,h}^+\left(\sigma_x \geq 1 \text{ and } \mathcal T_L = T\right)
        \end{aligned}
    \end{equation}
    where the first sum is over all $L \in \set{\tfrac N2 + 1 + \tfrac N{4d}, \dots, N}$ and the second sum is over all subgraphs $T \subset B_L$ connected such that $x \in T$, $\abs{T} \geq L$ and $\sqrt N\abs{T \cap \partial B_L} \leq 4d \abs{T}$.
    For such a $T$ we have
    \begin{equation}
    \begin{aligned}
        &\phntm\mathcal Z_{\Lambda;\beta,\rho,h}^+ 
        \bar\nu_{\Lambda;\beta,\rho,h}^+\left(\sigma_x \geq 1 \text{ and } \mathcal T_L = T\right)\\
        &= \sum_{\substack{\sigma \in \Sigma_\Lambda^+\\ \sigma_y \geq 1  \text{ for }y \in T}} 
                \prod_{i \in \Lambda} e^{-h(\sigma_i)-\rho\sigma_i^2}
                \prod_{\substack{\set{i,j} \in \mathcal E^b(\Lambda)\\i,j\notin T}} e^{-\beta W(\sigma_i-\sigma_j)}
                \prod_{\substack{\set{i,j} \in T}} e^{-\beta W(\sigma_i-\sigma_j)}q(\sigma_i, \sigma_j) \\
                &\qquad\times \prod_{\substack{\set{i,j} \in \partial_{\exterior} T \cap \mathcal E\left(B_{L}\right)}} e^{-\beta W(\sigma_i-\sigma_j)} (1-q(\sigma_i, \sigma_j))
                \prod_{\substack{\set{i,j} \in \partial_{\exterior} T \setminus \mathcal E\left(B_{L}\right)}} e^{-\beta W(\sigma_i-\sigma_j)} \\
        &\leq e^{c_0\beta\abs{\partial_{\exterior} T \setminus \mathcal E\left(B_{L}\right)}} \sum_{\substack{\sigma \in \Sigma_\Lambda^+\\ \sigma_y \geq 1 \text{ for } ~y \in T}} 
                \prod_{i \in \Lambda} e^{-h(\sigma_i)-\rho\sigma_i^2}
                \prod_{\substack{\set{i,j} \in \mathcal E^b(\Lambda)\\i,j\notin T}} e^{-\beta W(\sigma_i-\sigma_j)}\\
                &\qquad\times \prod_{\substack{\set{i,j} \in T}} e^{-\beta W(\sigma_i-\sigma_j)}q(\sigma_i, \sigma_j)
                \prod_{\substack{\set{i,j} \in \partial_{\exterior} T}} e^{-\beta W\left(\abs{\sigma_i-\frac12}+\abs{\sigma_j-\frac12}\right)} \\
        &= e^{c_0\beta\abs{\partial_{\exterior} T \setminus \mathcal E\left(B_{L}\right)}} \mathcal Z_{\Lambda;\beta,\rho,h}^+\bar\nu_{\Lambda;\beta,\rho,h}^+\left(\sigma_x \geq 1 \text{ and } \mathcal T_L = T \text{ and } \omega_e = 0 \text{ for all } e \in \partial_{\exterior} T\right)\\
    \end{aligned}
    \end{equation}
    where
    \begin{equation}
        c_0 := W(2S-1) - W(0).
    \end{equation}
    If we now flip all spins in $T$ around $\frac 12$ and leave the spins outside $T$ as they are, then the interaction stays the same for every edge and we gain a factor of order $\abs{T}$ from the magnetic field. Precisely,
    \begin{equation}
    \begin{aligned}
        &\phntm 
        \bar\nu_{\Lambda;\beta,\rho,h}^+\left(\sigma_x \geq 1 \text{ and } \mathcal T_L = T\right)\\
        &\leq e^{c_0\beta\abs{\partial_{\exterior} T \setminus \mathcal E\left(B_{L}\right)}-c_1 \abs{T}}
        \bar\nu_{\Lambda;\beta,\rho,h}^+\left(\sigma_x \leq 0 \text{ and } \mathcal T_L = T \text{ and } \omega_e = 0 \text{ for all } e \in \partial_{\exterior} T\right),
    \end{aligned}
    \end{equation}
    with
    \begin{equation}
        c_1 
        := \min_{s \in \set{1, \dots, S}} h(s)-h(1-s)+\rho(2s - 1) > 0
    \end{equation}
    by assumption. Therefore,
    \begin{equation}
        \bar\nu_{\Lambda;\beta,\rho,h}^+\left(\sigma_x \geq 1 \text{ and } \mathcal T_L = T\right)
        \leq \expa{-\abs{T} \left(c_1 - \frac{4d^2\beta c_0}{\sqrt N}\right) } \bar\nu_{\Lambda;\beta,\rho,h}^+\left(\sigma_x \leq 0 \text{ and } \mathcal T_L = T\right)
    \end{equation}
    where we used that $\abs{\partial_{\exterior} T \setminus \mathcal E\left(B_{L}\right)} \leq d \abs{T \cap \partial B_L} \leq \frac{4d^2}{\sqrt N} \abs{T}$.
    Choosing $N$ large enough and inserting into \eqref{eq:T_N_with_fixed_L*_as_sum_over_graphs} we obtain
    \begin{equation}\label{eq:T_N_final_bound_fixed_L*}
        \bar\nu_{\Lambda;\beta,\rho,h}^+\left(\sigma_x \geq 1 \text{ and } \mathcal T_N \cap \partial B_N \neq \emptyset\right)
        \leq \frac N2\expa{-\frac{N c_1}4}
    \end{equation}
    where we used $\abs T \geq \frac N2$ and
    $\sum_L\sum_T \bar\nu_{\Lambda;\beta,\rho,h}^+\left(\mathcal T_L = T\right) \leq \frac N2$.
    Therefore, for all $N \in \N$,
    \begin{equation*}
        \bar\nu_{\Lambda;\beta,\rho,h}^+\left(\sigma_x \geq 1 \text{ and } x \overset{\text{open}}{\longleftrightarrow} \partial B_{N}(x)\right)
        \leq c_2 \expa{-\frac{ c_1}5 N},
    \end{equation*}
    for some constant $c_2(S,h,W,d,\beta)>0$.
\end{proof}

\subsection{The two dimensional case at all temperatures}
\label{sec:minority_percolation_2d}

The two dimensional case is proved by showing that in the measure $\nu_{\beta,\rho,h}^+$, with $h = \lambda_0 \I_{k \geq 1}$ and $\lambda_0 > 0$, the magnetic field beats the boundary condition in the sense that spins $\leq 0$ are more likely to percolate than spins $\geq 1$. In two dimensions, this is enough to show that spins $\geq 1$ do not percolate at all (Zhang's argument), which by \zcref{thm:equivalence_no_percolation_uniqueness} proves that there is a unique infinite volume Gibbs measure $\mu_{\beta, \rho}$.

We extend \zcref{def:connections} as follows.
\begin{definition}
     For $x, y \in \Z^d$ we write 
    \begin{equation}\begin{split}
        \phantom{{}={}}
        \set{x \overset{\leq 0}{\longleftrightarrow} y}
        := \{&\sigma \in \Omega \mid \exists k \in \N, \, z_0, \dots z_k \in \Z^d \colon \\
        &z_0 = x, \,z_k=y, \,z_i \sim z_{i+1},  \,\sigma_{z_i} \leq 0 \text{ for all } i \}
    \end{split}\end{equation}
    for the configurations in which $x$ is connected to $y$ by a path of spins $\leq 0$. 
    For $A \subset \Z^d$ we define
    \begin{equation}
        \set{x \overset{\leq 0}{\longleftrightarrow} A} := \bigcup_{y \in A} \set{ x \overset{\leq 0}{\longleftrightarrow} y}.
    \end{equation}
\end{definition}

\begin{lemma}\label{thm:s-12-ising_minus_infty_cluster_dominate_plus}
    Let $\beta > 0$, $\rho \in \R$ and $h \colon \set{-S+1,\dots,S} \to \R$ such that for all $s \in \set{1, \dots, S}$
    \begin{equation}
        h(s)-h(1-s)+\rho(2s-1) > 0.
    \end{equation}
    Then,
    \begin{equation}\label{eq:s-12-ising_minus_infty_cluster_dominate_plus}
        \nu_{\beta,\rho,h}^+ \left(0\overset{+}{\longleftrightarrow} \infty\right)
        \leq \nu_{\beta,\rho,h}^+ \left(0\overset{\leq 0}{\longleftrightarrow} \infty\right).
    \end{equation}
\end{lemma}
\begin{proof}
    Let $N \in \N$ and let $\Lambda \Subset \Z^d$ such that $B_{2N} \subset \Lambda$.  Recall \zcref{def:connections}: $\mathscr C(\omega)$ denotes the open clusters of $\omega$, $C_x$ is the unique open cluster containing $x \in \Z^d$, and $\mathscr C_\Delta(\omega)$ is the set of all open clusters intersecting $\Delta \subset \Z^d$.
    Define the event
    \begin{equation}
        \mathcal A := \set{\forall x \in B_N \colon C_x \subset \Lambda} \in \mathcal G.
    \end{equation}
    With the help of \zcref{thm:>12_cluster_exponential_decay} we can estimate
    \begin{equation}\label{eq:probability_no_12_circuit}
    \begin{aligned}
        \bar\nu_{\Lambda;\beta,\rho,h}^+ \left(\mathcal A^c\right)
        &\leq \sum_{x \in \partial B_{N}} \bar\nu_{\Lambda;\beta,\rho,h}^+\left(C_x \cap \Lambda^c \neq \emptyset, ~\sigma_x \geq 1\right) + \bar\nu_{\Lambda;\beta,\rho,h}^+\left(C_x \cap \Lambda^c \neq \emptyset, ~\sigma_x \leq 0\right) \\
        &\leq \sum_{x \in \partial B_{N}} \bar\nu_{\Lambda;\beta,\rho,h}^+\left(x \overset{\text{open}}{\longleftrightarrow} \partial B_{N}(x) \text{ and }\sigma_x \geq 1\right) \\
        &\leq c_1 N^{d-1} e^{-c_2 N},
    \end{aligned}
    \end{equation}
    where we used that $\bar\nu_{\Lambda;\beta,\rho,h}^+\left(C_x \cap \Lambda^c \neq \emptyset, ~\sigma_x \leq 0\right) = 0$, because otherwise there exists an open edge $\set{u,v} \in \mathcal E(C_x(\omega))$ such that $(\sigma_u-\frac12)(\sigma_v-\frac12) \leq 0$. But in that case $q(\sigma_u, \sigma_v) = 0$, therefore such a configuration has probability zero.

    We decompose the event $\mathcal A$ as follows:
	\begin{equation}\label{eq:decomposition_clusters_BN_in_Lambda}
		\mathcal A = \bigcup_{m \in \N} \left(\bigcup_{C_0, \dots, C_m} \set{\mathscr C_{B_N} = \set{C_0, \dots, C_m}}\right),
	\end{equation}
	where $C_0, \dots, C_m \subset \Lambda$ are pairwise (vertex-) disjoint subgraphs of $\Lambda$ satisfying $B_N \subset \bigcup_{l=0}^m C_l =: \bar C$. We always assume the labels are chosen such that $0 \in C_0$.
	
	Fix $\omega \in \mathcal A$ such that $\mathscr C_{B_N}(\omega) = \set{C_0, \dots, C_m}$, $C_l \subset \Lambda$ for all $l \in \set{0,\dots, m}$. Let $\set{i,j} \in \E^b(\Lambda)$. To see how this edge contributes to the measure $\bar\nu_{\Lambda;\beta,\rho,h}^+$, we need to consider the following three possible cases:
	\begin{itemize}
		\item If $i,j \in \Z^d \setminus \bar C$ the contribution is $e^{-\beta W(\sigma_i-\sigma_j)}q(\sigma_i, \sigma_j)^{\omega_{ij}} (1-q(\sigma_i, \sigma_j))^{1-\omega_{ij}}$.
		\item If $i \in C_l$ and $\omega_{ij}=1$ the edge belongs to the cluster $C_l$ and thus also $j \in C_l$. Hence, the contribution is
		\begin{equation}
			e^{-\beta W(\sigma_i-\sigma_j)} q(\sigma_i, \sigma_j) 
			= e^{-\beta W(\abs{\sigma_i - \frac12} - \abs{\sigma_j-\frac12})}\left(1- e^{\beta W(\abs{\sigma_i - \frac12} - \abs{\sigma_j-\frac12})-\beta W(\abs{\sigma_i - \frac12} + \abs{\sigma_j-\frac12})}\right)
		\end{equation}
        if $\sign (\sigma_i - \frac12) = \sign (\sigma_j - \frac12)$ and $0$ otherwise.
		\item If $i \in C_l$ and $\omega_{ij} = 0$ the contribution is
		\begin{equation}
			e^{-\beta W(\sigma_i-\sigma_j)} (1-q(\sigma_i, \sigma_j))
			=e^{-\beta W(\abs{\sigma_i - \frac12} + \abs{\sigma_j-\frac12})}.
		\end{equation}
	\end{itemize}
	In particular, conditioning on  $\mathscr C_{B_N}=\set{C_0, \dots, C_m}$, $\abs{\sigma_i - \frac12}$ for all $i \in \bar C$ and on $\sigma_i$ for all $i \in \Z^d \setminus \bar C$, the signs of $\sigma - \frac12$ on each cluster $C_l$ are almost surely constant and the signs on two different clusters are independent. 
	
	More precisely, let $\kappa \in \set{0,1}^{\bar C}$ and define $\pi(s) = 1$ if $s \geq \frac12$ and $\pi(s)=0$ otherwise. Then, writing $\sigma' = \sigma - \frac12$,
	\begin{equation}
		\begin{aligned}
			&\phntm\bar\nu_{\Lambda;\beta,\rho,h}^+\left(\forall i \in \bar C \colon \pi(\sigma_i) = \kappa_i \mid \mathscr C_{B_N} = \set{C_0, \dots C_m}, \,\abs{\sigma'_{\bar C}}, \, \sigma_{\Z^d \setminus \bar C}\right)\\
			&= \begin{cases}
				\prod_{l=0}^m \left(\left(\sum_{\bar\kappa \in \set{0,1}}\prod_{i \in C_l} e^{-J(\abs{\sigma'_i}+\frac12)\bar\kappa}\right)^{-1} 
				\prod_{i \in C_l} e^{-J(\abs{\sigma'_i}+\frac12)\kappa_i}\right) &\text{$\kappa$ const. on every $C_l$},\\
				0 & \text{else},
			\end{cases}
		\end{aligned}
	\end{equation}
	where for $s \in \set{1, \dots, S}$
	\begin{equation}
		J(s) := h(s) - h(1-s) + \rho(2s-1) > 0,
	\end{equation}
    and where for $\Delta \subset \Z^d$ we write $\sigma_\Delta$ for the configuration restricted to $\Delta$.
    
	Consider a graph structure $G=(V,E)$ on $V=\set{0, \dots, m}$, where $\set{i,j} \in E$ if and only if there exist $x \in C_i$ and $y \in C_j$ such that $\set{x,y} \in \E(\Z^d)$. 
	Let
	\begin{equation}
		p_l(\abs{\sigma'}) := \left(1 + e^{\sum_{i \in C_l} J\left(\abs{\sigma'_i}+\tfrac12\right)}\right)^{-1} < \frac12.
	\end{equation}
	and let $\P_{p(\abs{\sigma'})}^{G}$ be the Bernoulli site percolation measure on $G$, with a site $l \in \set{0, \dots, m}$ being open with probability $p_l(\abs{\sigma'})$.
	For any $\kappa \in \set{0,1}^{m+1}$, we have
	\begin{equation}
			\P_{p(\abs{\sigma'})}^{G}(\kappa)
			= \bar\nu_{\Lambda;\beta,\rho,h}^+\left(\forall l \in \set{0, \dots, m}, \,i \in C_l \colon \pi(\sigma_i) = \kappa_l \mid \mathscr C_{B_N} = \set{C_0, \dots C_m}, \,\abs{\sigma'_{\bar C}}, \, \sigma_{\Z^d \setminus \bar C}\right).
	\end{equation}
	Therefore, an open site in $\P_{p(\abs{\sigma'})}^{G}$ corresponds to an open cluster with spins $\geq 1$ in $\bar\nu_{\Lambda;\beta,\rho,h}^+$, while a closed site corresponds to a cluster with spins $\leq 0$. Let $\Delta := \set{l \in \set{0, \dots,m} \colon C_l \cap\partial B_N \neq \emptyset}$. Then,
	\begin{equation}\label{eq:site_perc_open}
		\begin{aligned}
			\bar\nu_{\Lambda;\beta,\rho,h}^+\left(0 \overset{+}{\longleftrightarrow} \partial B_N \mid \mathscr C_{B_N} = \set{C_0, \dots C_m}, \,\abs{\sigma'_{\bar C}}, \, \sigma_{\Z^d \setminus \bar C}\right)
			= \P^G_{p(\abs{\sigma'})}(0 \overset{\text{open}}{\longleftrightarrow} \Delta). 
		\end{aligned}
	\end{equation}
	Similarly,	
	\begin{equation}\label{eq:site_perc_closed}
		\begin{aligned}
			\bar\nu_{\Lambda;\beta,\rho,h}^+\left(0 \overset{\leq 0}{\longleftrightarrow} \partial B_N \mid \mathscr C_{B_N} = \set{C_0, \dots C_m}, \,\abs{\sigma'_{\bar C}}, \, \sigma_{\Z^d \setminus \bar C}\right)
			= \P^G_{p(\abs{\sigma'})}(0 \overset{\text{closed}}{\longleftrightarrow} \Delta). 
		\end{aligned}
	\end{equation}
	where $0 \overset{\text{open}}{\longleftrightarrow} \Delta$ (respectively $0 \overset{\text{closed}}{\longleftrightarrow} \Delta$) refers to the existence of a path from $0$ to $\Delta$ in $G$, such that all vertices along this path are open (respectively closed; with the convention $0 \overset{\text{open}}{\longleftrightarrow} \Delta$ requiring the start and end point of such a path to be open, and the same for ``closed''). Since $p(\abs{\sigma'}) < \frac12$, monotonicity in $p$ for Bernoulli (site) percolation implies that
	\begin{equation}
		\begin{aligned}
			\P^G_{p(\abs{\sigma'})}(0 \overset{\text{open}}{\longleftrightarrow} \Delta)
			\leq\P^G_{1-p(\abs{\sigma'})}(0 \overset{\text{open}}{\longleftrightarrow} \Delta)
            = \P^G_{p(\abs{\sigma'})}(0 \overset{\text{closed}}{\longleftrightarrow} \Delta).
		\end{aligned}
	\end{equation}
	In particular, summing over all possible $\mathscr C_{B_N}$ from the decomposition \eqref{eq:decomposition_clusters_BN_in_Lambda} and using \eqref{eq:site_perc_open} and \eqref{eq:site_perc_closed}, 
	\begin{equation}
		\bar\nu_{\Lambda;\beta,\rho,h}^+\left(0 \overset{+}{\longleftrightarrow} \partial B_N \mid \mathcal A\right) \\
		\leq \bar\nu_{\Lambda;\beta,\rho,h}^+\left(0 \overset{\leq 0}{\longleftrightarrow} \partial B_N \mid \mathcal A\right). 
	\end{equation}
    Together with \eqref{eq:probability_no_12_circuit} this yields
    \begin{equation}
        \nu_{\Lambda;\beta,\rho,h}^+\left(0 \overset{+}{\longleftrightarrow} \partial B_N\right)
        \leq \nu_{\Lambda;\beta,\rho,h}^+\left(0 \overset{\leq 0}{\longleftrightarrow} \partial B_N\right)
        + c_1 N^{d-1} e^{-c_2 N}.
    \end{equation}
    The claim \eqref{eq:s-12-ising_minus_infty_cluster_dominate_plus} follows by letting $\Lambda \uparrow \Z^d$ first and $N \to \infty$ second. 
\end{proof}

\begin{proposition}\label{thm:s-12-ising_no_minority_percolation}
    Assume that $d=2$. Let $\beta > 0$, $\rho \in \R$ and $h \colon \set{-S+1,\dots,S} \to \R$ such that for all $s \in \set{1, \dots, S}$
    \begin{equation}\label{eq:cond_h_rho}
        h(s)-h(1-s)+\rho(2s-1) > 0.
    \end{equation}
    Then
    \begin{equation}\label{eq:s-12-ising_no_minority_percolation}
        \nu_{\beta,\rho, h}^+ \left(0\overset{+}{\longleftrightarrow} \infty\right) = 0
    \end{equation}
\end{proposition}
\begin{proof}
    Since the proof is a classic argument in two dimensional percolation, we only give a sketch here.
    First, by ergodicity of the measure $\nu_{\beta,\rho,h}^+$ we have that
    \begin{equation}
        \nu_{\beta,\rho,h}^+\left(\exists \text{infinite cluster of spins $\geq 1$}\right), 
        \,\nu_{\beta,\rho,h}^+\left(\exists \text{infinite cluster of spins $\leq 0$}\right) \in \set{0,1}.
    \end{equation}
    Thus, if we assume that 
    \begin{equation}
        \nu_{\beta,\rho, h}^+ \left(0\overset{+}{\longleftrightarrow} \infty\right) > 0,
    \end{equation}
    there exists an infinite cluster of spins $\geq 1$ $\nu_{\beta,\rho, h}^+$-almost surely. By \zcref{thm:s-12-ising_minus_infty_cluster_dominate_plus} there thus also exists an infinite cluster of spins $\leq 0$ $\nu_{\beta,\rho, h}^+$-almost surely.

    Second, the Burton--Keane argument (cf.~\cite{Burton:DensityUniquenessPercolation1989,Grimmett:Percolation1999}) adapts to our setting and shows that in fact both the $\geq 1$ and $\leq 0$ cluster are unique.

    Third, Zhang's argument (cf.~\cite{Grimmett:Percolation1999, Bollobas:PercolationDualLattices2008}) shows that this is impossible in two dimensions. Roughly, there is $N \in \N$ large enough, such that the event that the unique infinite $\geq 1$ cluster touches the left and right boundaries of the box $B_N$, and the unique infinite $\leq 0$ cluster touches the top and bottom boundaries of $B_N$, has a positive probability. However, this is geometrically impossible in two dimensions, because otherwise the clusters would need to intersect.
\end{proof}

A key lemma to proving exponential decay is the following.
\begin{lemma}[{\cite[Lemma 3.1]{Duminil-Copin:SharpPhaseTransition2019}}]
\label{thm:exponential_decay_differential_inequality}
    Consider a pointwise converging sequence of increasing differentiable functions $f_n \colon  [0, T_0] \to [0, M]$ satisfying
    \begin{equation}
        f_n' \geq \frac{n}{\Sigma_n} f_n
    \end{equation}
    for all $n \geq 1$, where $\Sigma_n := \sum_{k=0}^{n-1} f_k$. Then there exists $T_1 \in [0,T_0]$ such that
    \begin{enumerate}[label=(P\arabic*)]
        \item for any $t < T_1$, there exists $c_t > 0$ such that $f_n(t) \leq M\expa{-c_t n}$ for all $n \in \N$.\label{item:exponential_diff_inequality_exp_decay}
        \item for any $t > T_1$, $f = \lim_{n \to \infty} f_n$ satisfies $f(t) \geq t - T_1$.
    \end{enumerate}
\end{lemma}

In our application of the lemma above, the functions $f_n$ will be probabilities of the events $0 \overset{+}{\longleftrightarrow} \partial B_n$, where the argument will indicate the strength of the magnetic field. Since we have seen that there is no percolation for any positive magnetic field (cf.~\zcref{thm:s-12-ising_no_minority_percolation}), we are always `subcritical', so that \zcref[noname]{item:exponential_diff_inequality_exp_decay} will imply exponential decay. Next, we define the probability measure we consider for the definition of the $f_n$.

\begin{definition}
    We define the following probability measure on $\set{0,1}^{\Lambda}$:
    \begin{equation}
        \hat\nu_{\Lambda;\beta,\rho,h}^+(\kappa)
        := \nu_{\Lambda;\beta,\rho,h}^+(\pi(\sigma) = \kappa),
        \quad \kappa \in \set{0,1}^{\Lambda},
    \end{equation}
    where for $i \in \Lambda$ we defined $\pi(\sigma)_i = 1$ if $\sigma_i \geq 1$ and $\pi(\sigma)_i=0$ otherwise. 
\end{definition}

\begin{lemma}\label{thm:covariance_bound_decision_tree}
    Let $\Lambda \Subset \Z^d$ connected and $N \in \N$ such that $B_N \subset \Lambda$. 
    Then
    \begin{equation}
        \sum_{x \in B_N} \widehat\Cov_{\Lambda;\beta,\rho,h}^+(\I_{0 \overset{+}{\leftrightarrow}\partial B_N}, \,\kappa_x) 
        \geq \frac{N}{1+4d Q_N} \hat\nu_{\Lambda;\beta,\rho,h}^+\left(0 \overset{+}{\longleftrightarrow}\partial B_N\right) \left(1- \hat\nu_{\Lambda;\beta,\rho,h}^+\left(0\overset{+}{\longleftrightarrow}\partial B_N\right)\right),
    \end{equation}
    where
    \begin{equation}
        Q_N := \max_{x \in B_N} \sum_{k=0}^{N-1} \hat\nu_{\Lambda;\beta,\rho,h}^+\left(x \overset{+}{\longleftrightarrow}\partial B_k(x)\right),
    \end{equation}
    and where $\widehat\Cov_{\Lambda;\beta,\rho,h}^+$ denotes the covariance with respect to $\hat\nu_{\Lambda;\beta,\rho,h}^+$.
\end{lemma}
\begin{proof}
    The four functions theorem (cf.~\cite{Ahlswede:InequalityWeightsTwo1978}) implies that $\hat\nu_{\Lambda;\beta,\rho,h}^+$ satisfies the lattice FKG condition and is therefore \emph{monotonic} (in the sense of \cite{Duminil-Copin:SharpPhaseTransition2019}). The proof of~\cite[Lemma 3.2]{Duminil-Copin:SharpPhaseTransition2019} adapts almost verbatim to site percolation measures. Precisely, we use the same argument as in~\cite[Lemma 8.2]{Gunaratnam:ExistenceTricriticalPoint2024a} to establish the claim.
\end{proof}

\begin{proof}[Proof of \zcref{thm:main_two_dimensions}]
    By \zcref{thm:equivalence_no_percolation_uniqueness,thm:two_point_correlation_relation_percolation,thm:two_point_correlations_dominate}, it suffices to show that 
    \begin{equation}\label{eq:2d_exp_decay_open_cluster_size}
        \bar\mu_{\beta,\rho}^+\left(0 \overset{\text{open}}{\longleftrightarrow} \partial B_N\right) \leq c_1 e^{-c_2 N}
    \end{equation}
    holds for positive constants independent of $N$. 
    By \zcref{thm:bound_clustersize_open_by_plus} and \zcref{thm:dominate_M_by_M12} there exists $\lambda_0(\rho)>0$ independent of $N$, such that for all $\Lambda \Subset \Z^d$ with $B_{2N} \subset \Lambda$
    \begin{equation}
        \bar\mu_{\Lambda;\beta,\rho}^+\left(0 \overset{\text{open}}{\longleftrightarrow} \partial B_N\right) 
        \leq 2\nu_{\Lambda;\beta,\rho, \lambda_0 H}^+\left(0 \overset{+}{\longleftrightarrow} \partial B_N\right)
        \leq 2\nu_{B_{2N};\beta,\rho, \lambda_0 H}^+\left(0 \overset{+}{\longleftrightarrow} \partial B_N\right),
    \end{equation}
    where $H(k) = \I_{k \geq 1}$ and where the second inequality follows by FKG.
    Taking $\Lambda \uparrow \Z^d$ yields
    \begin{equation}\label{eq:bound_connectivity_by_hat_nu}
        \bar\mu_{\beta,\rho}^+\left(0 \overset{\text{open}}{\longleftrightarrow} \partial B_N\right)
        \leq 2\nu_{B_{2N};\beta,\rho, \lambda_0 H}^+\left(0 \overset{+}{\longleftrightarrow} \partial B_N\right).
    \end{equation}
    Let 
    \begin{equation}\label{eq:defJ}
        J(\rho) := \lambda_0(\rho) + \rho(2S-1)
    \end{equation} 
    We set $\rho_* := \sup \set{\bar\rho \in (0, \infty) \mid J(\rho) > 0 \text{ for all } \rho \in (-\rho_*,\infty)} > 0$, because $J$ is continuous (since $\lambda_0$ depends continuously on $\rho$) and $J([0,\infty)) > 0$. 
    We show that the right hand side of \eqref{eq:bound_connectivity_by_hat_nu} is bounded by $e^{-c_{\lambda_0} N}$ for all $\rho \in (-\rho_*, \infty)$.
    
    For $\lambda \in [0, 2\lambda_0]$ define
    \begin{equation}
        \theta_N(\lambda) := \hat\nu_{B_{2N};\beta,\rho,(2\lambda_0-\lambda) H}^+\left(0 \overset{+}{\longleftrightarrow} \partial B_N\right)
        \quad\text{and}\quad
        \Sigma_N(\lambda) := \sum_{k=0}^{N-1} \theta_k(\lambda).
    \end{equation}%
    A computation shows (for $N \geq 1$)
    \begin{equation}
        \theta_N'(\lambda) 
        = \sum_{x \in B_{2N}} \widehat\Cov_{B_{2N};\beta,\rho,(2\lambda_0-\lambda)H}\left(\I_{0 \overset{+}{\leftrightarrow} \partial B_N}, \kappa_x\right)
        \geq \sum_{x \in B_N} \widehat\Cov_{B_{2N};\beta,\rho,(2\lambda_0-\lambda)H}\left(\I_{0 \overset{+}{\leftrightarrow} \partial B_N}, \kappa_x\right),
    \end{equation}
    by FKG, since $\I_{0 \overset{+}{\longleftrightarrow} \partial B_N}$ and $\kappa_x$ are increasing functions.
    Using \zcref{thm:covariance_bound_decision_tree} we obtain
    \begin{equation}
        \theta_N'(\lambda) 
        \geq \frac{N}{1+4d \max_{x \in B_N} \sum_{k=0}^{N-1} \hat\nu_{B_{2N};\beta,\rho,(2\lambda_0-\lambda)H}^+\left(x \overset{+}{\leftrightarrow} \partial B_k(x)\right)} \theta_N(\lambda)\left(1- \theta_N(\lambda)\right).
    \end{equation}
    By FKG, we have for $x \in B_N$
    \begin{equation}
    \begin{split}
        &\phntm\sum_{k=0}^{N-1} \hat\nu_{B_{2N};\beta,\rho,(2\lambda_0-\lambda)H}^+\left(x \overset{+}{\longleftrightarrow}\partial B_k(x)\right) \\
        &\leq 2 \sum_{k=0}^{N/2} \nu_{B_{2N};\beta,\rho,(2\lambda_0-\lambda)H}^+\left(x \overset{+}{\longleftrightarrow}\partial B_k(x) \mid \sigma_j=S \text{ for all } j \in B_{2N} \setminus B_{2k}(x)\right) \\
        &= 2 \sum_{k=0}^{N/2} \hat\nu_{B_{2k};\beta,\rho,(2\lambda_0-\lambda)H}^+\left(0 \overset{+}{\longleftrightarrow}\partial B_k\right) \\
        &\leq 2 \Sigma_N(\lambda).
    \end{split}
    \end{equation}
    Moreover,
    \begin{equation}
        \Sigma_N(\lambda) \geq \theta_0(\lambda) \geq \theta_0(0) >0.
    \end{equation}
    Therefore,
    \begin{equation}
        \theta_N'(\lambda) 
        \geq \frac{N}{(8d + \theta_0^{-1}(0)) \Sigma_N(\lambda)} \theta_N(\lambda)\left(1- \theta_N(\lambda)\right)
        \geq \frac{c_0 N}{\Sigma_N(\lambda)} \theta_N(\lambda),
    \end{equation}
    where $c_0 = \frac1{8d + \theta_0^{-1}(0)} (1-\theta_1(2\lambda_0))>0$. Thus, $f_N(\lambda) := \theta_N(\lambda)/c_0$ satisfies the differential inequality of \zcref{thm:exponential_decay_differential_inequality}. 
    
    We claim that there exists $\delta > 0$ such that $f_N(\lambda) \to 0$ as $N \to \infty$ for all $\lambda \in (0,\lambda_0+\delta)$. Then \zcref{thm:exponential_decay_differential_inequality} implies that there exists a constant $c_{\lambda_0} > 0$ such that
    \begin{equation}
        \theta_N(\lambda_0)=c_0 f_N(\lambda_0) \leq e^{-c_{\lambda_0}N},
    \end{equation}
    yielding \eqref{eq:2d_exp_decay_open_cluster_size} via \eqref{eq:bound_connectivity_by_hat_nu}.

    Fix $\lambda \in [0,2\lambda_0]$. Let $\varepsilon > 0$ and $N_\eps \in \N$ such that
    \begin{equation}
        \nu_{\beta,\rho,(2\lambda_0 -\lambda)H}^+\left(0 \overset{+}{\longleftrightarrow}\partial B_{N_\eps}\right)
        \leq \nu_{\beta,\rho,(2\lambda_0 -\lambda)H}^+\left(0 \overset{+}{\longleftrightarrow} \infty\right) + \eps.
    \end{equation}
    Then,
    \begin{equation}
        \lim_{N \to \infty} c_0f_N(\lambda) 
        \leq \lim_{N \to \infty} \nu_{B_{2N};\beta,\rho,(2\lambda_0 -\lambda)H}^+\left(0 \overset{+}{\longleftrightarrow} \partial B_{N_\eps}\right)
        \leq \nu_{\beta,\rho,(2\lambda_0 -\lambda)H}^+\left(0 \overset{+}{\longleftrightarrow} \infty\right) + \eps.
    \end{equation}
    By \zcref{thm:s-12-ising_no_minority_percolation} $\nu_{\beta,\rho,(2\lambda_0 -\lambda)H}^+\left(0 \overset{+}{\longleftrightarrow} \infty\right)=0$ whenever \eqref{eq:cond_h_rho} holds for $h=(2\lambda_0 -\lambda)H$. When $\rho \geq 0$ it holds for all $\lambda \in (0,2\lambda_0)$ so we can set $\delta = \lambda_0$. 
    If $\rho \in (-\rho_*, 0)$, we use that $J(\rho)>0$, where $J$ was defined in \eqref{eq:defJ}. Thus \eqref{eq:cond_h_rho} holds true for some $\delta > 0$ small enough. Therefore indeed $f_N(\lambda) \to 0$ as $N \to \infty$ for all $\lambda \in (0, \lambda_0 + \delta)$.
\end{proof}

\begin{remark}
    We remark that the exponential decay of
    \begin{equation}
        \bar\mu_{B_{2N};\beta,\rho}^+\left(0 \overset{\text{open}}{\longleftrightarrow} \partial B_N\right) \leq e^{-cN}
    \end{equation}
    implies the finite volume estimates
    \begin{equation}
        \mu_{B_{2N};\beta,\rho}^+\left(\sigma_0\right) \leq 2Se^{-cN}
        \qquad\text{and}\qquad
        \abs{\mu_{B_{2N};\beta,\rho}^+\left(\sigma_0 \sigma_x\right) - \mu_{B_{2N};\beta,\rho}^+\left(\sigma_0\right)\mu_{B_{2N};\beta,\rho}^+\left(\sigma_x\right)} \leq 2S^2 e^{-cN}
    \end{equation}
    for all $\abs{x} \geq N$, as can be seen from \zcref{thm:two_point_correlation_relation_percolation} and the proof of \zcref{thm:equivalence_no_percolation_uniqueness}.
\end{remark}

\subsection{The general case at low temperatures}
\label{sec:minority_percolation_low_temperature}

At low temperatures, we can strengthen the statement of \zcref{thm:main_low_temperature} to finite volumes. This follows by proving that minority spins do not percolate in the measure 
$\nu_{\Lambda;\beta,\rho,h}^+$ for low enough temperatures, because even the smallest magnetic field strongly favors spins $\leq 0$. The proof is based on a Peierls-type argument and the main difficulty is controlling the influence of the boundary conditions, which work in the opposite direction of the magnetic field. This is done using \zcref{thm:>12_cluster_exponential_decay}.

\begin{proposition}\label{thm:no_minority_percolation_low_temp}
    Let $\beta > 0$, $\rho \in \R$ and $h \colon \set{-S+1,\dots,S} \to \R$ such that for all $s \in \set{1, \dots, S}$
    \begin{equation}\label{eq:cond_h_rho2}
        h(s)-h(1-s)+\rho(2s-1) > 0.
    \end{equation}
    Moreover, assume $W(1)>W(0)=0$.
    There exist constants
    $c_1 \equiv c_1(S, h, \rho, W, d, \beta) > 0$ and
    $c_2 \equiv c_2(S, h, \rho) > 0$ 
    such that the following holds:
    
    If $\beta > \frac{4(\ln d+\ln 8)}{W(1) d}$ and $N,M \in \N$, then
    \begin{equation}\label{eq:exp_decay_minority_percolation_low_temperature}
        \nu_{B_{N+2M};\beta,\rho,h}^+ \left(0 \overset{+}{\longleftrightarrow} \partial B_N\right)
        \leq c_1(N+M+1)^{d-1} e^{-c_2 M} + \frac{\lambda(\beta)^N}{1-\lambda(\beta)},
    \end{equation}
    where $\lambda(\beta) := e^{-\frac\beta2 W(1) +  \frac{2(\ln d+\ln 8)}d} < 1$. 
\end{proposition}
\begin{remark}
    We note that $\beta \geq \frac{c \ln d}d$ is the best order possible for this theorem. Indeed, there are rigorous results for the Ising model on a high dimensional lattice~\cite{Aizenman:PercolationMinoritySpins1987}, which show that in high enough dimensions both spins percolate in the regime $\beta \in (\frac cd, \frac{c\ln d}d)$. Existence of this phenomenon already in $d=3$ is supported e.g. by~\cite{Jiang:PercolationBothSigns2025}, where percolation of both spins is proved on a particular three dimensional lattice.
\end{remark}
\begin{proof}
    To simplify the notation let $\Lambda := B_{N+2M}$. For $\sigma \in \Sigma$ consider the configuration $\sigma^N$ obtained from $\sigma$ by changing $\sigma^N_i = -S+1$ for any $i \in \Z^d \setminus B_N$ and leaving $\sigma^N_i = \sigma_i$ for all $i \in B_N$. Moreover, let
    \begin{equation}
        \mathcal V_N^+(\sigma) := \set{x \in B_N \mid \sigma^N \in \set{0 \overset{+}{\longleftrightarrow}x}}
    \end{equation}
    be all points in $B_N$ which are connected to the origin by a path of spins $\geq 1$ which uses only vertices in $B_N$. Let
    \begin{equation}
        \Delta'(\sigma, \omega) := \bigcup_{x \in \mathcal V^+_N} C_x(\omega)
    \end{equation}
    where the union is to be understood as subset of $\Z^d$. In particular, since $\mathcal V^+_N$ is connected, so is $\Delta'$. 

    If $\omega \in \set{0,1}^{\E(\Z^d)}$ such that $\omega_e = 1$ for all $e \in \E(\Z^d) \setminus \E^b(\Lambda)$, then $\Delta'(\sigma,\omega)$ can be infinite, but there are at most finitely many finite connected components $C_1(\sigma, \omega), \dots, C_{m(\omega)}(\sigma, \omega)$ of $\Z^d \setminus \Delta'(\sigma, \omega)$, because every vertex of $\Lambda^c$ belongs to the same infinite component of $\omega$. We define
    \begin{equation}
        \Delta(\sigma, \omega) := \Delta'(\sigma, \omega) \cup \left(\bigcup_{k=1}^{m(\omega)} C_k(\sigma,\omega)\right).
    \end{equation}
    Note that $\Delta$ is $\Delta'$ with all ``holes'' filled and that we have the following relation
    \begin{equation}
        \bar\nu_{\Lambda;\beta,\rho,h}^+ \left(0 \overset{+}{\longleftrightarrow} \partial B_N\right)
        = \bar\nu_{\Lambda;\beta,\rho,h}^+ \left(\Delta \cap \partial B_N \neq \emptyset\right).
    \end{equation}
    We define the following two events:
    \begin{equation}
        \mathcal A_1 := \set{\Delta(\sigma, \omega) \subset B_{N+M}} \in  \mathcal F \otimes \mathcal G
    \end{equation}
    and
    \begin{equation}
        \mathcal A_2 := \set{\forall x \in \partial B_{N+M+1} \colon C_x(\omega) \cap \Lambda^c = \emptyset} \in \mathcal G.
    \end{equation}
    Using \zcref{thm:>12_cluster_exponential_decay}, we estimate
    \begin{equation}\label{eq:prob_A1}
        \bar\nu_{\Lambda;\beta,\rho,h}^+(\mathcal A_1^c)
        \leq \bar \nu_{\Lambda;\beta,\rho,h}^+\left(\exists x \in \partial B_N \colon \sigma_x \geq 1 \text{ and } x \overset{\text{open}}{\longleftrightarrow} \partial B_{N+M} \right)
        \leq c_1 N^{d-1} e^{-c_2M}
    \end{equation}
    and
    \begin{equation}\label{eq:prob_A2}
        \bar\nu_{\Lambda;\beta,\rho,h}^+ \left(\mathcal A_2^c\right)
        \leq \bar \nu_{\Lambda;\beta,\rho,h}^+\left(\exists x \in \partial B_{N+M+1} \colon \sigma_x \geq 1 \text{ and } x \overset{\text{open}}{\longleftrightarrow} \partial \Lambda \right)
        \leq c_1 (N+M+1)^{d-1} e^{-c_2 M},
    \end{equation}
    as in the proof of \zcref{thm:s-12-ising_minus_infty_cluster_dominate_plus}.
    We remark that by definition of $\Delta$, the set $\Delta^c$ is connected on the event $\mathcal A_1$. Moreover, we claim that $\Delta$ has at least $N$ boundary edges on the event $\mathcal A_1 \cap \set{\Delta \cap \partial B_N \neq \emptyset}$. To see this, let $z \in \Delta \cap \partial B_N$ and assume without loss of generality $z_1=N$. Consider the hyperplanes $H_k := \set{y \in \Z^d \colon y_1 = k}$. Since $\Delta$ is connected, we must have $\abs{H_k \cap \Delta} \geq 1$ for all $k \in \set{0, \dots, N}$. Finally, since $\Delta$ is finite on $\mathcal A_1$, there exist $u_k, v_k \in H_k$ such that $u_k \in \Delta$, $v_k \notin \Delta$ and $\set{u_k, v_k} \in \E(\Z^d)$. In particular, $\set{u_k, v_k} \in \partial_{\exterior} \Delta$.
    
    We can therefore write
    \begin{equation}
    \begin{aligned}
        \bar\nu_{\Lambda;\beta,\rho,h}^+ \left(\Delta \cap \partial B_N \neq \emptyset \text{ and } A_1 \text{ and } \mathcal A_2\right)
        &\leq \sum_{k=N}^\infty \bar\nu_{\Lambda;\beta,\rho,h}^+ \left(\abs{\partial_{\exterior} \Delta} = k \text{ and } \mathcal A_1 \text{ and } \mathcal A_2\right) \\
        &= \sum_{k=N}^\infty \sum_{\substack{\abs{\partial_{\exterior} V} = k}} \bar\nu_{\Lambda;\beta,\rho,h}^+ \left(\Delta = V \text{ and } A_2\right),
    \end{aligned}
    \end{equation}
    where the second sum is over all $V \subset B_{N+M}$ connected such that $0 \in V$, $\abs{\partial_{\exterior} V} = k$ and such that $V^c$ is connected, using that we were on the event $\mathcal A_1$. We claim that for such $V$
    \begin{equation}\label{eq:open_contours_exponentially_unlikely}
        \bar\nu_{\Lambda;\beta,\rho,h}^+ \left(\Delta = V \mid \mathcal A_2\right) 
        \leq e^{-\frac\beta2 W(1) k}.
    \end{equation}
    This implies
    \begin{equation}\label{eq:probability_delta_with_fixed_number_of_edges}
        \bar\nu_{\Lambda;\beta,\rho,h}^+ \left(\abs{\partial_{\exterior} \Delta} = k \mid \mathcal A_2\right)
        \leq e^{-\frac\beta2 W(1) k} \mathcal S(k),
    \end{equation}
    where $\mathcal S(k)$ is the number of connected subgraphs of $\Z^d$ which contain the origin, have exactly $k$ boundary edges, and whose complement is connected.
    It is known~\cite[Theorem 6]{Balister:CountingRegionsBounded2007} that $\mathcal S(k) \leq \expa{2\frac{\ln d+\ln 8}dk}$ when $d \geq 2$ (see also~\cite{Lebowitz:ImprovedPeierlsArgument1998}). We remark that the authors use the definition of rooted $d$-complexes with primitive boundary in their theorem. However, each such complex is in unique correspondence with a finite, connected subgraph of $\Z^d$ which contains the origin and whose complement is connected (cf.~\cite[Appendix A]{Bassan:NonconstantGroundConfigurations2025}).

    Together with \eqref{eq:probability_delta_with_fixed_number_of_edges} we obtain
    \begin{equation}\label{eq:+_connectivity_good_event}
        \bar\nu_{\Lambda;\beta,\rho,h}^+ \left(0 \overset{+}{\longleftrightarrow} \partial B_N \mid \mathcal A_2\right)
        \leq \frac{\lambda(\beta)^N}{1-\lambda(\beta)},
    \end{equation}
    where we set $\lambda(\beta) := e^{-\frac\beta2 W(1) + \frac{2(\ln d+\ln8)}d}$. 
    Combined with \eqref{eq:prob_A1} and \eqref{eq:prob_A2} this gives the desired claim \eqref{eq:exp_decay_minority_percolation_low_temperature} whenever $\lambda(\beta) < 1$.

    It remains to show \eqref{eq:open_contours_exponentially_unlikely}. 
    Let $V \subset B_{N+M}$ connected such that $\abs{\partial_{\exterior} V} = k$. Enumerate the edges of $\partial_{\exterior} V \equiv \set{e_1, \dots, e_k} \subset \mathcal E(B_{N+M})$. Then,
    \begin{equation}\label{eq:bound_prob_Delta=V_by_k_closed_edges}
        \bar\nu_{\Lambda;\beta,\rho,h}^+ \left(\Delta = V \mid \mathcal A_2\right)
        \leq \bar\nu_{\Lambda;\beta,\rho,h}^+ \left(\omega_{e_i} = 0 \text{ for all $i\in \set{1, \dots, k}$}  \mid \mathcal A_2\right),
    \end{equation}
    where we used that by construction of $\Delta$ any edge in $\partial_{\exterior}\Delta$ is closed.
    For $i \in \Z^d$ define $\sigma'_i := \sigma_i - \frac12$. Then the contribution of an edge $\set{i,j} \in \E^b(\Lambda)$ to the measure $\bar\nu_{\Lambda;\beta,\rho,h}^+$ is 
    \begin{equation}\label{eq:rewrite_single_edge_contribution_extended_ising}
    \begin{aligned}
         e^{-\beta W(\sigma_i - \sigma_j)}q(\sigma_i, \sigma_j)^{\omega_{ij}}(1 - q(\sigma_i, \sigma_j))^{1-\omega_{ij}}
        = e^{-\beta W(\abs{\sigma_i'} + \abs{\sigma_j'})}\left(\gamma(\abs{\sigma'_i}, \abs{\sigma'_j})\I_{\sign \sigma_i' = \sign\sigma_j'} \right)^{\omega_{ij}},
    \end{aligned}
    \end{equation}
    where
    \begin{equation}
        \gamma(a, b) := e^{\beta W(a+b) - \beta W(a-b)}-1.
    \end{equation}
    From \eqref{eq:bound_prob_Delta=V_by_k_closed_edges} we have
    \begin{equation}
    \begin{aligned}
        \bar\nu_{\Lambda;\beta,\rho,h}^+ \left(\Delta = V \mid \mathcal A_2\right)
        \leq \bar\nu_{\Lambda;\beta,\rho,h}^+ \left(\bar\nu_{\Lambda;\beta,\rho,h}^+\left(\omega_{e_1} = 0\mid \omega_{\langle e_1 \rangle}, \,\abs{\sigma'}, \,\mathcal A_2\right) \prod_{i=2}^k \I_{\omega_{e_i} = 0} \mid \mathcal A_2\right),
    \end{aligned}
    \end{equation}
    where we recall that $\omega_{\langle e_1\rangle}$ is the restriction of $\omega$ to $\mathcal E(\Z^d) \setminus \set{e_1}$.    
    We claim that 
    \begin{equation}\label{eq:condition_expectation_closed_edge_exponentially_small}
        \bar\nu_{\Lambda;\beta,\rho,h}^+\left(\omega_{e_1} = 0\mid \omega_{\langle e_1 \rangle},\,\abs{\sigma'},\,\mathcal A_2\right)
        \leq e^{-\frac\beta2 W(1)},
    \end{equation}
    which implies the claim \eqref{eq:open_contours_exponentially_unlikely} inductively.
    For $\omega \in \set{0,1}^{\mathcal E(\Z^d)}$ define two configurations $\omega^0$, $\omega^1$ by
    \begin{equation}
        \omega^0_e := \begin{cases}
            0 & \text{if } e=e_1,\\
            \omega_e & \text{else},
        \end{cases}
        \qquad\text{and}\qquad
        \omega^1_e := \begin{cases}
            1 & \text{if } e=e_1,\\
            \omega_e & \text{else}.
        \end{cases}
    \end{equation}
    If $e_1 = \set{u,v}$, then $u$ and $v$ belong to the same cluster in $\omega^1$, i.e., $C_{u}(\omega^1) = C_{v}(\omega^1)$.
    For a set $C \subset \Lambda$, $\kappa \in \set{\pm}$ and $\sigma \in \Sigma$ define
    \begin{equation}
        H_\kappa(\abs{\sigma'}; C) := \prod_{i \in C} e^{-h\left(\kappa \abs{\sigma'_i} + \frac12\right)-\rho\left(\kappa \abs{\sigma'_i} + \frac12\right)^2}.
    \end{equation}
    Since $\abs{\sigma'}$ will always be fixed, we drop it from the notation of $H$ in what follows. We obtain
    \begin{equation}
    \begin{aligned}
        &\phntm\bar\nu_{\Lambda;\beta,\rho,h}^+\left(\omega_{e_1} = 1\mid \omega_{\langle e_1 \rangle}, \,\abs{\sigma'}, \,\mathcal A_2\right)\\
        &= \frac{
            \gamma(\abs{\sigma'_{u}}, \abs{\sigma'_{v}}) 
            \prod_{C \in \mathscr C_{\Lambda}(\omega^1)} \sum_{\substack{\kappa \in \set{\pm}\\\kappa = + \text{ if } C \cap \Lambda^c \neq \emptyset}} H_\kappa(C)
        }{
            \gamma(\abs{\sigma'_{u}}, \abs{\sigma'_{v}}) 
            \prod_{C \in \mathscr C_{\Lambda}(\omega^1)} \sum_{\substack{\kappa \in \set{\pm}\\\kappa = + \text{ if } C \cap \Lambda^c \neq \emptyset}} H_\kappa(C)
            +
            \prod_{C \in \mathscr C_{\Lambda}(\omega^0)} \sum_{\substack{\kappa \in \set{\pm}\\\kappa = + \text{ if } C \cap \Lambda^c \neq \emptyset}} H_\kappa(C)
        }.
    \end{aligned}
    \end{equation}
    If $u$ and $v$ belong to the same cluster in $\omega^0$, then all clusters in the two configurations are equal and the expression simplifies to
    \begin{equation}
        \bar\nu_{\Lambda;\beta,\rho,h}^+\left(\omega_{e_1} = 1\mid \omega_{\langle e_1 \rangle}, \abs{\sigma'}, \,\mathcal A_2\right)
        = \frac{\gamma(\abs{\sigma'_{u}}, \abs{\sigma'_{v}})}{\gamma(\abs{\sigma'_{u}}, \abs{\sigma'_{v}})+1}.
    \end{equation}
    If on the other hand $u$ and $v$ do not belong to the same cluster in $\omega^0$, we still have that all other clusters are the same in $\omega^0$ and $\omega^1$ and that $C_{u}(\omega^1) = C_{u}(\omega^0) \cup C_{v}(\omega^0)$. On the event $\mathcal A_2$ it moreover holds $C_u(\omega^0) \cap \Lambda^c = C_v(\omega^0) \cap \Lambda^c = \emptyset$, because $u,v \in B_{N+M+1}$. Therefore, both clusters can be assigned either spins $\geq 1$ or spins $\leq 0$. In particular,
    \begin{equation}
    \begin{aligned}
        &\phntm\bar\nu_{\Lambda;\beta,\rho,h}^+\left(\omega_{e_1} = 1\mid \omega_{\langle e_1 \rangle}, \,\abs{\sigma'}, \,\mathcal A_2\right)\\
        &= \frac{
            \gamma(\abs{\sigma'_{u}}, \abs{\sigma'_{v}}) 
            \sum_{\kappa \in \set{\pm}} H_\kappa(C_{u}(\omega^1))
        }{
            \gamma(\abs{\sigma'_{u}}, \abs{\sigma'_{v}}) 
            \sum_{\kappa \in \set{\pm}} H_\kappa(C_{u}(\omega^1))
            +
            \sum_{\kappa_u, \kappa_v \in \set{\pm}} H_{\kappa_u}(C_{u}(\omega^0)) H_{\kappa_v}(C_{v}(\omega^0))
        }\\
        &= \frac{
            \gamma(\abs{\sigma'_{u}}, \abs{\sigma'_{v}}) 
            \sum_{\kappa \in \set{\pm}} H_\kappa(C_{u}(\omega^1))
        }{
            \left[\gamma(\abs{\sigma'_{u}}, \abs{\sigma'_{v}})+1\right]
            \sum_{\kappa \in \set{\pm}} H_\kappa(C_{u}(\omega^1))
            +
            \sum_{\kappa \in \set{\pm}} H_{\kappa}(C_{u}(\omega^0)) H_{-\kappa}(C_{v}(\omega^0))
        }.
    \end{aligned}
    \end{equation}
    By assumption on $h$ and $\rho$ we have
    \begin{equation}
    \begin{aligned}  
        \sum_{\kappa \in \set{\pm}} H_{\kappa}(C_{u}(\omega^0)) H_{-\kappa}(C_{v}(\omega^0))
        \leq \sum_{\kappa \in \set{\pm}} H_{\kappa}(C_{u}(\omega^0)) H_{\kappa}(C_{v}(\omega^0))
        = \sum_{\kappa \in \set{\pm}} H_{\kappa}(C_{u}(\omega^1))
    \end{aligned}
    \end{equation}
    and hence
    \begin{equation}
        \bar\nu_{\Lambda;\beta,\rho,h}^+\left(\omega_{e_1} = 1\mid \omega_{\langle e_1 \rangle}, \,\abs{\sigma'}, \, \mathcal A_2\right)
        \geq \frac{\gamma(\abs{\sigma'_{u}}, \abs{\sigma'_{v}})}{\gamma(\abs{\sigma'_{u}}, \abs{\sigma'_{v}})+2}.
    \end{equation}    
    By convexity of $W$ we have that
    \begin{equation}
        \gamma(\abs{\sigma'_{u}}, \abs{\sigma'_{v}}) \geq e^{\beta W(1)} - 1 > 0
    \end{equation}
    and therefore (in both cases considered above)
    \begin{equation}
        \bar\nu_{\Lambda;\beta,\rho,h}^+\left(\omega_{e_1} = 1\mid \omega_{\langle e_1 \rangle}, \,\abs{\sigma'}, \, \mathcal A_2\right)
        \geq \frac{e^{\beta W(1)}-1}{e^{\beta W(1)}+1}
        \geq 1-e^{-\frac\beta2 W(1)},
    \end{equation}
    where we used that $\frac{e^x-1}{e^x+1} \geq 1-e^{-x/2}$. This gives \eqref{eq:condition_expectation_closed_edge_exponentially_small} as claimed. 
\end{proof}

\begin{proof}[Proof of \zcref{thm:main_low_temperature}]
    By \zcref{thm:equivalence_no_percolation_uniqueness,thm:two_point_correlation_relation_percolation}, it suffices to show that 
    \begin{equation}
        \bar\mu_{\beta,\rho}^+\left(0 \overset{\text{open}}{\longleftrightarrow} \partial B_N\right) \leq c_1 e^{-c_2 N}
    \end{equation}
    holds for positive constants independent of $N$. 
    By \zcref{thm:bound_clustersize_open_by_plus} and \zcref{thm:dominate_M_by_M12} there exists $\lambda_0(\rho) > 0$ such that for all $\Lambda \Subset \Z^d$
    \begin{equation}
        \bar\mu_{\Lambda;\beta,\rho}^+\left(0 \overset{\text{open}}{\longleftrightarrow} \partial B_N\right) 
        \leq 2\nu_{\Lambda;\beta,\rho, \lambda_0 H}^+\left(0 \overset{+}{\longleftrightarrow} \partial B_N\right),
    \end{equation}
    with $H(k) := \I_{k \geq 1}$. Let 
    \begin{equation}
        J(\rho) := \lambda_0(\rho) + \rho(2S-1)
    \end{equation} 
    We set $\rho_* := \sup \set{\bar\rho \in (0, \infty) \mid J(\rho) > 0 \text{ for all } \rho \in (-\rho_*,\infty)} > 0$, because $J$ is continuous (since $\lambda_0$ depends continuously on $\rho$) and $J([0,\infty)) > 0$. 
    
    If $B_{3N} \subset \Lambda$, then by the FKG inequality
    \begin{equation}
        \nu_{\Lambda;\beta,\rho, \lambda_0 H}^+\left(0 \overset{+}{\longleftrightarrow} \partial B_N\right)
        \leq \nu_{B_{3N};\beta,\rho, \lambda_0 H}^+\left(0 \overset{+}{\longleftrightarrow} \partial B_N\right).
    \end{equation}
    If we verify \eqref{eq:cond_h_rho2} for $h = \lambda_0(\rho) H$, then \zcref{thm:no_minority_percolation_low_temp} yields that 
    \begin{equation}
        \bar\mu_{\Lambda;\beta,\rho}^+\left(0 \overset{\text{open}}{\longleftrightarrow} \partial B_N\right) 
        \leq \nu_{B_{3N};\beta,\rho, \lambda_0 H}^+\left(0 \overset{+}{\longleftrightarrow} \partial B_N\right)
        \leq c_1 e^{-c_2N}
    \end{equation}
    whenever $\beta > \frac{4(\ln d+\ln 8)}{W(1) d}$. Since the event on the left hand side is local, we can let $\Lambda \uparrow \Z^d$ to conclude. Indeed, \eqref{eq:cond_h_rho2} holds whenever $\rho \geq 0$. If $\rho \in (-\rho_*, 0)$, we use that $J(\rho) > 0$ which also implies \eqref{eq:cond_h_rho2}.
\end{proof}

\begin{remark}[Connection with percolation of finite clusters]\label{rem:percolation of finite clusters}

Grimmett--Holroyd--Kozma~\cite{Grimmett:PercolationFiniteClusters2014a} initiated the study of percolation of finite clusters in Bernoulli percolation $\P_p$ on $\Z^d$ with parameter $p$. Let $X$ be the complement of the (unique) infinite cluster, setting $X=\Z^d$ when no such cluster exists. They asked for which parameter $p$ does $X$ contain an infinite connected component (as a subset of $\Z^d$). 
Precisely, they defined the threshold
\begin{equation}
    p_{\text{fin}}(d):=\sup\{p\in[0,1]\colon \P_p(X\text{ has an infinite connected component})>0\}
\end{equation}
Necessarily, $p_{\text{fin}}(d)\ge p_c(d)$ (with $p_c(d)$ the Bernoulli percolation critical probability) and their main interest is whether strict inequality holds. They show this for $d\ge 19$, strengthened to $d\ge 11$ as a consequence of Fitzner--van der Hofstad~\cite{Fitzner:MeanfieldBehaviorNearestneighbor2017} and to $d\ge 10$ by Bock--Damron--Newman--Sidoravicius~\cite{Bock:PercolationFiniteClusters2020}. In~\cite[Remark 1.4]{Bock:PercolationFiniteClusters2020} the question of the asymptotics of $p_{\text{fin}}$ as $d\to\infty$ was raised, and it was shown that
\begin{equation}\label{eq:lower bound on p_fin}
    \liminf_{d\to\infty} \frac{p_{\text{fin}}(d)}{\frac{\ln d}{2d}}\ge 1
\end{equation}
(which also implies the strict inequality $p_{\text{fin}}(d)>p_c(d)$ in high dimensions, as $p_c(d)\sim\frac{1}{2d}$ as $d\to\infty$).
We make two remarks in connection with this.

First, Bernoulli percolation of parameter $p$ is known to be dominated by the random-cluster model with parameter $2p$ and cluster weight $q=2$~\cite[Theorem 3.21]{Grimmett:RandomClusterModel2006}, i.e., by the FK-Ising model with $2p = 1 - e^{-2\beta}$. Aizenman--Bricmont--Lebowitz~\cite{Aizenman:PercolationMinoritySpins1987} showed that there is simultaneous percolation of $+$ and $-$ spins when $\frac{1}{2d}+o(1/d)\le \beta\le \frac{\ln d}{4d} - \mathcal O(1/d)$ as $d\to\infty$. As it is simple to see that such simultaneous percolation can only occur when the union of the finite clusters (of the random-cluster model) is infinite, this gives an alternative proof, with a worse constant, of~\eqref{eq:lower bound on p_fin}.

Second, the techniques used in the proof of \zcref{thm:no_minority_percolation_low_temp} imply that
\begin{equation}
    p_{\text{fin}}(d)\le \frac{2(\ln d + \ln 8)}{d}
\end{equation}
for $d\ge 2$, making progress on~\cite[Remark 1.4]{Bock:PercolationFiniteClusters2020} by determining the order of magnitude of $p_{\text{fin}}(d)$ as $d\to\infty$ (determining the exact constant in the asymptotics is an interesting open problem). To see this, observe that if the origin lies in an infinite connected component of $X$, then for any $m$ there is a connected subset $A$ of $\Z^d$ containing the origin, whose complement is connected, having $\ge m$ boundary edges, all of which are closed in the percolation process (this subset will be a union of finite clusters of the percolation). The probability of this event tends to $0$ as $m\to\infty$, by a union bound over all such subsets $A$, as the probability of a specific $A$ with $k$ boundary edges is at most $(1-p)^k\le e^{-pk}$ (as its boundary is closed), while $\mathcal S(k)$, the number of such subsets with exactly $k$ boundary edges, satisfies $\mathcal S(k) \leq \expa{2(\frac{\ln d+\ln 8}d)k}$ (when $d \geq 2$) by~\cite[Theorem 6]{Balister:CountingRegionsBounded2007}.
\end{remark}

\subsection{Entropic repulsion for confined systems}
\label{sec:confined_systems}

Using the domination of $\mu_{\beta,\rho}^+$ by $\nu_{\beta,\rho,h}^+$, we can deduce the results on entropic repulsion \zcref{thm:constrained Gibbs measures}, \zcref{thm:constrained Gibbs measures even} and \zcref{thm:real_valued_entropic_repulsion} for arbitrary even and convex $W$.

\begin{proof}[Proof of \zcref{thm:constrained Gibbs measures}]
    Let $f \colon \Omega \to \R$ be a local, non-decreasing function and $N \in \N$. We will later take $f(\sigma) = \sigma_0$ and $f(\sigma) = \I_{\sigma_0 \geq 1}$. By the domination from \zcref{thm:dominate_M_by_M12} there exists $\lambda_0(\rho) > 0$ independent of $f$ and $N$ such that
    \begin{equation}
        \mu_{B_N;\beta,\rho}^+ (f) 
        \leq \nu_{B_N;\beta,\rho,\lambda_0 H}^+ (f)
    \end{equation}
    with $H(k) := \I_{k \geq 1}$. Let 
    \begin{equation}
        J(\rho) := \lambda_0(\rho) + \rho(2S-1)
    \end{equation} 
    We set $\rho_* := \sup \set{\bar\rho \in (0, \infty) \mid J(\rho) > 0 \text{ for all } \rho \in (-\rho_*,\infty)} > 0$, because $J$ is continuous (since $\lambda_0$ depends continuously on $\rho$) and $J([0,\infty)) > 0$. From now on, we assume $\rho \in (-\rho_*, \infty)$. 

    Let $\mathcal A_N := \set{0 \overset{\text{open}}{\longleftrightarrow} \partial B_N}$. Then
    \begin{equation}\label{eq:decomposition_prob_sigma0_geq1}
        \nu_{B_N;\beta,\rho,\lambda_0 H}^+(\sigma_0 \geq 1) 
        = \bar\nu_{B_N;\beta,\rho,\lambda_0 H}^+(\sigma_0 \geq 1 \text{ and } \mathcal A_N) + \bar\nu_{B_N;\beta,\rho,\lambda_0 H}^+(\sigma_0 \geq 1 \text{ and } \mathcal A_N^c).
    \end{equation}
    Recall that $C_0(\omega) \in \mathscr C(\omega)$ denotes the open cluster of $\omega$ containing the origin. On $\mathcal A_N^c$ we have $C_0 \subset B_N$. Thus,
    \begin{equation}
        \bar\nu_{B_N;\beta,\rho,\lambda_0 H}^+(\sigma_0 \geq 1 \text{ and } \mathcal A_N^c)
        = \sum_G \bar\nu_{B_N;\beta,\rho,\lambda_0 H}^+(\sigma_0 \geq 1 \text{ and } C_0 = G),
    \end{equation}
    where the sum is over all connected subgraphs $G$ of $B_N$ containing the origin. We perform a spin flip as in \eqref{eq:spin_flip_mu}, this time around $\frac12$, i.e.,
    \begin{equation}
        \sigma_i \mapsto \begin{cases}
            1 - \sigma_i & \text{if } i \in G,\\
            \sigma_i & \text{otherwise}.
        \end{cases}
    \end{equation}
    As in \eqref{eq:spin_flip_mu} this leaves the pairwise interactions invariant, however the magnetic field contributes a factor $e^{-\lambda_0 \abs{G}}$. Since $\abs{G} \geq 1$, we obtain
    \begin{equation}\label{eq:nu_sigma_0_geq1_bd_by_leq0}
        \bar\nu_{B_N;\beta,\rho,\lambda_0 H}^+(\sigma_0 \geq 1 \text{ and } \mathcal A_N^c)
        \leq e^{-(\lambda_0+\min\set{\rho, \rho(2S-1)})} \bar\nu_{B_N;\beta,\rho,\lambda_0 H}^+(\sigma_0 \leq 0 \text{ and } \mathcal A_N^c)
    \end{equation}
    and thus
    \begin{equation}
    \bar\nu_{B_N;\beta,\rho,\lambda_0 H}^+(\sigma_0 \geq 1 \mid \mathcal A_N^c)
        \leq \frac{1}{1+e^{\lambda_0+\min\set{\rho, \rho(2S-1)}}}.
    \end{equation}
    Since $\rho > -\rho_*$ we can apply \zcref{thm:>12_cluster_exponential_decay}. Combining this with \eqref{eq:decomposition_prob_sigma0_geq1}, we thus obtain 
    \begin{equation}\label{eq:bo_nu+_prob_sigma0_geq1}
        \nu_{B_N;\beta,\rho,\lambda_0 H}^+(\sigma_0 \geq 1)  \leq \frac{1}{1+e^{\lambda_0+\min\set{\rho, \rho(2S-1)}}} + c_1 e^{-c_2 N}
    \end{equation}
    for some positive constants $c_1$, $c_2$ independent of $N$.
    Similarly,
    \begin{equation}\label{eq:decomposition_expect_sigma0}
        \nu_{B_N;\beta,\rho,\lambda_0 H}^+(\sigma_0) 
        = \bar\nu_{B_N;\beta,\rho,\lambda_0 H}^+(\sigma_0 \I_{\mathcal A_N}) + \bar\nu_{B_N;\beta,\rho,\lambda_0 H}^+(\sigma_0 \I_{\mathcal A_N^c}).
    \end{equation}
    On one hand,
    \begin{equation}
        \bar\nu_{B_N;\beta,\rho,\lambda_0 H}^+(\sigma_0 \I_{\mathcal A_N})
        \leq S \bar\nu_{B_N;\beta,\rho,\lambda_0 H}^+(\sigma_0 \geq 1 \text{ and }{\mathcal A_N}) \leq Sc_1 e^{-c_2 N}.
    \end{equation}
    On the other hand, by the same spin-flip argument as above,
    \begin{equation}
    \begin{aligned}
        \bar\nu_{B_N;\beta,\rho,\lambda_0 H}^+\left(\sigma_0 \mid {\mathcal A_N^c}\right)
        &= \bar\nu_{B_N;\beta,\rho,\lambda_0 H}^+\left(\abs{\sigma_0-\tfrac12} \I_{\sigma_0 \geq 1} \mid {\mathcal A_N^c}\right)
        - \bar\nu_{B_N;\beta,\rho,\lambda_0 H}^+\left(\abs{\sigma_0-\tfrac12} \I_{\sigma_0 \leq 0} \mid {\mathcal A_N^c}\right) + \tfrac12\\
        &\leq \left(e^{-(\lambda_0+\min\set{\rho, \rho(2S-1)})}-1\right)\bar\nu_{B_N;\beta,\rho,\lambda_0 H}^+\left(\abs{\sigma_0-\tfrac12} \I_{\sigma_0 \leq 0} \mid {\mathcal A_N^c}\right) + \tfrac12\\
        &\leq \frac12\left(1+\tfrac{\left(e^{-(\lambda_0+\min\set{\rho, \rho(2S-1)})}-1\right)}{\left(e^{-(\lambda_0+\min\set{\rho, \rho(2S-1)})}+1\right)}\right)
        = \frac1{1+e^{\lambda_0+\min\set{\rho, \rho(2S-1)}}},
    \end{aligned}
    \end{equation}
    using that $\bar\nu_{B_N;\beta,\rho,\lambda_0 H}^+\left(\abs{\sigma_0-\tfrac12} \I_{\sigma_0 \leq 0} \mid {\mathcal A_N^c}\right) \geq \frac1{2(e^{-(\lambda_0+\min\set{\rho, \rho(2S-1)})}+1)}$ by \eqref{eq:nu_sigma_0_geq1_bd_by_leq0}. Therefore, inserting into \eqref{eq:decomposition_expect_sigma0}, we obtain
    \begin{equation}
        \nu_{B_N;\beta,\rho,\lambda_0 H}^+(\sigma_0) 
        \leq \frac1{1+e^{\lambda_0+\min\set{\rho, \rho(2S-1)}}} + Sc_1 e^{-c_2 N}.
    \end{equation}
    Let $\mu \in \mathscr G(\beta,\rho)$ be a Gibbs measure for the Hamiltonian \eqref{eq:def_hamiltonian}. The DLR condition \eqref{eq:DLR} together with FKG imply
    \begin{equation}
        \mu(\sigma_0 \geq 1) 
        = \int_{\Omega} \mu_{B_{N};\beta,\rho}^\eta (\sigma_0 \geq 1) \d \mu(\eta)
        \leq \mu_{B_{N};\beta,\rho}^+ (\sigma_0 \geq1) 
        \leq \nu_{B_{N};\beta,\rho,\lambda_0 H}^+ (\sigma_0 \geq1).
    \end{equation}
    Using \eqref{eq:bo_nu+_prob_sigma0_geq1} and taking the limit $N \to \infty$ this yields in particular $\mu(\sigma_0 \geq 1) < \frac12$, since $\lambda_0(\rho)+\min\set{\rho, \rho(2S-1)} >0$ by definition of $\rho_*$.
    Similarly, 
    \begin{equation}
        \mu(\sigma_0 \leq -1) 
        \leq \mu_{B_{2N};\beta,\rho}^- (\sigma_0 \leq -1)
        = \mu_{B_{2N};\beta,\rho}^+ \left(\sigma_0 \geq 1\right),
    \end{equation}
    where the last equality is due to a global spin flip. This implies $\mu(\sigma_0 \leq -1) < \frac12$ after taking the limit $N \to \infty$. A similar argument shows that $\abs{\mu(\sigma_0)} < \frac12$.
\end{proof}

We continue with the proof of entropic repulsion in the case $\sigma \colon \Z^d \to \set{-S+\frac12, \dots, S-\frac12}$.

\begin{proof}[Proof of \zcref{thm:constrained Gibbs measures even}]
    The proof is analogous to that of \zcref{thm:constrained Gibbs measures}. The only difference is that in the domination \zcref{thm:dominate_M_by_M12} we dominate by a measure with a magnetic field $H(k) := \I_{k > \tfrac12}$ supported on $\set{-S+1+\frac12, \dots, S-\frac12}$. We reflect again around $\frac12$, which now is an element of the spin space. However, this does not change the argument and an adaptation is straightforward.
\end{proof}

We conclude with the proof of entropic repulsion in the real-valued case.

\begin{proof}[Proof of \zcref{thm:real_valued_entropic_repulsion}]
    We now consider the model with Hamiltonian \eqref{eq:def_hamiltonian} with spins taking values in the interval $[-1,1]$ and $W \colon [-2,2] \to \R$ even and convex.
    Arguing as in \zcref{thm:unique_gibbs_measure_if_magnetization_zero}, uniqueness of Gibbs measures follows by showing that every Gibbs measure $\mu$ satisfies $\mu(\sigma_0)=0$.

    The proof of $\mu(\sigma_0)=0$ is similar to that of \zcref{thm:constrained Gibbs measures}, so we will leave some of the details to the reader.
    Let $\Lambda \Subset \Z^d$ and $\eps > 0$. Moreover, let $\mu_{\Lambda;\beta,\rho}^+$ be the measure on $[-1,1]^\Lambda$ with $+1$ boundary conditions and Hamiltonian \eqref{eq:def_finite_volume_hamiltonian} and $\nu_{\Lambda;\beta,\rho, h}^{+,\eps}$ be the probability measure on $[-1+2\eps, 1]^{\Lambda}$ with Hamiltonian \eqref{eq:def_hamiltonian_with_h} and boundary conditions $+1$. We claim that the stochastic domination
    \begin{equation}\label{eq:continuous_domination}
        \mu_{\Lambda;\beta,\rho}^+ \leq_{st} \nu_{\Lambda;\beta,\rho, \lambda_\eps H_\eps}^{+,\eps}
    \end{equation}
    holds with 
    \begin{equation}\label{eq:value of lambda eps}
        \lambda_\eps\equiv \lambda_\eps(\rho):=-\ln \mu_{\set{0};\beta,\rho}^{+} (\sigma_0 \geq -1+2\eps) > 0
    \end{equation}
    and $H_\eps(x) := \I_{x > \eps}$ (in fact, one may take any measurable $H_\eps\in[0,1]$).
    This is the continuous spin analog of \zcref{thm:dominate_M_by_M12}. It is obtained by repeating the proof of that proposition in our continuous setting (either directly, or by viewing the continuous setup as a limit of discrete ones).
    
    We extend the measure $\nu_{\Lambda;\beta,\rho,\lambda_\eps H_\eps}^{+,\eps}$ with boolean variables on the edges, and denote the new measure by $\bar\nu_{\Lambda;\beta,\rho, \lambda_\eps H_\eps}^{+,\eps}$. As in the discrete case, we chose the edge variables $\omega$ to be independent Bernoulli variables (depending on the environment), and we chose the probability of an edge $\set{i,j}\in \E^b(\Lambda)$ to be open to be
    \begin{equation}
        q_\eps(\sigma_i, \sigma_j) := 1 - e^{-\beta W(\abs{\sigma_i - \eps} + \abs{\sigma_j - \eps}) + \beta W(\sigma_i - \sigma_j)} \in [0,1).
    \end{equation}
    Assume that there exists $\rho_*\equiv\rho_*(\beta) > 0$ such that 
    \begin{equation}\label{eq:rho threshold}
    \text{$\lambda_\eps(\rho) + 4\rho(\eps s - \eps^2) > 0$ for all $\rho \in (-\rho_*, \infty)$ and $s \in (\eps, 1]$}
    \end{equation}
    (this is clearly satisfied for $\rho\ge0$; we show below that it extends to small negative $\rho$). Then one can rerun the proof of \zcref{thm:>12_cluster_exponential_decay} and obtain
    \begin{equation}
        \bar\nu_{\Lambda;\beta,\rho, \lambda_\eps H_\eps}^{+,\eps}\left(\sigma_0 > \eps \text{ and } 0 \overset{\text{open}}{\longleftrightarrow} \partial B_N\right) \leq c_1(\eps) e^{-c_2(\eps) N}
    \end{equation}
    for all $N\in \N$ such that $B_N \subset \Lambda$ and some positive constants $c_1(\eps)$, $c_2(\eps)$ depending on $\eps$ but not on $N$ or $\Lambda$. In particular, the spin flip argument from the proof of \zcref{thm:constrained Gibbs measures} (this time letting $\sigma_i \mapsto 2\eps - \sigma_i$ for $i \in C_0$) implies
    \begin{equation}
        \nu_{B_N;\beta,\rho, \lambda_\eps H_\eps}^{+,\eps}(\sigma_0) 
        \leq \eps + c_1(\eps) e^{-c_2(\eps) N}.
    \end{equation}
    Thus, using the stochastic domination and letting $N \to \infty$ one obtains that
    \begin{equation}
        \mu(\sigma_0) \leq \eps
    \end{equation}
    using the DLR condition and FKG as in the proof of \zcref{thm:constrained Gibbs measures}. Applying a symmetric argument, one deduces that $\abs{\mu(\sigma_0)} \leq \eps$. We conclude that $\mu(\sigma_0)=0$, as we wanted to show, since $\eps > 0$ was arbitrary.
    
    It remains to show that there exists $\rho_* > 0$ for which~\eqref{eq:rho threshold} holds. Using~\eqref{eq:value of lambda eps} and the fact that the inequality is clear for $\rho\ge0$, it suffices to show that
    \begin{equation}\label{eq:eps growth}
        \liminf_{\eps\downarrow0}\inf_{\rho\in[-1,0]}\frac{\lambda_\eps(\rho)}{\eps}>0.
    \end{equation}
    To this end, observe that for $\rho\in[-1,0]$,
    \begin{equation}
    \begin{split}
        \lambda_\eps(\rho)
        &= -\ln \mu_{\set{0};\beta,\rho}^{+} (\sigma_0 \geq -1+2\eps) = -\ln\left(1- \mu_{\set{0};\beta,\rho}^{+} (\sigma_0 < -1+2\eps)\right)\\
        &\ge \mu_{\set{0};\beta,\rho}^{+} (\sigma_0 < -1+2\eps)=\int_{-1}^{-1+2 \eps} f_{\set{0};\beta,\rho}^{+}(s)ds\ge 2\eps\inf_{\substack{\bar\rho\in[-1,0]\\s\in[-1,1]}}f_{\set{0};\beta,\bar\rho}^{+}(s),
    \end{split}
    \end{equation}
    where we denoted the density of $\mu_{\set{0};\beta,\rho}^{+}$ by $f_{\set{0};\beta,\rho}^{+}$.
    Then~\eqref{eq:eps growth} follows as $f_{\set{0};\beta,\rho}^{+}(s)$ is strictly positive and continuous in $(\rho,s)$ (so uniformly positive on the compact $[-1,0]\times[-1,1]$).
\end{proof}

\section{Three and more dimensions for mixtures of Gaussians}
\label{sec:three_dimensions}

In this section we prove \zcref{thm:main_mixtures_of_gaussians}. To prove uniqueness we adapt the strategy of~\cite[Section 8]{DAlimonte:FreeEnergyAnalyticity2026} to the discrete setting. For exponential decay of correlations we show that $\overline\mu_{\beta,\rho}$ is a subcritical percolation model, similar to the approach of~\cite[Section 8]{Gunaratnam:ExistenceTricriticalPoint2024a}. 

\subsection{Uniqueness of the infinite volume limit}
\label{sec:uniqueness_mixture_of_gaussians}

For $\beta > 0$ and $\rho \geq 0$ we define the following analog of the integer-valued GFF on $\Z^\Lambda$ with boundary conditions $\eta \in \Z^{\Z^d}$:
\begin{equation}\label{eq:def_IVGFF}
    \zeta_{\Lambda;\beta,\rho}^\eta(\sigma)
    \propto \expa{-\beta \sum_{\set{i,j} \in \E^b(\Lambda)} W(\sigma_i - \sigma_j) -  \rho \sum_{i \in \Lambda} \sigma_i^2},
    \qquad \sigma \in \Z^{\Z^d}, ~\sigma_i = \eta_i \text{ for } i \in \Lambda^c.
\end{equation}

\begin{lemma}\label{thm:exponential_brascamp_lieb}
    Let $\beta,\rho>0$. In addition to convexity, assume that $W\colon \R \to \R$ is defined for all real values and is a mixture of Gaussians. Then
    \begin{equation}\label{eq:integer_valued_field_MGF}
        \zeta_{\Lambda;\beta,\rho}^0\left(e^{v \cdot \sigma}\right)
        \leq \expa{\frac{\abs{v}^2}{2\rho}}.
    \end{equation}
\end{lemma}
\begin{proof}
     Let $\hat\zeta_{\Lambda;\beta,\rho}^\eta$ denote the finite volume Gibbs measure on $\R^\Lambda$ with density proportional to \eqref{eq:def_IVGFF} and with boundary condition $\eta$. Then there is a stochastic domination of moment generating functions
     \begin{equation}
         \zeta_{\Lambda;\beta,\rho}^0\left(e^{v \cdot \sigma}\right)
         \leq \hat\zeta_{\Lambda;\beta,\rho}^0\left(e^{v \cdot \sigma}\right),
     \end{equation}
     proved in~\cite[Lemma D.1, cf.~also the proof of Theorem 1.2]{Aizenman:DepinningIntegerrestrictedGaussian2022}. 

     Let
     \begin{equation}
         \psi(t) := \ln\left(\hat\zeta_{\Lambda;\beta,\rho}^0\left(e^{tv \cdot \sigma}\right)\right).
     \end{equation}
     Then
     \begin{equation}\label{eq:psi''_var}
         \psi''(t) = \Var_t(v \cdot \sigma),
     \end{equation}
     where $\Var_t$ is the variance with respect to the measure
     \begin{equation}
         \Phi_t(A) := \frac{\hat\zeta^0_{\Lambda;\beta,\rho}(\I_A e^{tv \cdot \sigma})}{\hat\zeta^0_{\Lambda;\beta,\rho}(e^{tv \cdot \sigma})}, \quad A \in \mathcal B(\R^\Lambda).
     \end{equation}
     The measure $\Phi_t$ has a density proportional to $e^{-H_t}$, where
     \begin{equation}
         H_t(\sigma) := \beta\sum_{\set{i,j} \in \E^b(\Lambda)} W(\sigma_i - \sigma_j) +  \rho \sum_{i \in \Lambda} \sigma_i^2 - tv \cdot \sigma.
     \end{equation}
     By convexity of $W$ it holds $H_t'' \geq \rho \Id$ as quadratic forms, where $H_t''$ is the Hessian of $H_t$. 
     Thus, applying the Brascamp--Lieb inequality \cite{brascamp1975some,Brascamp:ExtensionsBrunnMinkowskiPrekopaLeindler1976} in \eqref{eq:psi''_var}, we have
     \begin{equation}
         \psi''(t) \leq \frac{\abs{v}^2}{\rho}.
     \end{equation}
     Integrating twice (using symmetry of $\hat\zeta_{\Lambda;\beta,\rho}^0$) one obtains that
     \begin{equation}
         \psi(1) \leq \frac{\abs{v}^2}{2\rho}
     \end{equation}
     which yields the claim after taking the exponential on both sides.
\end{proof}

The magnetization of the integer valued field with a mass in a box of size $N$ decays to zero as $N \to \infty$.
\begin{proposition}\label{thm:decay_magnetization_IVGFF}
    Let $\beta > 0$, $\rho > 0$ and $S \in \N$. In addition to convexity and evenness, assume that $W \colon \R \to \R$ is defined for all real values and is a mixture of Gaussians and super-Gaussian.
    Define the magnetization in the box of size $N$ by $A_N := \abs{B_N}^{-1} \sum_{i \in B_N} \sigma_i$. Then
    \begin{equation}
        \zeta_{B_N;\beta,\rho}^S\left(\abs{A_N}\right) \to 0
        \quad\text{as } N \to \infty.
    \end{equation}
    Here the exponent $S$ denotes the constant boundary condition $\eta_i = S$ for all $i \in \Z^d$.
\end{proposition}
\begin{proof}
    For $t,u>0$ we have that
    \begin{equation}\label{eq:magnetization_split_markov}
        \zeta_{B_N;\beta,\rho}^S\left(\abs{A_N} \right)
        \leq t + \zeta_{B_N;\beta,\rho}^S\left(\abs{A_N} \I_{\abs{A_N} > t}\right)
        \leq t + \frac{e^{-ut}}{e u}\zeta_{B_N;\beta,\rho}^S\left(e^{2u \abs{A_N}}\right).
    \end{equation}
    We have,
    \begin{equation}
        \zeta_{B_N;\beta,\rho}^S\left(e^{2u \abs{A_N}}\right)
        =
        \frac{\zeta_{B_{N+1};\beta,\rho}^{0}\left(e^{2u \abs{A_N}} \prod_{i \in \partial B_{N+1}} \I_{\sigma_i = S}\right)}{\zeta_{B_{N+1};\beta,\rho}^{0}\left(\prod_{i \in \partial B_{N+1}} \I_{\sigma_i = S}\right)}.
    \end{equation}
    By \zcref{thm:exponential_brascamp_lieb} the numerator is bounded by
    \begin{equation}\label{eq:upper_bound_numerator_MGF_new_boundary_conditions}
        \zeta_{B_{N+1};\beta,  \rho}^{0}\left(e^{2u \abs{A_N}} \prod_{i \in \partial B_{N+1}} \I_{\sigma_i = S}\right)
        \leq 2 \expa{\frac{2u^2}{\rho \abs{B_N}}}.
    \end{equation}
    To bound the denominator, enumerate $\partial B_N = \set{x_1, \dots, x_K}$. Then, by the absolute value FKG inequality (cf.~\zcref{remark:FKG_for_convex_W})
    \begin{equation}
    \begin{aligned}\label{eq:lower_bound_denominator_MGF_new_boundary_conditions}
        \zeta_{B_{N+1};\beta,\rho}^{0}\left(\sigma_{\partial B_{N+1}} = S\right)
        &\geq \zeta_{B_{N+1};\beta,\rho}^{0}\left(\sigma_{\partial B_{N+1}} = S\mid \sigma_{\partial B_{N}} = 0\right) \zeta_{B_{N+1};\beta,\rho}^{0}\left(\sigma_{\partial B_{N}} = 0\right) \\
        &\geq \zeta_{B_{N+1};\beta,\rho}^{0}\left(\sigma_{\partial B_{N+1}} = S\mid \sigma_{\partial B_{N}} = 0\right) \prod_{k=1}^K\zeta_{B_{N+1};\beta,\rho}^{0}\left(\sigma_{x_k} = 0 \right).
    \end{aligned}
    \end{equation}
    A computation shows that
    \begin{equation}
        \zeta_{B_{N+1};\beta,\rho}^{0}\left(\sigma_{\partial B_{N+1}} = S\mid \sigma_{\partial B_{N}} = 0\right)
        \geq e^{-(\rho S^2+c_\rho+2d\beta W(S))\abs{\partial B_{N+1}}}
    \end{equation}
    where $c_\rho := \ln \sum_{k \in \Z} e^{-\rho k^2}$.
    We claim that there exists a constant $c' > 0$ such that for every $x \in \partial B_N$
    \begin{equation}\label{eq:IVGFF_prob_of_zero}
        \zeta_{B_{N+1};\beta,\rho}^{0}\left(\sigma_x = 0\right) \geq e^{-c'}.
    \end{equation}
    Combining \eqref{eq:upper_bound_numerator_MGF_new_boundary_conditions} and \eqref{eq:lower_bound_denominator_MGF_new_boundary_conditions} we obtain, with $c'' = \rho S^2 + c_\rho + 2d\beta W(S) + c'$
    \begin{equation}
        \zeta_{B_N;\beta,\rho}^S\left(e^{2u \abs{A_N}}\right)
        \leq 2e^{c'' \abs{\partial B_{N+1}}+\frac{2u^2}{\rho \abs{B_N}}}.
    \end{equation}
    Therefore, choosing $u=\frac {\rho t}4 \abs{B_N}$ and inserting into \eqref{eq:magnetization_split_markov}, we have that
    \begin{equation}
        \zeta_{B_N;\beta,\rho}^S\left(\abs{A_N} \right)
        \leq t + \frac{8e^{-\frac {\rho t^2}8 \abs{B_N}+ c'' \abs{\partial B_{N+1}}}}{e \rho t \abs{B_N}}.
    \end{equation}
    Thus, choosing $t = N^{-\frac13}$, the right hand side converges to zero as $N \to \infty$. It remains to prove \eqref{eq:IVGFF_prob_of_zero}. 
    For this, first note that the distribution of $\sigma_x$ is symmetric when we have $0$ boundary conditions:
    \begin{equation}
        \zeta_{B_{N+1};\beta,\rho}^{0}\left(\sigma_x = k\right) 
        =\zeta_{B_{N+1};\beta,\rho}^{0}\left(\sigma_x = -k\right).
    \end{equation}
    Moreover, it is log-concave:
    \begin{equation}\label{eq:log-concavity}
        \begin{aligned}
            \zeta_{B_{N+1};\beta,\rho}^{0}\left(\sigma_x = k\right)^2
            &\geq \zeta_{B_{N+1};\beta,\rho}^{0}\left(\sigma_x = k+1\right) \zeta_{B_{N+1};\beta,\rho}^{0}\left(\sigma_x = k-1\right).
        \end{aligned}
    \end{equation}
    To see this, set $\tilde\zeta_{B_{N+1}}^\eta(\sigma) := \zeta^{\eta-1}_{B_{N+1}}(\sigma -1)$ for $\sigma \in \Z^{\Z^d}$. Then, 
    \begin{equation}
        \zeta_{B_{N+1};\beta,\rho}^{0}\left(\sigma_x = k\right)
        =\tilde\zeta_{B_{N+1};\beta,\rho}^{1}\left(\sigma_x = k+1\right)
        =\tilde\zeta_{B_{N+2};\beta,\rho}^{0}\left(\sigma_x = k+1 \mid \sigma_{\partial B_N} = 1\right).
    \end{equation}
    One verifies that 
    \begin{equation}
        \zeta_{B_{N+2};\beta,\rho}^{0}\left(\sigma\right)
        \tilde\zeta_{B_{N+2};\beta,\rho}^{0}\left(\bar\sigma\right)
        \leq 
        \zeta_{B_{N+2};\beta,\rho}^{0}\left(\sigma \wedge \bar\sigma\right) 
        \tilde\zeta_{B_{N+2};\beta,\rho}^{0}\left(\sigma \vee \bar\sigma\right)
    \end{equation}
    by checking this inequality on configurations that differ only at a single vertex (using $\rho \geq 0$), combined with the FKG inequality for $\zeta_{B_{N+2};\beta,\rho}^{0}$ (c.f. \cite[Theorem 2.6]{Grimmett:RandomClusterModel2006}).
    Together with the four-function theorem, this implies
    \begin{equation}
    \begin{aligned}
        &\phntm\zeta_{B_{N+2};\beta,\rho}^{0}\left(\sigma_x = k \mid \sigma_{\partial B_{N+1}}=0\right)
        \tilde\zeta_{B_{N+2};\beta,\rho}^{0}\left(\sigma_x = k+1\mid \sigma_{\partial B_{N+1}}=1\right)\\
        &\geq 
        \zeta_{B_{N+2};\beta,\rho}^{0}\left(\sigma_x = k+1\mid \sigma_{\partial B_{N+1}}=0\right)
        \tilde\zeta_{B_{N+2};\beta,\rho}^{0}\left(\sigma_x = k\mid \sigma_{\partial B_{N+1}}=1\right),
    \end{aligned}
    \end{equation}
    which is equivalent to \eqref{eq:log-concavity}.
    
    Log-concavity and symmetry of $\zeta_{B_{N+1};\beta,\rho}^{0}$ imply that
    \begin{equation}
        \N \ni k \mapsto \zeta_{B_{N+1};\beta,\rho}^{0}\left(\sigma_x = k\right)
    \end{equation}
    is non-increasing in $k$ and maximal at $k=0$. Combining this monotonicity with the exponential Markov inequality and \zcref{thm:exponential_brascamp_lieb}, we have for all $k \in \N$
    \begin{equation}
        2k\zeta_{B_{N+1};\beta,\rho}^{0}\left(\sigma_x = 0\right)
        \geq \sum_{l=0}^{k-1} \zeta_{B_{N+1};\beta,\rho}^{0}\left(\abs{\sigma_x} = l\right)
        = 1 - \zeta_{B_{N+1};\beta,\rho}^{0}\left(\abs{\sigma_x} \geq k\right)
        \geq 1 - 2e^{-\frac{\rho k^2}2}.
    \end{equation}
    Hence, choosing $k$ large enough, we obtain \eqref{eq:IVGFF_prob_of_zero} for a constant $c'(d, \rho)$ independent of $x$ and $N$.
\end{proof}

\begin{remark}
    The condition that $W$ is a mixture of Gaussians could be removed, if one is able to prove \eqref{eq:integer_valued_field_MGF} differently.
\end{remark}

\begin{proof}[Proof of uniqueness in \zcref{thm:main_mixtures_of_gaussians}.]
    We claim there exists $\rho_*(\beta) > 0$ such that for any $\rho \in (-\rho_*, \infty)$ there exists a $R(\beta,\rho) > \max(0, \rho)$ such that for every $x \in \Lambda$
    \begin{equation}\label{eq:domination_IVGFF_magnetization}
        \mu_{\Lambda;\beta,\rho}^+(\sigma_x)
        \leq \zeta_{\Lambda;\beta,R}^S(\sigma_x)
    \end{equation}
    and such that both $\rho_*$ and $R$ depend continuously on $\beta$. Assuming this claim, we bound
    \begin{equation}
        0 \leq \mu_{\beta,\rho}^+(\sigma_0)
        = \mu_{\beta,\rho}^+(A_N)
        \leq \mu_{B_N;\beta,\rho}^+(A_N)
        \leq \zeta_{B_N;\beta,R}^S(A_N)
    \end{equation}
    and by \zcref{thm:decay_magnetization_IVGFF} the right hand side converges to zero as $N \to \infty$. Then \zcref{thm:unique_gibbs_measure_if_magnetization_zero} implies uniqueness of the infinite volume Gibbs measure.
    
    To prove \eqref{eq:domination_IVGFF_magnetization}, we follow the proof of~\cite[Lemma 8.2]{DAlimonte:FreeEnergyAnalyticity2026}. Let $\lambda(\sigma) := \prod_{i \in \Lambda \cup \set{*}} \lambda_i(\sigma_i)$ be the product of measures on $\Z$, and on $\Z^{\Lambda \cup\set{*}}$ define the measure
    \begin{equation}
        \phi_\lambda(\sigma) \propto \expa{-\beta \sum_{\set{i,j} \in \E^b(\Lambda)} W(\sigma_i - \sigma_j)}\lambda(\sigma),
    \end{equation}
    where we understand $\sigma_j=\sigma_*$ for all $j \notin \Lambda$.
    For $\rho \in \R$, $R > 0$, $\Delta \subset \Lambda$ and $i \in \Lambda \cup \set{*}$, we define
    \begin{equation}
        \lambda_i^\Delta(\sigma_i) := \begin{cases}
            \delta_{-S} + \delta_{S} & \text{if } i = *,\\
            \sum_{k=-S}^S e^{-\rho k^2} \delta_{k} & \text{if } i \neq * \text{ and } i \notin \Delta,\\
            \sum_{k \in \Z} e^{-R k^2} \delta_{k} & \text{if } i \neq * \text{ and } i \in \Delta.\\
        \end{cases}
    \end{equation}
    Note that with these definitions
    \begin{equation}
        \mu_{\Lambda;\beta,\rho}^+(\sigma_x)
        = \frac1S\phi_{\lambda^\emptyset}(\sigma_x \sigma_*)
        \quad\text{and}\quad
        \zeta_{\Lambda;\beta,R}^S(\sigma_x)
        = \frac1S\phi_{\lambda^\Lambda}(\sigma_x \sigma_*).
    \end{equation}
    By induction it suffices to show that for $\Delta \subset \Lambda$ and $u \in \Lambda \setminus \Delta$
    \begin{equation}\label{eq:one_step_measure_change}
        \phi_{\lambda^\Delta}(\sigma_x \sigma_*)
        \leq \phi_{\lambda^{\Delta \cup \set{u}}}(\sigma_x \sigma_*).
    \end{equation}
    A computation shows that for $\xi \in \N^{\Lambda \cup \set*}$
    \begin{equation}
        \phi_{\lambda^\Delta}(\sigma_x \sigma_* \mid \abs{\sigma}=\xi)
        = h(\xi)
        =\phi_{\lambda^{\Delta \cup \set{u}}}(\sigma_x \sigma_* \mid \abs{\sigma} = \xi),
    \end{equation}
    because conditional on $\abs{\sigma}$ the two measures are the same. Here,    
    \begin{equation}
        h(\xi) := \frac{S\xi_x \sum_{\kappa \in \set{\pm}^{\Lambda \cup \set{*}}} \kappa_x\kappa_* e^{\beta \sum_{\set{i,j} \in \E^b(\Lambda)} J(\xi_i, \xi_j) \kappa_i\kappa_j}}{\sum_{\kappa \in \set{\pm}^{\Lambda \cup \set{*}}} e^{\beta \sum_{\set{i,j} \in \E^b(\Lambda)} J(\xi_i, \xi_j) \kappa_i\kappa_j}}
    \end{equation}
    is the two point correlation of a ferromagnetic Ising model with coupling constants
    \begin{equation}
        J(\xi_i, \xi_j) := \frac12\left(W(\xi_i+\xi_j) - W(\xi_i - \xi_j)\right).
    \end{equation}
    By convexity of $W$, these coupling constant are increasing in $\xi$. Griffiths second inequality for the Ising model implies that correlation functions of the Ising model are increasing in the coupling constants and thus that $h$ is increasing in $\xi$.
    Hence,
    \begin{equation}
        \phi_{\lambda^\Delta}(\sigma_x \sigma_*)
        = \phi_{\lambda^\Delta}(h(\abs{\sigma}))
        \quad\text{and}\quad
        \phi_{\lambda^{\Delta \cup \set{u}}}(\sigma_x \sigma_*)
        = \phi_{\lambda^{\Delta \cup \set{u}}}(h(\abs{\sigma}))
    \end{equation}
    Assume for now that for all $t \geq 0$
    \begin{equation}\label{eq:stochastic_domination_absolute_value_sigma}
        \phi_{\lambda^\Delta}(\abs{\sigma_u} > t) 
        \leq \phi_{\lambda^{\Delta \cup \set{u}}}(\abs{\sigma_u} > t).
    \end{equation}
    We will prove this below. We claim that this inequality implies
    \begin{equation}
        \phi_{\lambda^\Delta}(f(\abs{\sigma})) 
        \leq \phi_{\lambda^{\Delta \cup \set{u}}}(f(\abs{\sigma}))
    \end{equation}
    for any $f \colon \N^{\Lambda \cup \set{*}} \to \R_{\geq 0}$ which is non-decreasing, and thus in particular \eqref{eq:one_step_measure_change}.
    
    As before,
    \begin{equation}
        \phi_{\lambda^\Delta}(f(\abs{\sigma}) \mid \abs{\sigma_u} = t)
        = g(t)
        =\phi_{\lambda^{\Delta \cup \set{u}}}(f(\abs{\sigma}) \mid \abs{\sigma_u} = t).
    \end{equation}
    It follows from the absolute value lattice FKG condition~\cite[Theorem 6.3]{Lammers:DelocalisationAbsolutevalueFKGSolidonsolid2024} and Holley's inequality~\cite[Theorem 2.1]{Grimmett:RandomClusterModel2006} that $g$ is non-decreasing in $t$. Hence,
    \begin{equation}
        \phi_{\lambda^\Delta}(g(\abs{\sigma_u}))
        = \int_{[0,\infty)} \phi_{\lambda^\Delta}(g(\abs{\sigma_u}) > t) \d t
        \leq \int_{[0,\infty)} \phi_{\lambda^{\Delta \cup \set{u}}}(g(\abs{\sigma_u}) > t) \d t
        =\phi_{\lambda^{\Delta \cup \set{u}}}(g(\abs{\sigma_u})),
    \end{equation}
    where the inequality follows from \eqref{eq:stochastic_domination_absolute_value_sigma}. 
    
    It remains to prove \eqref{eq:stochastic_domination_absolute_value_sigma}. Write
    \begin{equation}
        \phi := \phi_{\lambda^\Delta}, \quad 
        \phi' := \phi_{\lambda^{\Delta \cup \set{u}}}(\cdot \mid \abs{\sigma_u} \leq S), \quad 
        \phi'' := \phi_{\lambda^{\Delta \cup \set{u}}}.
    \end{equation}
    Clearly, the claim is true for $t \geq S$. Thus assume that $t < S$. Then
    \begin{equation}
        \phi''(\abs{\sigma_u} \leq t)
        = \phi'(\abs{\sigma_u} \leq t) \phi''(\abs{\sigma_u} \leq S).
    \end{equation}
    For any $R \geq \rho$ it holds that
    \begin{equation}
        \phi'(\abs{\sigma_u} \leq t) 
        = \frac{\phi\left(\I_{\abs{\sigma_u}\leq t} e^{-(R- \rho)\sigma_u^2}\right)}{\phi\left(e^{-(R - \rho) \sigma_u^2}\right)}
        \leq e^{(R - \rho) S^2}\phi\left(\abs{\sigma_u} \leq t\right).
    \end{equation}
    Moreover, by the absolute value FKG inequality for $\phi''$ we have that
    \begin{equation}
    \begin{aligned}
        \phi''(\abs{\sigma_u} \leq S) 
        &\leq \phi''(\abs{\sigma_u} \leq S \mid \abs{\sigma_i} = 0 \text{ for all } i \notin \set{u, *})\\
        &\leq \max_{\alpha \in \set{0, \dots, 2d}}\frac{\sum_{\sigma_* \in \set{-S, S}}\sum_{\sigma = -S}^S e^{-R \sigma^2 -\beta \alpha W(\sigma)- \beta (2d-\alpha)W(\sigma - \sigma_*)}}{\sum_{\sigma_* \in \set{-S, S}}\sum_{\sigma \in \Z} e^{-R \sigma^2 -\beta \alpha W(\sigma)- \beta (2d - \alpha)W(\sigma - \sigma_*)}} 
        := \delta_\beta(R).
    \end{aligned}
    \end{equation}
    We make the choice
    \begin{equation}\label{eq:def_widehat_m}
        R_\beta := \min\left(\abs{\rho}+1, \rho - S^{-2}{\ln \delta_\beta(\abs{\rho}+1)}\right).
    \end{equation}
    Since $\delta_\beta < 1$ we have that $R_\beta > \rho$. Since $R_\beta$ is continuous in $\rho$ and $R_\beta(0)>0$, it follows that
    \begin{equation}\label{eq:def_rho_*_3D}
        \rho_* (\beta) := \sup\set{\bar\rho > 0 \mid R_\beta(\rho)>0 \text{ for all } \rho \in (-\bar\rho, \infty)} >0.
    \end{equation}
    In particular, $R_\beta>0$ for all $\rho \in (-\rho_*, \infty)$. 
    By the absolute value FKG inequality we have that $\delta_\beta$ is increasing in $R$ and thus 
    \begin{equation}
    \begin{aligned}
        \phi''(\abs{\sigma_u} \leq t)
        &\leq \delta_\beta(R_\beta) e^{(R_\beta - \rho)S^2} \phi(\abs{\sigma_u} \leq t) \\
        &\leq \delta_\beta(\abs{\rho} +1) e^{(R_\beta - \rho)S^2} \phi(\abs{\sigma_u} \leq t) \\
        &\leq \phi(\abs{\sigma_u} \leq t),
    \end{aligned}
    \end{equation}
    where both inequalities follow from the definition of $R_\beta$. This concludes the proof of \eqref{eq:stochastic_domination_absolute_value_sigma}. 
    
    For future reference we establish that $\rho_*$ is lower semi-continuous. Let $\beta_0 > 0$ and $\eps > 0$ and set $\rho_\eps := \rho_*(\beta_0) - \eps$. By continuity of $R_\beta(\rho)$ in $\beta$ and $\rho$, there exists a constant $c_0 > 0$ such that $R_{\beta_0}(\rho) \geq c_0$ for all $\rho \in [-\rho_\eps,0]$, since $[-\rho_\eps,0] \subset (-\rho_*(\beta_0), \infty)$. By continuity and compactness we thus find $\delta > 0$ such that $R_{\beta}(\rho) \geq \frac{c_0}2$ for all $\beta \in (\beta_0 - \delta, \beta_0 +\delta)$ and $\rho \in [-\rho_\eps,0]$. In particular, for all such $\beta$, it holds that $\rho_*(\beta) \geq \rho_\eps$. Since $\eps$ was arbitrary, it thus holds that $\liminf_{\beta \to \beta_0} \rho_*(\beta)\geq \rho_*(\beta_0)$.
\end{proof}

\begin{remark}
    In the special case $W(x)=x^2$ one can use an inequality of van Beijeren--Sylvester \cite[Corollary 2.4]{VanBeijeren:PhaseTransitionsContinuousspin1978} to obtain the domination \eqref{eq:domination_IVGFF_magnetization}, providing an alternative to the method of \cite{DAlimonte:FreeEnergyAnalyticity2026} outlined above. This was also remarked on in \cite[Equation (3.20)]{Bricmont:RandomSurfacesStatistical1986a}.
\end{remark}

\subsection{Exponential decay of correlations}
\label{sec:exp_decay_mixture_of_gaussians}

As in the two dimensional case, we prove exponential decay of correlations by showing that our model is in a certain sense a ``subcritical'' percolation model. To do so, it will be convenient to work in the dilute random cluster representation introduced in \zcref{sec:dilute_random_cluster}. Recall its definition \eqref{eq:def_DRC} and the definition of $\Psi^\xi_{\Lambda;p,\lambda}$ from \zcref{def:Psi_Phi_successive_expectations}.

The following is the key estimate, which is proved by an extension of the OSSS inequality to our setting. Since the proof is mainly a verbatim extension of the arguments from \cite{Duminil-Copin:SharpPhaseTransition2019} and \cite[Section 8]{Gunaratnam:ExistenceTricriticalPoint2024a}, we include most of the details in \zcref{appendix:OSSS}. Our OSSS inequality applies to the following class of measures.

\begin{definition}[Weakly monotonic measure]\label{def:weak_monotonicity}
    Let $(V,E)$ be a finite graph.
    We say that a positive measure $\pi$ on $\set{0, \dots, S}^V \times \set{0,1}^E$ is \emph{weakly monotonic} if the following holds for any $X \subset V$, $u \in V \setminus X$ and any $F \subset E$ such that if $\set{i,j} \in F$ then $i,j \in X$:

    For any choice $\eta_1, \eta_2 \in \set{0, \dots, S}^X$ and $\xi_1, \xi_2 \in \set{0,1}^{F}$ satisfying $(\eta_1, \xi_1) \leq (\eta_2, \xi_2)$, and 
    $\pi(\varphi_X=\eta_1 \text{ and } \omega_F = \xi_1) > 0$ and $\pi(\varphi_X = \eta_2 \text{ and } \omega_F = \xi_2) > 0$,
    the following two conditions hold
    \begin{enumerate}[label=(\roman*)]
        \item For any $y \in V \setminus X$ and $s \in \set{0,\dots,S-1}$:
        \begin{equation}
        \begin{aligned}
            &\phntm\pi(\varphi_y \geq s+1 \mid \varphi_y \geq s, \,  \varphi_X = \eta_1 \text{ and } \omega_F = \xi_1)\\
            &\leq 
            \pi(\varphi_y \geq s+1 \mid \varphi_y \geq s,\, \varphi_X = \eta_2 \text{ and } \omega_F = \xi_2),
        \end{aligned}
        \end{equation}
        \item For any $e \in E \setminus F$:
        \begin{equation}
            \pi(\omega_e = 1 \mid \varphi_X = \eta_1 \text{ and } \omega_F = \xi_1)
            \leq 
            \pi(\omega_e = 1 \mid \varphi_X = \eta_2 \text{ and } \omega_F = \xi_2).
        \end{equation}
    \end{enumerate}
\end{definition}

\begin{lemma}[{\cite[Lemma 8.2]{Gunaratnam:ExistenceTricriticalPoint2024a},~\cite[Lemma 3.2]{Duminil-Copin:SharpPhaseTransition2019}}]
\label{thm:weakly_monotonic_measure_covariance_lower_bound}
    Let $p, \lambda \colon \N \times \N \to [0,\infty)$ be symmetric and assume that $p \in [0,1)$ and $\lambda > 0$. Additionally, assume that $p(0,b)=0$ and $a\mapsto {p(a,b)}$ is non-decreasing for any choice of $b \in \set{0,\dots,S}$ and assume that for any choice of $E \subset \E^b(\Lambda)$ and $\xi \in \set{0,1}^E$ the measure $\Psi_{\Lambda;p,\lambda}^\xi$ satisfies the lattice FKG condition.

    Then, for every $N \in \N$
    \begin{equation}
    \begin{aligned}
        &\phntm\sum_{i \in \Lambda} \Cov_{\Lambda;p,\lambda}\left(\I_{0 \overset{\text{open}}{\longleftrightarrow} \partial B_N}; \phi_i\right) + 
        \sum_{e \in \E^b(\Lambda)} \Cov_{\Lambda;p,\lambda}\left(\I_{0 \overset{\text{open}}{\longleftrightarrow} \partial B_N}; \omega_e \right) \\
        &\geq \frac{N}{4d Q_N} \pi^+_{\Lambda;p,\lambda}\left(0 \overset{\text{open}}{\longleftrightarrow} \partial B_N\right)(1- \pi^+_{\Lambda;p,\lambda}\left(0 \overset{\text{open}}{\longleftrightarrow} \partial B_N)\right),
    \end{aligned}
    \end{equation}
    where $Q_N := \max_{x \in B_N}\sum_{k=0}^{N-1} \pi^+_{\Lambda;p,\lambda}\left(x \overset{\text{open}}{\longleftrightarrow} \partial B_k(x)\right)$ and where $\Cov_{\Lambda;p,\lambda}$ is the covariance with respect to $\pi^+_{\Lambda;p,\lambda}$. As before, the connectivity events refer to connectivity using open edges in $\omega$.
\end{lemma}
\begin{proof}
    We can view the measure $\pi^+_{\Lambda;p,\lambda}$ as a measure on $\set{0,\dots,S}^{\Lambda \cup \set{*}}\times\set{0,1}^{\E^b(\Lambda)}$, where $*$ is a single vertex representing all vertices in $\Lambda^c$. 
    Then the result follows from \zcref{thm:weakly_monotonic_measure_covariance_lower_bound_appendix}, if we show that the measure $\pi^+_{\Lambda;p,\lambda}$ is \emph{weakly monotonic}.

    Indeed, by \zcref{thm:stochastic_ordering_dilute_measures} we have that
    \begin{equation}
    \begin{aligned} 
        &\phntm\pi^+_{\Lambda;p,\lambda}(\varphi_y \geq s+1 \mid \varphi_y \geq s,\, \varphi_X = \eta_1 \text{ and } \omega_E = \xi_1) \\
        &= \Psi^{\xi_1}_{\Lambda;p,\lambda}(\varphi_y \geq s+1 \mid \varphi_y \geq s,\, \varphi_X=\eta_1)\\
        &\leq \Psi^{\xi_2}_{\Lambda;p,\lambda}(\varphi_y \geq s+1 \mid \varphi_y \geq s,\, \varphi_X=\eta_2)\\
        &=\pi^+_{\Lambda;p,\lambda}(\varphi_y \geq s+1 \mid \varphi_y \geq s,\, \varphi_X = \eta_2 \text{ and } \omega_E = \xi_2) .
    \end{aligned}
    \end{equation}
    Since $\phi^{\xi_1}_{\Lambda;p(\varphi)}(\omega_e=1)$ is increasing in $\varphi$, we can apply \zcref{thm:mu_conditioned_as_iterated_measures} to obtain
    \begin{equation}
    \begin{aligned}
        \pi^+_{\Lambda;p,\lambda}(\omega_e = 1 \mid \varphi_X = \eta_1 \text{ and } \omega_E = \xi_1) 
        &= \Psi^{\xi_1}_{\Lambda;p,\lambda}(\phi^{\xi_1}_{\Lambda;p(\varphi)}(\omega_e=1) \mid \varphi_X=\eta_1) \\
        &\leq \Psi^{\xi_2}_{\Lambda;p,\lambda}(\phi^{\xi_1}_{\Lambda;p(\varphi)}(\omega_e=1) \mid \varphi_X=\eta_2).
    \end{aligned}
    \end{equation}
    where we used again \zcref{thm:stochastic_ordering_dilute_measures}.
    Finally, we use that the random cluster measures are stochastically increasing in the boundary configurations, i.e. 
    \begin{equation}
        \phi^{\xi_1}_{\Lambda;p(\varphi)}(\omega_e=1)  \leq \phi^{\xi_2}_{\Lambda;p(\varphi)}(\omega_e=1)
    \end{equation}
    concluding the proof.
\end{proof}

A central assumption of the previous lemma is that $\Psi_{\Lambda;p,\lambda}^\xi$ satisfies the FKG lattice condition. We show that this holds for the class of models we are interested in by checking the conditions of \zcref{thm:vertex_marginal_FKG_general}. The next lemma shows that \eqref{eq:FKG_condition_p_am} is satisfied for the interaction-energies we are considering.

\begin{lemma}\label{thm:super-gaussian_potential_satisfies_p_lambdam_condition}
    Let $\beta > 0$ and $\rho \in \R$. Assume that $W$ is even, convex and super-Gaussian.
    Define 
    \begin{equation}
        p(a,b) := 1-e^{-\beta (W(a+b) - W(a-b))}
        \qquad\text{and}\qquad
        \lambda(a,b) = e^{-\rho(a^2+b^2)-\beta W(a+b)}.
    \end{equation}
    Then \eqref{eq:FKG_condition_p_am} is satisfied for all $a,b \in \N$. The inequality is strict if $W$ is strictly convex.
\end{lemma}
\begin{proof}
    A computation shows that
    \begin{equation}
        \frac{\lambda(a+1,b+1)\lambda(a,b)}{\lambda(a+1,b) \lambda(a,b+1)} 
        = e^{-\beta W^{(2)}(a+b+1)}
    \end{equation}
    and that
    \begin{equation}
        \frac{1-p(a+1,b+1)}{1-p(a+1,b)} \frac{1-p(a,b)}{1-p(a,b+1)} 
        = e^{-\beta W^{(2)}(a+b+1)-\beta W^{(2)}(a-b)}.
    \end{equation}
    Thus we need to show that
    \begin{equation}
        \frac{2-p(a+1,b)}{2-p(a+1,b+1)}\frac{2-p(a,b+1)}{2-p(a,b)} 
        \leq  e^{\beta W^{(2)}(a-b)}.
    \end{equation}
    Write $w_k := e^{-\beta W(k)}$, $u=a-b$, $v=a+b$. After elementary computations, the previous inequality is equivalent to
    \begin{equation}
        w_{u}^2-w_{u+1} w_{u-1} - \left(w_{v+1}^2-w_{v+2} w_{v}\right) + w_u(w_{v+2}+w_v)- w_{v+1}(w_{u+1}+w_{u-1}) \geq 0.
    \end{equation}
    First, note that
    \begin{equation}
        \begin{aligned}
            w_u^2\left(1- \frac{w_{u+1} w_{u-1}}{w_u^2}\right) 
            &= e^{-2\beta W(u)}\left(1-e^{-\beta W^{(2)}(u)}\right)\\
            &\geq e^{-2\beta W(v+1)}\left(1-e^{-\beta W^{(2)}(v+1)}\right)\\
            &=w_{v+1}^2\left(1- \frac{w_{v+2} w_{v}}{w_{v+1}^2}\right),
        \end{aligned}
    \end{equation}
    where we used that $W$ is even and convex (and thus non-decreasing on the non-negative integers) and that the second derivative is non-increasing. Note that if $W^{(2)}>0$, i.e. $W$ is strictly convex, this inequality is strict (because $u<v+1$).
    Moreover,
    \begin{equation}
        w_u(w_{v+2}+w_v)- w_{v+1}(w_{u+1}+w_{u-1}) \geq 0
    \end{equation}
    if the map
    \begin{equation}
        f(k) := e^{\beta W(k)-\beta W(k+1)}+e^{\beta W(k)-\beta W(k-1)}
    \end{equation}
    satisfies $f(k+1) \geq f(k)$ for all $k \in \set{0, \dots, 2S-1}$. This follows since $W$ is even, convex, and super-Gaussian.
\end{proof}

We are now in a position to prove the remaining part of \zcref{thm:main_mixtures_of_gaussians}.

\begin{proof}[Proof of exponential decay of correlations in \zcref{thm:main_mixtures_of_gaussians}.]
    We now consider $\beta > 0$, $\rho \in (-\rho_*, \infty)$ such that $\mathscr G(\beta, \rho) = \set{\mu_{\beta, \rho}}$. We want to show that the correlations of $\mu_{\beta, \rho}$ decay exponentially. Equivalently, by \zcref{thm:two_point_correlation_relation_percolation},
    \begin{equation}\label{eq:3d_exponential_decay_inf_volume}
        \bar\mu_{\beta, \rho}^+ \left(0 \overset{\text{open}}{\longleftrightarrow} \partial B_N\right)
        \leq e^{-c_{\beta, \rho} N}.
    \end{equation}
    For $a,b \in \set{0, \dots, S}$ define
    \begin{equation}
        p_{\beta}(a,b) 
        := 1 - e^{-\beta (W(a+b)-W(a-b))}
        \qquad\text{and}\qquad
        \lambda_{\beta,\rho}(a,b) := e^{-\frac\rho{2d}(a^2+b^2) - \beta W(a+b)}.
    \end{equation}
    We consider the following perturbations of $p_\beta$ and $\lambda_{\beta;\rho}$:
    \begin{equation}
        p^t_\beta(a,b) := \frac{e^{t} p_\beta(a,b)}{1-p_\beta(a,b) + e^{t} p_\beta(a,b)}
    \end{equation}
    and
    \begin{equation}
        \lambda_{\beta,\rho}^t(a,b) := e^{-\frac\rho{2d}(a^2+b^2) + \frac{t}{2d}(a+b) - \beta W(a+b)}.
    \end{equation}
    First, we note that the functions $\lambda_{\beta,\rho}^t$ and $p^t_\beta$ satisfy all conditions of \zcref{thm:weakly_monotonic_measure_covariance_lower_bound} for all $t \geq 0$. 
    Second, we claim that there exists $T > 0$ such that for every $t \in [0,T)$  and that $\Psi_{\Lambda;p_\beta^t,\lambda_{\beta,\rho}^t}^\xi$ satisfies the lattice FKG condition for any choice of $E \subset \E^b(\Lambda)$ finite and $\xi \in \set{0,1}^E$. Indeed, we apply \zcref{thm:vertex_marginal_FKG_general}, so it remains to verify \eqref{eq:FKG_condition_p_am}.
    In the special case $W(x) = |x|$, a direct computation shows that it holds for all $t \in \R$ (for completeness, we give a proof in \zcref{sec:SOS_computation}). On the other hand, if $W$ is strictly convex, even and super-Gaussian, then the condition holds with a strict inequality at $t=0$ by \zcref{thm:super-gaussian_potential_satisfies_p_lambdam_condition}, and therefore by continuity it extends to some interval $[0,T)$.
    For $t \in [0, T)$ define
    \begin{equation}
        \theta_N(t) 
        :={\pi^+_{B_{2N};p_\beta^{t}, \lambda_{\beta,{\rho}}^t}\left({0 \overset{\text{open}}{\longleftrightarrow} \partial B_N}\right)}.
    \end{equation}
    Then, using \zcref{thm:weakly_monotonic_measure_covariance_lower_bound}, 
    \begin{equation}
        \begin{aligned}
        \theta_N'(t) = \sum_{i \in B_{2N}} \Cov_{B_{2N}, p_\beta^{t}, \lambda_{\beta,{\rho}}^t} (\I_{0 \overset{\text{open}}{\longleftrightarrow} \partial B_N}; \varphi_i) + &\sum_{e \in \E^b(B_{2N})} \Cov_{B_{2N}, p_\beta^{t}, \lambda_{\beta,{\rho}}^t} (\I_{0 \overset{\text{open}}{\longleftrightarrow} \partial B_N}; \omega_e)\\
        &\geq \frac{N}{4d Q_N(t)} \theta_N(t)(1-\theta_N(t)).
        \end{aligned}
    \end{equation}
    where 
    \begin{equation}
        Q_N(t):= \max_{x \in B_{N}}\sum_{k=0}^{N-1} {\pi^+_{B_{2N};p_\beta^{t}, \lambda_{\beta,{\rho}}^t}\left({x \overset{\text{open}}{\longleftrightarrow} \partial B_k(x)}\right)}.
    \end{equation}
    We estimate for $x \in B_N$
    \begin{equation}
    \begin{aligned}
        \sum_{k=0}^{N-1} \pi^+_{B_{2N};p_\beta^{t}, \lambda_{\beta,{\rho}}^t}\left({x \overset{\text{open}}{\longleftrightarrow} \partial B_k(x)}\right)
        &\leq 2\sum_{k=0}^{N/2} \pi^+_{B_{2N};p_\beta^{t}, \lambda_{\beta,{\rho}}^t}\left({x \overset{\text{open}}{\longleftrightarrow} \partial B_k(x)}\right)\\
        &\leq 2\sum_{k=0}^{N/2} \pi^+_{B_{2k}(x);p_\beta^{t}, \lambda_{\beta,{\rho}}^t}\left({x \overset{\text{open}}{\longleftrightarrow} \partial B_k(x)}\right)
    \end{aligned}
    \end{equation}
    where we used the domain monotonicity \zcref{thm:pi_domain_monotonicity} of $\pi^+$ in the last line. By translation invariance of the lattice, we thus have
    \begin{equation}
        Q_N(t) \leq 2\sum_{k=0}^{N/2} \pi^+_{B_{2k};p_\beta^{t}, \lambda_{\beta,{\rho}}^t}\left({0 \overset{\text{open}}{\longleftrightarrow} \partial B_k}\right) \leq 2 \sum_{k=0}^{N-1} \theta_k(t).
    \end{equation}
    Hence, for $t \in [0,T)$,
    \begin{equation}
         \theta_N'(t) \geq \frac{c_0 N \theta_N(t)}{\sum_{k=0}^{N-1} \theta_k(t)}
    \end{equation}
    where $c_0 := \frac{1-\theta_1(T)}{8d}$.
    Let $f_N(t) := c_0^{-1}\theta_N(t)$. Then,
    \begin{equation}
        f_N'(t) \geq \frac{N f_N(t)}{\sum_{k=0}^{N-1}f_k(t)} > 0.
    \end{equation}
    We claim that there exists $\tau \in [0, T)$ such that $f_N(t) \to 0$ as $N \to \infty$ for all $t \in [0,\tau]$. Then \zcref{thm:exponential_decay_differential_inequality} implies that $f_N(t) \leq c_0^{-1} e^{-c_t N}$ for all $t \in [0, \tau)$. In particular, 
    by \zcref{thm:RC_as_spin_model},
    \begin{equation}
        \bar\mu_{B_{2N};\beta, \rho}^+ \left(0 \overset{\text{open}}{\longleftrightarrow} \partial B_N\right)
        = \pi^+_{B_{2N};p_\beta, \lambda_{\beta,{\rho}}}\left({0 \overset{\text{open}}{\longleftrightarrow} \partial B_N}\right)
        = c_0 f_N(0) \leq e^{-c N}
    \end{equation}
    which implies \eqref{eq:3d_exponential_decay_inf_volume}.
    We now show that this claim holds. 
    Let 
    \begin{equation}
        \bar\beta(t) := \inf\set{\tilde\beta > \beta \colon p_{\tilde\beta}(a,b) \geq p_\beta^{t}(a,b) \text{ for all } a,b \in \set{0, \dots, S}}.
    \end{equation}
    By the monotonicity of the random cluster measure in $p$, we have
    \begin{equation}
        \begin{aligned}
            \theta_N(t)
            = \Psi^{\emptyset}_{B_{2N};p_\beta^{t}, \lambda_{\beta,\rho}^t}\left(\phi^{\emptyset}_{B_{2N};p_\beta^{t}(\varphi)}\left(0 \overset{\text{open}}{\longleftrightarrow} \partial B_N\right)\right)
            \leq \Psi^{\emptyset}_{B_{2N};p_\beta^{t}, \lambda_{\beta,\rho}^t}\left(\phi^{\emptyset}_{B_{2N};p_{\bar\beta(t)}(\varphi)}\left(0 \overset{\text{open}}{\longleftrightarrow} \partial B_N\right)\right)
        \end{aligned}
    \end{equation}
    Let
    \begin{equation}
        R_t := \rho - t - 2d(\bar\beta(t) - \beta) (W(2S) - W(2S-1)) \leq \rho.
    \end{equation}
    Assume for now that for any  $0 \leq a \leq S-1$ and $0 \leq b \leq S$
    \begin{equation}\label{eq:lambda/lambda_non-increasing}
        \lambda_{\beta,\rho}^t(a+1,b) \lambda_{\bar\beta,R_t}(a,b)
        \leq \lambda_{\beta, \rho}^t(a,b) \lambda_{\bar\beta,R_t}(a+1,b)
    \end{equation}
    and that for any $0 \leq a \leq S-1$ and $1 \leq b \leq S$
    \begin{equation}\label{eq:p_beta_p_beta_m_relation}
        \frac{p_\beta^t(a,b)}{1-p_\beta^t(a,b)}\frac{1-p_\beta^t(a+1,b)}{p_\beta^t(a+1,b)}
        \geq\frac{p_{\bar\beta}(a,b)}{1-p_{\bar\beta}(a,b)}\frac{1-p_{\bar\beta}(a+1,b)}{p_{\bar\beta}(a+1,b)}.
    \end{equation}
    Then, by \zcref{thm:stochastic_ordering_dilute_measures} we have that $\Psi^{\emptyset}_{B_{2N};p_\beta^t, \lambda_{\beta;\rho}^t} \leq \Psi^{\emptyset}_{B_{2N};p_{\bar\beta}, \lambda_{\bar\beta;R_t}}$. Hence,
    \begin{equation}
    \begin{aligned}
        \theta_N(t)
        &\leq \Psi^{\emptyset}_{B_{2N};p_{\bar\beta}, \lambda_{\bar\beta;R_t}}\left(\phi^{\emptyset}_{B_{2N};p_{\bar\beta}(\varphi)}(0 \leftrightarrow \partial B_N)\right)\\
        &= \pi^+_{B_{2N};p_{\bar\beta}, \lambda_{\bar\beta;R_t}}\left(0 \overset{\text{open}}{\longleftrightarrow} \partial B_N\right)
        = \bar\mu_{B_{2N};\bar \beta, R_t}^+\left(0 \overset{\text{open}}{\longleftrightarrow} \partial B_N\right).
    \end{aligned}
    \end{equation}
    As $t \downarrow 0$ we have that $\bar\beta \downarrow \beta$ and $R_t \uparrow \rho$. By assumption, $R_{t=0} + \rho_*(\bar\beta(t=0)) = \rho + \rho_*(\beta) >0$. By lower semi-continuity of $\rho_*$ in $\beta$, we thus find $\tau > 0$ such that $R_{\tau} > -\rho_*(\bar\beta(\tau))$.
    In particular, by uniqueness at all temperatures, $\abs{\mathscr G(\bar\beta(\tau),R_{\tau})} = 1$. By \zcref{thm:equivalence_no_percolation_uniqueness} this implies that $\theta_N(t) \to 0$ as $N \to \infty$ for all $t \in [0, \tau]$ as claimed.

    It remains to check the two conditions \eqref{eq:lambda/lambda_non-increasing} and \eqref{eq:p_beta_p_beta_m_relation}. The first follows by convexity of $W$ and our choice of $R_t$:
    \begin{equation}
        \frac{\lambda_{\beta;\rho}^t(a,b)}{\lambda_{\bar\beta;R_t}(a,b)} \frac{\lambda_{\bar\beta;R_t}(a+1,b)}{\lambda_{\beta;\rho}^t(a+1,b)}
        = e^{-\frac1{2d}(R_t - \rho)(2a + 1) - \frac{t}{2d} - (\bar\beta - \beta) (W(a+b+1) - W(a+b))} \geq 1.
    \end{equation}
    Inserting the definitions of $p_\beta^t$ and $p_{\bar\beta}$, the second condition becomes
    \begin{equation}
        \frac{e^{\beta (W(a+b)-W(a-b))}-1}{e^{\beta (W(a+1+b)-W(a+1-b))}-1}
        \geq 
        \frac{e^{\bar\beta (W(a+b)-W(a-b))}-1}{e^{\bar\beta (W(a+1+b)-W(a+1-b))}-1}. 
    \end{equation}
    Since $W$ is convex we have that $W(a+b)-W(a-b) =: \gamma \leq \bar\gamma := W(a+1+b)-W(a+1-b)$ and one only needs to check that
    \begin{equation}
        x \mapsto \frac{e^{x \gamma}-1}{e^{x \bar\gamma}-1}
    \end{equation}
    is non-increasing, which holds e.g. because the derivative in $x$ is non-positive for $0\leq\gamma \leq \bar\gamma$. 
\end{proof}

\begin{remark}
    We remark on an alternative proof of exponential decay when $W(x)=x^2$ using known results, once uniqueness of the infinite volume limit has been established. Define the critical temperature $\beta_c := \inf_{\beta} \set{\mu_{\beta,\rho}^+(\sigma_0)>0}$. For the case $S=1$, \cite[Theorem 1.5]{Gunaratnam:ExistenceTricriticalPoint2024a} implies exponential decay of correlations of $\mu_{\beta,\rho}^+$ for all $\beta < \beta_c$. This result was extended to $S \in \Z_{>0}$ (and even greater generality) in \cite[Theorem 1.1]{Panagiotis:SubcriticalSharpnessRealvalued2026}. Our uniqueness result (cf.~\zcref{sec:uniqueness_mixture_of_gaussians}) implies that $\beta_c = \infty$. We emphasize that both results cited here require $W(x)=x^2$.
\end{remark}

\section{Examples for non-uniqueness}
\label{sec:non-uniqueness}

In this section we show that there is a class of measures with even interaction energy $W$, for which there is no uniqueness in the infinite volume at low and intermediate temperatures, when the dimension $d$ is sufficiently large. The result does not require that $W$ is convex. We follow the general setup of \cite{Peled:LongrangeOrderDiscrete2020}.

Let $A,B \subset \set{-S, \dots, S}$. We call $(A,B)$ a \emph{pattern} if $W(0)=0$ and $W(a-b)=0$ for all $a \in A$, $b \in B$. Further, we call a pattern $(A,B)$ \emph{dominant}, if it maximizes the weight
\begin{equation}
    w_{(A,B)} :=\left(\sum_{a \in A} e^{-\rho a^2}\right)\left(\sum_{b \in B} e^{-\rho b^2}\right)
\end{equation}
among all patterns. Two patterns $(A,B)$ and $(A',B')$ are called equivalent if there exists a bijection $\phi \colon \set{-S, \dots, S} \to \set{-S, \dots, S}$ such that 
\begin{equation}
    (A',B') = (\phi(A), \phi(B)),
    \quad e^{-\rho \sigma^2}=e^{-\rho \phi(\sigma)^2}
    \quad\text{and}\quad W(\sigma-\bar\sigma) = W(\phi(\sigma)-\phi(\bar\sigma))
\end{equation}
for all $\sigma,\bar\sigma \in \set{-S, \dots, S}$.

\begin{proof}[Proof of \zcref{thm:non-uniqueness}]
    Let $(A,B)$ be a pattern. Under the assumptions of \zcref{thm:non-uniqueness} this means that $\abs{a-b}\leq k$ for all $a \in A$ and $b \in B$.
    We claim that there are exactly two dominant patterns, and they are equivalent. Then \cite[Theorem 1.2]{Peled:LongrangeOrderDiscrete2020} implies the desired result.
    For the rest of this proof, let $(A,B)$ be a dominant pattern and we assume that $\abs{A}\geq \abs{B}$. If the reverse is true, exchange the roles of $A$ and $B$ in the following arguments. Since $(A,B)$ is a pattern, this implies that $\abs{B} \leq k+1$. Our goal is to show that every dominant pattern is of the form $(A,A)$ for some discrete interval $A$ of size $k+1$.
    
    Let $a_{\min}, a_{\max} \in A$ be the smallest and largest element of $A$ respectively, and similarly $b_{\min},b_{\max} \in B$ for the set $B$. We first show that $A$ and $B$ must be discrete intervals, i.e., that $A=\set{a_{\min}, \dots, a_{\max}}$ and similarly for $B$. Indeed, assume that there exists $a_0 \in \set{a_{\min}, \dots, a_{\max}} \setminus A$. Define $t \in (0,1)$ such that $a_0 = t a_{\max} + (1-t) a_{\min}$. For any $b \in B$ it holds
    \begin{equation}
        \abs{b - a_0} 
        = \abs{b- t a_{\max} - (1-t) a_{\min}}
        \leq t \abs{b-a_{\max}} + (1-t) \abs{b-a_{\min}}
        \leq k,
    \end{equation}
    because $(A,B)$ is a pattern. Thus $(A \cup \set{a_0}, B)$ is also a pattern and hence $w_{(A,B)}$ was not maximal. Therefore $A$ must be a discrete interval. A similar argument shows that also $B$ must be a discrete interval.

    Now we show that $B \subset A$. Indeed, if there exists $b_0 \in B \setminus A$, then $\abs{b - b_0} \leq b_{\max} - b_{\min} \leq k$ for any $b \in B$, because $B$ is a discrete interval and $\abs{B} \leq k+1$. In particular, $(A \cup \set{b_0}, B)$ is a pattern and hence $w_{(A,B)}$ was not maximal.

    Assume now that $A \neq B$. Let $A_1 := \set{a_{\min},\dots, b_{\min} - 1}$ and $A_2 := \set{b_{\max} + 1,\dots, a_{\max}}$. At least one of the two sets is non-empty and $A=A_1 \cup A_2 \cup B$.
    Define the set $C \subset A$ as follows:
    \begin{equation}
        C:= \begin{cases}
            A_1 \cup B & \text{if } \sum_{a \in A_1} e^{-\rho a^2} \geq \sum_{a \in A_2} e^{-\rho a^2},\\
            A_2 \cup B & \text{otherwise}. \\
        \end{cases}
    \end{equation}
    Assume that $C=A_1 \cup B$, the other case is analogous. Then $b_{\max}-a_{\min}\leq k$, since $(A,B)$ is a pattern, and therefore $(C,C)$ is also a pattern.
    We will show that $w_{(A,B)} < w_{(C,C)}$, which implies that for any dominant pattern we must have $A=B$. Indeed, since $A_1 \neq \emptyset$,
    \begin{equation}
    \begin{aligned}
        w_{(C,C)} - w_{(A,B)}
        &= \left(\sum_{b \in B} e^{-\rho b^2}\right) \left[\sum_{a \in A_1} e^{-\rho a^2} - \sum_{a \in A_2} e^{-\rho a^2}\right] + \left(\sum_{a \in A_1} e^{-\rho a^2}\right)^2 >0.
    \end{aligned}
    \end{equation}
    We have thus shown that all dominant patterns are of the form $(A,A)$ for a discrete interval $A$. To maximize $w_{(A,A)}$ we must have $\abs{A}=k+1$. Recall that $k$ is odd when $\rho > 0$. Hence, the only two dominant patterns are the equivalent patterns $(A,A)$ and $(-A,-A)$ with
    \begin{equation}
        A = \begin{cases}
            \set{-S, \dots, -S+k} & \text{if } \rho \leq 0,\\
            \set{-m, \dots, m, m+1}  & \text{if } \rho > 0 \text{ and } k=2m+1,
        \end{cases}
    \end{equation}
    because translating the interval closer to the edges of $\set{-S,\dots,S}$ increases the weight $w_{(A,A)}$ of the patterns when $\rho < 0$, and centering it at the origin increases the weight when $\rho > 0$. When $\rho = 0$ we assumed $k=2S-1$ and there are only two intervals of this diameter.
\end{proof}

\section*{Acknowledgments}
We are indebted to Thomas Spencer for asking us about uniqueness and exponential decay at all temperatures for $S=1$ and interaction $W(x)=x^2$ (the spin-1 Ising model), initiating this research. We are grateful to Margherita Disertori and Thomas Spencer for initial discussions of the problem, and to Daniel Hadas for a simulation and discussion of the mean-field case which supported the uniqueness for $W(x) = x^2$ and suggested multiplicity in an intermediate temperature regime for $W(x) = |x|^p$ with $p$ large. We thank Michael Aizenman for useful discussions on the problems of this work and related models. We thank Jonas K\"oppl for an interesting discussion of the ideas of the paper
~\cite{van2011discrete} and its relevance to our work.
We thank Eyal Lubetzky for bringing to our attention the question of entropic repulsion for height functions constrained to lie above a floor (\zcref{cor:delocalization}), and we thank Aernout van Enter for drawing our attention to several relevant references.

The research of MM is supported by the German Research Foundation (Deutsche Forschungsgemeinschaft, DFG) under Germany’s Excellence Strategy-GZ 2047/1 -- 390685813 and by the DFG -- CRC 1060 -- 211504053. MM acknowledges the hospitality of the Institute for Advanced Study (IAS), where this research was initiated. The research of RP is partially supported by the National Science Foundation grant DMS-2451133 and the Brin professorship at the University of Maryland. 

\paragraph{Statement of AI use.} We have consulted ChatGPT (versions from 5.4 to 5.6) in our literature search, including clarification of proof details of existing papers and questions about the necessity of their assumptions and the possibility of using alternative assumptions. In addition, ChatGPT provided counterexamples to the validity of some classical correlation and domination inequalities in regimes of parameters outside of their standard statements and also helped to verify some of our proofs.  We have come up ourselves with the proofs of all of our main results and all of our writeup was done without AI assistance. Moreover, we used Claude Opus 5.5 to proofread \texttt{arXiv~v1} of this manuscript, which discovered several minor errors and typos that were fixed in the present version.

\printbibliography

\appendix

\section{Exponential decay via the OSSS inequality}
\label{appendix:OSSS}

For completeness, we present here a version of the OSSS inequality for weakly monotonic measures in the sense of \zcref{def:weak_monotonicity}. This extends the OSSS inequality for monotonic measures on $\set{0,1}^\Lambda$ which was originally published in~\cite{Duminil-Copin:SharpPhaseTransition2019} and which has been successively extended to broader classes of measures in~\cite{Gunaratnam:ExistenceTricriticalPoint2024a} and~\cite{Panagiotis:SubcriticalSharpnessRealvalued2026}.

Let $G=(V,E)$ be a finite graph and let $\Gamma$ be the set of ordered (finite) sequences of elements in $V \sqcup E$, which satisfy the following:
\begin{enumerate}[label=(\roman*)]
    \item Each element of $V \sqcup E$ appears at most once in a given sequence, and
    \item if $e=\set{i,j} \in E$ appears in a given sequence, then $i$ and $j$ must also appear in the sequence before $e$.
\end{enumerate}

\begin{definition}[Decision tree]
    Let $N = \abs{V \sqcup E}$ and let $\bar\Omega := \set{0, \dots, S}^V \times \set{0,1}^E$.
    Given $\psi \in \bar\Omega$, a \emph{decision tree} is a pair $T(\psi) = \left(z_1, \left(\phi_t\right)_{t = 2, \dots, N}\right)$, where 
    $z_1 \in V$, and
    for each $t$, $\phi_t$ is a function which takes as input a tuple $\left((z_1, \dots, z_{t-1}), (\psi_{z_1}, \dots, \psi_{z_{t-1}})\right)$ and returns an element $z_t \in (V \sqcup E) \setminus \set{z_1, \dots, z_{t-1}}$. Here, $(z_1, \dots, z_{t-1}) \in \Gamma$.

    Given a configuration $\psi \in \bar\Omega$, let $\mathbf{z} = (z_1, \dots, z_N)$ be the $N$-tuple defined recursively by $T(\psi)$. For $f \colon \bar\Omega \to [0,1]$, define
    \begin{equation}
        \tau(\psi) := \tau_{f, T}(\psi) 
        := \min \set{t \geq 1 \mid \forall \psi' \in \bar\Omega \colon ~\psi_{\mathbf{z}_{[t]}} = \psi'_{\mathbf{z}_{[t]}} \Rightarrow f(\psi) = f(\psi')},
    \end{equation}
    where $\mathbf{z}_{[t]} := (z_1, \dots, z_t)$. This is the time it takes the algorithm $T$ to determine the value $f(\psi)$.
\end{definition}

Recall the definition of weakly monotonic measures from \zcref{def:weak_monotonicity}.

\begin{theorem}[OSSS inequality for weakly monotonic measures]
\label{thm:OSSS_weakly_monotonic}
    Let $f \colon \bar\Omega \to [0,1]$ be an increasing function. Then for any weakly monotonic measure $\pi$ on $\bar\Omega$ and any decision tree $T$,
    \begin{equation}\label{eq:OSSS_for_weakly_monotonic}
        \Var_\pi(f) \leq \sum_{z \in V\sqcup E} \delta_z(f, T) \Cov_{\pi}(f, \psi_z),
    \end{equation}
    where $\delta_z(f, T) := \pi(\exists t\leq \tau(\psi) \colon z_t = z)$ is the revealment of $f$ for the decision tree $T$.
\end{theorem}
\begin{proof}
    The case $S=1$ has been proved in~\cite[Theorem 8.1]{Gunaratnam:ExistenceTricriticalPoint2024a}. We explain how to extend it to $S \geq 2$ by reducing it to the case $S=1$. 
    
    We consider a graph $\tilde G = (\tilde V, \tilde E)$ where
    \begin{equation}
        \tilde V := \bigcup_{v \in V} \set{(v,1), \dots, (v,S)}
    \end{equation}
    are $S$ copies of $V$ and
    \begin{equation}
        \tilde E 
        := \set{\set{(u,1), (v,1)} \subset \tilde V \mid \set{u,v} \in E}.
    \end{equation}
    Let $\tilde\psi \in \set{0,1}^{\tilde V \sqcup \tilde E}$. We define a configuration $\psi \in \bar\Omega$ as follows:
    \begin{equation}
        \psi_z := \begin{cases}
            \tilde\psi_{(u,1)} + \dots + \tilde\psi_{(u,S)} & \text{if } z=u \in V,\\
            \tilde\psi_{\set{(u,1), (v,1)}} & \text{if } z=\set{u,v} \in E.
        \end{cases}
    \end{equation}
    Moreover, we say that $\tilde\psi$ is admissible if $\tilde\psi_{(u,1)} \geq \dots \geq \tilde\psi_{(u,S)}$ for all $u \in V$. For an admissible configuration, we think of $\tilde\psi_{(u,k)}$ encoding the information $\psi_u \geq k$.
    Moreover, we define a measure $\tilde\pi$ on $\set{0,1}^{\tilde V \sqcup \tilde E}$ by setting
    \begin{equation}
        \tilde\pi(\tilde \psi) := \begin{cases}
            \pi(\psi) & \text{if $\tilde\psi$ is admissible},\\
            0 & \text{otherwise}.
        \end{cases}
    \end{equation}
    Finally, we modify the decision tree $T$ as follows: If $T$ explores the vertex $u$, then the modified tree $\tilde T$ explores all vertices $(u,1), \dots, (u,S)$ in this order. Similarly, if $T$ explores the edge $\set{u,v}$, $\tilde T$ explores the edge $\set{(u,1), (v,1)}$. 

    Define the set $\tilde U_V \subset \mathcal P(\tilde V \sqcup \tilde E)$ as follows: A set $\set{d_1, \dots, d_k} \subset \tilde V \sqcup \tilde E$ is in $\tilde U_V$, if every element appears at most once, for every edge $e \in \set{d_1, \dots, d_k}$ both endpoints are in $\set{d_1, \dots, d_k}$, for every vertex $(u,l) \in \set{d_1, \dots, d_k}$ with $l \geq 2$ also $(u, l-1) \in \set{d_1, \dots, d_k}$, and there is at most one $u \in V$ such that $(u,1) \in \set{d_1, \dots, d_k}$ but $(u,S) \notin \set{d_1,\dots,d_k}$.
    
    If one changes the definition of weak monotonicity of \cite[Section 8.1]{Gunaratnam:ExistenceTricriticalPoint2024a} to use the set $\tilde U_V$, then the measure $\tilde\pi$ is weakly monotonic (in that sense).
    The proof of~\cite[Theorem 8.1]{Gunaratnam:ExistenceTricriticalPoint2024a} goes through verbatim, which proves that
    \begin{equation}
        \Var_{\tilde \pi}(g) \leq \sum_{z \in \tilde V\sqcup \tilde E} \tilde\delta_z(g, \tilde T) \Cov_{\tilde\pi}(g, \tilde \psi_z),
    \end{equation}
    where $\tilde\delta_z(g, \tilde T) := \tilde\pi(\exists t\leq \tau_{\tilde T}(\tilde\psi) \colon z_t = z)$.

    We claim that applying this OSSS inequality to $g(\tilde\psi) := f(\psi)$ yields the result \eqref{eq:OSSS_for_weakly_monotonic}. To this end, first note that $g$ is increasing, since $f$ is increasing. Moreover, since admissible $\tilde\psi$ are in one-to-one correspondence with configurations of $\bar\Omega$, we have
    \begin{equation}
        \Var_{\tilde \pi}(g) = \Var_{\pi}(f)
        \quad\text{and}\quad 
        \Cov_{\tilde\pi}(g, \tilde \psi_z) = \begin{cases}
            \Cov_{\pi}(f, \I_{\psi_u \geq k}) & \text{if } z= (u,k) \\
            \Cov_{\pi}(f, \psi_{\set{u,v}}) & \text{if } z= \set{(u,1), (v,1)} \\
            0 & \text{else}.
        \end{cases} 
    \end{equation}
    Finally, for $u \in V$ and $k \in \set{1, \dots S}$, 
    \begin{equation}
        \tilde\delta_{(u,k)}(g, \tilde T)
        \leq \delta_u(f, T)
    \end{equation}
    by construction of $\tilde T$. Hence,
    \begin{equation}
    \begin{aligned}
         \Var_{\pi}(f) 
         &\leq \sum_{v \in V} \delta_v(f, T) \sum_{k=1}^S \Cov_{\pi}(f, \I_{\psi_v \geq k}) + \sum_{e \in E} \delta_e(f, T) \Cov_{\pi}(f, \psi_e)\\
         &= \sum_{z \in V\sqcup E} \delta_z(f, T) \Cov_{\pi}(f, \psi_z),
    \end{aligned}
    \end{equation}
    concluding the proof.
\end{proof}

\begin{remark}
    It should be possible to use the same proof of \cite[Theorem 8.1]{Gunaratnam:ExistenceTricriticalPoint2024a} for $S\geq 2$ directly, without the reduction outlined above. In that case one could remove the conditioning on $\set{\phi_y \geq s}$ from the definition of weak monotonicity \zcref{def:weak_monotonicity}. We chose not to do so, as it does not matter in our application.
\end{remark}

Weak monotonicity gives the following useful lower bound on the covariance of a connectivity event with the vertex and edge spin variables. As in the main text, we use the notation $x \overset{\text{open}}{\longleftrightarrow} y$ to denote the existence of an open path in $\omega$, i.e., the edge variables.

\begin{lemma}[{\cite[Lemma 8.2]{Gunaratnam:ExistenceTricriticalPoint2024a},~\cite[Lemma 3.2]{Duminil-Copin:SharpPhaseTransition2019}}]
\label{thm:weakly_monotonic_measure_covariance_lower_bound_appendix}
    Assume that $0 \in V$. Let $\pi$ be a weakly monotonic measure on $\bar\Omega$ such that for every edge $\set{i,j} \in E$ it holds $\pi(\omega_{ij} = 0 \mid \varphi_i\varphi_j = 0) = 1$.
    Then for every $N \in \N$
    \begin{equation}
    \begin{aligned}
        &\sum_{i \in V} \Cov_\pi(\I_{0 \overset{\text{open}}{\longleftrightarrow} \partial B_N}; \varphi_i) + 
        \sum_{e \in E} \Cov_\pi(\I_{0 \overset{\text{open}}{\longleftrightarrow} \partial B_N}; \omega_e) \\
        &\geq \frac{N}{4d Q_N} \pi(0 \overset{\text{open}}{\longleftrightarrow} \partial B_N)(1- \pi(0 \overset{\text{open}}{\longleftrightarrow} \partial B_N)),
    \end{aligned}
    \end{equation}
    where $Q_N := \max_{x \in B_N}\sum_{k=0}^{N-1} \pi(x \overset{\text{open}}{\longleftrightarrow}  \partial B_k(x))$.
\end{lemma}
\begin{proof}
    One uses the same proof as for~\cite[Lemma 8.2]{Gunaratnam:ExistenceTricriticalPoint2024a}, but uses the version of the OSSS inequality from \zcref{thm:OSSS_weakly_monotonic}.
\end{proof}

\section{A computation for the \texorpdfstring{$SOS$-model}{SOS-model}}
\label{sec:SOS_computation}

We give a sketch that \eqref{eq:FKG_condition_p_am} holds for
\begin{equation}
    \lambda(a,b) := e^{-\frac\rho{2d}(a^2+b^2) + \frac{t}{2d}(a+b) - \beta \abs{a+b}}
    \quad\text{and}\quad
    p(a,b) := \frac{e^{t} (1-e^{-\beta \abs{a+b} + \beta \abs{a-b}})}{e^{-\beta \abs{a+b} + \beta \abs{a-b}} + e^{t} (1 - e^{-\beta \abs{a+b} + \beta \abs{a-b}})}
\end{equation}
for all $\rho \in \R$, $\beta>0$, $t \in \R$.
Let $0 \leq a, b \leq S-1$ and set $a'=a+1$, $b' = b+1$. By symmetry we can assume $a \geq b$. Then,
\begin{equation}
        \frac{\lambda(a',b')\lambda(a,b)}{\lambda(a',b) \lambda(a,b')} = 1.
\end{equation}
Further, writing $m(a,b) = \min(a,b)$ and using that $a\geq b$, we obtain
\begin{equation}
    \frac{1-p(a',b')}{1-p(a',b)} \frac{1-p(a,b)}{1-p(a,b')} 
    = \frac{1+e^t(e^{2\beta m(a', b)}-1)}{1+e^t(e^{2\beta m(a, b)}-1)} \frac{1+e^t(e^{2\beta m(a, b')}-1)}{1+e^t(e^{2\beta m(a', b')}-1)}
    = \begin{cases}
        1 & \text{if } a > b,\\
         \frac{1+e^t(e^{2\beta a}-1)}{1+e^t(e^{2\beta a'}-1)} & \text{if } a = b.
    \end{cases}
\end{equation}
Since the right hand side is $\leq 1$, it remains to verify
\begin{equation}
    \frac{1-p(a',b')}{1-p(a',b)} \frac{1-p(a,b)}{1-p(a,b')} \frac{2-p(a',b)}{2-p(a',b')}\frac{2-p(a,b')}{2-p(a,b)} \leq 1.
\end{equation}
It holds that
\begin{equation}
    \frac{2-p(a,b)}{1-p(a,b)} = 2 + e^t(e^{2 \beta m(a,b)}-1).
\end{equation}
Thus
\begin{equation}
\begin{aligned}
    \phntm\frac{1-p(a',b')}{1-p(a',b)} \frac{1-p(a,b)}{1-p(a,b')} \frac{2-p(a',b)}{2-p(a',b')}\frac{2-p(a,b')}{2-p(a,b)}
    &= \frac{2+e^t(e^{2\beta m(a', b)}-1)}{2+e^t(e^{2\beta m(a, b)}-1)} \frac{2+e^t(e^{2\beta m(a, b')}-1)}{2+e^t(e^{2\beta m(a', b')}-1)}\\
    &=\begin{cases}
        1 & \text{if } a > b,\\
         \frac{2+e^t(e^{2\beta a}-1)}{2+e^t(e^{2\beta a'}-1)} & \text{if } a = b.
    \end{cases}
\end{aligned}
\end{equation}
As before, this quantity is $\leq 1$, concluding the proof.\qed

\end{document}